\documentclass[11pt]{amsart}

\usepackage{fullpage}

\usepackage{ graphicx, amsmath, amssymb}

\usepackage{subcaption}

\usepackage{hyperref}
\usepackage[nameinlink]{cleveref}	

\usepackage{siunitx}

\usepackage{placeins} 

\usepackage{csquotes}

\usepackage{wrapfig}

\usepackage{amsthm}
\newtheorem{lem}{Lemma}
\newtheorem{thm}[lem]{Theorem}

\numberwithin{equation}{section}

\usepackage[
	backend=biber,
]{biblatex}
\newcommand{\abs}[1]{\left | #1 \right |}
\renewcommand{\d}{\textnormal{d}}

\newcommand{\renop} { \mathcal {R}  }
\newcommand{\hatrenop} { \hat{\mathcal {R}}  }

\newcommand{\hatrenopk}[1] { \hatrenop \left[ #1 \right]  }

\title{Dyson-Schwinger Equations in Minimal Subtraction}
\author{Paul-Hermann Balduf}

\thanks{The author thanks Dirk Kreimer, David Broadhurst and Gerald Dunne for several helpful discussions and for comments and suggestions on the draft.}

\begin{document}

\begin{abstract}
	We compare the solutions of one-scale Dyson-Schwinger equations in the Minimal Subtraction (MS) scheme to the solutions in kinematic (MOM) renormalization schemes. We establish that the MS-solution can be interpreted as a MOM-solution, but with a shifted renormalization point, where the shift itself is a function of the coupling. We derive  relations between this shift and various renormalization group functions and counter terms in perturbation theory. 
	
	As concrete examples, we examine three different one-scale Dyson-Schwinger equations, one based on the 1-loop multiedge graph in D=4 dimensions, one for D=6 dimensions and one mathematical toy model. For each of the integral kernels, we examine both the linear and nine different non-linear Dyson-Schwinger equations. For the linear cases, we empirically find exact functional forms of the shift between MOM and MS renormalization points.  For the non-linear DSEs, the results for the shift suggest a factorially divergent power series. We determine the leading asymptotic growth parameters and find them in agreement with the ones of the anomalous dimension. Finally, we present a tentative  exact solution to one of the non-linear DSEs of the toy model.
\end{abstract}
	
\maketitle

\section{Introduction}

\subsection{Motivation}

So far, the systematic Hopf-algebraic treatment \cite{kreimer_etude_2006,kreimer_etude_2008,yeats_growth_2008,bellon_approximate_2010,bellon_renormalization_2008,bellon_efficient_2010,klaczynski_resurgent_2016} of Dyson-Schwinger equations (DSEs) \cite{dyson_matrix_1949,schwinger_green_1951}   has relied on kinematic renormalization schemes, called MOM hereafter. The solution of a DSE is the renormalized Green function $G (p^2)$. MOM schemes assign a value to the renormalized Green function $G (p^2)$ at one particular momentum $p^2=\mu^2$, and thereby have a transparent interpretation as boundary condition for the DSE.  For a single (non-coupled) DSE, the only remaining unknown object is the anomalous dimension $\gamma(\alpha)$ as a function of the renormalized coupling $\alpha$. If MOM-conditions are used systematically, explicit regularization of divergent integrals is not necessary \cite{kreimer_etude_2006,kreimer_etude_2008} and one only deals with finite quantities at all stages. In fact, $\gamma(\alpha)$ can be computed from a differential equation without any divergent integral. The earlier works \cite{delbourgo_dimensional_1997,delbourgo_dimensional_1996} do regulate the integrals explicitly using dimensional regularization \cite{thooft_regularization_1972,bollini_dimensional_1972}, but  they   use  MOM renormalization conditions nonetheless.

On the other hand, many perturbative computations in quantum field theory use the combination of dimensional regularization and Minimal Subtraction (MS) renormalization conditions. In this scheme, all singular terms in the regulator $\epsilon$ are subtracted, but the resulting amplitude does not respect any particular kinematic boundary condition. The MS-solution of a DSE can only be found by explicitly regulating and renormalizing the divergent integrals at each loop order.

At the same time, different renormalization schemes should not lead to different physical outcomes, therefore, there ought to be \emph{some} particular momentum $\hat \mu$ such that Minimal Subtraction agrees with  kinematic renormalization using that very reference momentum $\hat \mu$. Indeed, the existence of this correspondence is proven on an abstract algebraic level \cite{panzer_renormalization_2015,kreimer_renormalization_2013}. Finding the explicit relationship is of practical relevance since the perturbative solution is typically only known to a few orders. In that case, the truncated MS-renormalized results are truly different from the MOM ones and a priori ony valid around their respective renormalization points \cite{celmaster_renormalizationprescription_1979}. Knowing the physical value of the MS-renormalization point $\hat \mu$ is then crucial for the validity range of the truncated Green function. Furthermore, renormalization group functions differ between both schemes. In MS, but presumably not in MOM, the beta function is expected to be dominated by subdivergence-free diagrams \cite{mckane_perturbation_2019}. A more concrete understanding of the relation between MOM and MS can help to translate such conjectures to the other scheme and subsequently attack them with a different set of tools.

\subsection{Content}
The present paper begins in \cref{theory} with a pedagogical discussion of $Z$-factors and renormalization group equations in both MOM and MS and their various relations and identities. In \cref{MS_as_shifted}, we establish that, in the physical limit $\epsilon \rightarrow 0$, the Green function in MS can be interpreted as a MOM Green function with shifted renormalization point. We derive several ways to compute this shift.

In the remainder, we analyse propagator-type Dyson-Schwinger equations based on three different primitive integral kernels. The first one is the 4-dimensional bubble (=1-loop-multiedge) graph appearing e.g. in the fermion propagator in Yukawa theory. The second one is the same graph in 6 dimensions, contributing to the $\phi^3$-propagator, and the third is a mathematical toy model which has occasionally  been used to study renormalization. In all three cases, we consider the recursive insertion of the Green function into only one place in the kernel graph.  

First, we insert the Green function itself, which amounts to a linear DSE known as the rainbow approximation. In \cref{D4}, the algorithm for linear DSEs is developed for the case of the 4-dimensional bubble integral. Subsequently, this is applied to the $D=6$ case in \cref{D6} and to the toy model in \cref{toy}. As a counter example, we demonstrate in \cref{chains} that the chain-approximation, not arising from a DSE, does not allow for a shift between MS and MOM.

Secondly, we insert the Green function raised to a power $\in \left \lbrace -4, \ldots , +6 \right \rbrace $, producing non-linear DSEs. In \cref{nonlinear_4d}, the algorithm for non-linear DSEs is explained and used for the 4-dimensional model.  The last \cref{nonlinear_6d,nonlinear_toy} contain the results for the non-linear DSEs for the two remaining models, $D=6$ and the toy model. \Cref{conclusion} is a summary of the results.

\section{Theoretical Background}\label{theory}

\subsection{Unrenormalized Dyson-Schwinger equation}
We consider DSEs of the form
\begin{align}\label{dse_general_un}
	G_0(\alpha_0, x) &= 1 + \alpha_0 \int \d y\;  K(x,y) Q(G_0(\alpha_0, y)) G_0(\alpha_0, y) 
\end{align}
where $G_0(\alpha_0,x)$ is the unrenormalized Green function and $K(x,y)$ is an integral kernel, determined by a single primitive Feynman diagram. Restricting the Green function to depend on only one single external scale $x$, means that $G_0(\alpha_0,x)$ represents a 1PI 2-point-function and $K(x,y)$ stems from a propagator-type diagram. For a more general Green function, one needs to include also scale-less \enquote{angle} variables \cite{brown_angles_2011}.   The  variable $x:= p^2 / \mu^2$ is the external momentum scaled to some fixed reference momentum $\mu$.  In the literature, propagator-type DSEs are often written with a minus sign, $G = 1-\alpha \int \ldots$. We will use a plus for all DSEs considered in this work.

In the cases we consider, the invariant charge is a monomial of the Green function,
\begin{align}\label{Qs}
	Q(G_0(\alpha_0, x)) = \big(G_0(\alpha_0,x)\big)^{s}.
\end{align}

\subsection{Renormalization}

Carrying out renormalization on an integral- (rather than integrand-)level requires us to regulate the divergent integrals. 
In dimensional regularization \cite{thooft_regularization_1972,bollini_dimensional_1972}, the space-time dimension is changed by a non-integer shift $\epsilon$, we choose $D=4-2\epsilon$ or $D=6-2\epsilon$. The unrenormalized amplitude of a finite graph is then a Laurent series in the regularization parameter $\epsilon$.

We will only consider DSEs of multiplicatively renormalizable theories. This means that it is possible to eliminate all divergences at order $m$ by two counter terms $Z_G^{(m)}$ and $Z_\alpha^{(m)}$. They act as a rescaling of the Green function and the coupling constant according to
\begin{align}\label{renormalization}
	\alpha_0 =  Z_\alpha (\alpha,\epsilon) \mu^{2\epsilon}\cdot \alpha, \qquad 	G(\alpha,\epsilon,x) &= Z_G (\alpha,\epsilon) \cdot G_0 \left( Z_\alpha(\alpha,\epsilon )  \alpha ,\epsilon, x \right) .
\end{align}
These $Z$-factors are themselves functions of the renormalized coupling $\alpha$ and of the regularization parameter $\epsilon$. 

Renormalization of the coupling constant is only necessary if the theory has a non-vanishing beta function or, equivalently, a running coupling. See \cite{kruger_log_2020} for a recent account. Eventually, $Z_\alpha$ is given by the renormalization of the invariant charge $Q$. In the present case where there is only one Green function $G(\alpha,x)$ which needs renormalization, the invariant charge   \cref{Qs} leads to the identity
\begin{align}\label{ZgZs}
	Z_\alpha(\alpha,\epsilon) &= (Z_G(\alpha,\epsilon))^{s}.
\end{align}

The derivative of the renormalized coupling constant with respect to  the renormalization point, at fixed unrenormalized coupling $\alpha_0$,  is the beta function of the theory,
\begin{align}\label{beta_def}
	\beta\left( \alpha ,\epsilon \right) &:= \mu^2 \frac{\partial}{\partial \mu^2} \ln \alpha(\alpha_0,\mu) +\epsilon=\frac{-\epsilon}{\alpha  \frac{\partial}{\partial \alpha} \ln \left( \alpha  \cdot Z_\alpha\left( \alpha ,\epsilon \right)   \right)} +\epsilon
\end{align}
The anomalous dimension, on the other hand, is defined as the derivative of the renormalized Green function at fixed $\alpha_0$, using \cref{ZgZs} one finds
\begin{align}\label{gamma_def}
	\gamma\left( \alpha ,\epsilon \right) &:= \mu^2 \frac{\partial}{\partial \mu^2}  \ln G \left( \alpha,\epsilon \right) =\left( \beta  ( \alpha ,\epsilon)-\epsilon \right) \alpha  \frac{\partial}{\partial \alpha } \ln Z_G \left( \alpha ,\epsilon \right).
\end{align}
Note that sometimes, $\gamma$ is defined as the derivative of the inverse (i.e. connected, not 1PI) 2-point Green function, or as the derivative with respect to $\mu$. These definitions are equivalent up to overall signs and factors.

Conversely, the renormalization group functions $\beta(\alpha,\epsilon)$ and $\gamma(\alpha,\epsilon)$ uniquely determine the counter terms via 
\begin{align}\label{Z_beta_gamma}
	 Z_\alpha \left( \alpha ,\epsilon \right) &= \exp \left(-\int \limits_0^{\alpha } \frac{\d u}{u} \frac{\beta\left( u,\epsilon \right) }{\epsilon -\beta\left(u,\epsilon \right)  }\right),\qquad
	 Z_G \left( \alpha ,\epsilon  \right) = \exp \left( -\int \limits_0^{\alpha } \frac{\d u}{u} \frac{\gamma \left(u,\epsilon \right)  }{\epsilon -\beta\left(u,\epsilon \right)  } \right).
\end{align}
Such relations have long been known in the traditional formulation (\enquote{Gross-'t Hooft relations}) \cite{thooft_dimensional_1973,collins_new_1974,gross_applications_1981} as well as in the Hopf-algebraic formulation ( \enquote{scattering type formula} )  of quantum field theory  \cite{connes_renormalization_2001},\cite[Sec. 7]{connes_noncommutative_2007}.

It follows from the above definitions that the renormalized Green function $G$ fulfils -- even for $\epsilon \neq 0$ and not only in MOM renormalization -- the Callan-Symanzik equation  \cite{callan_broken_1970,symanzik_small_1970}
 \begin{align}\label{cse}
 	\left( \gamma ( \alpha ,\epsilon)  +  \left(\beta(\alpha ,\epsilon) -\epsilon\right)   \alpha  \partial_\alpha  \right) G  \left( \alpha ,\epsilon,x \right) &= x\partial_x G   ( \alpha ,\epsilon,x  ).
 \end{align}

The renormalization group functions $\beta,\gamma$ as well as the renormalized Green function are generally non-trivial functions of the regularisation parameter $\epsilon$. Assuming that the counter terms \cref{renormalization} are chosen properly, their limit $\epsilon \rightarrow 0$ exists. 
\begin{align}\label{beta_limit}
	\beta(\alpha) &:= \lim_{\epsilon\rightarrow 0} \beta(\alpha,\epsilon), \qquad \gamma(\alpha) := \lim_{\epsilon \rightarrow 0}\gamma(\alpha,\epsilon), \qquad G(\alpha,x) := \lim_{\epsilon \rightarrow 0} G(\alpha,\epsilon,x).
\end{align}
This limit is usually implied when talking about the renormalized quantities.
If the invariant charge has the form \cref{Qs} and consequently \cref{ZgZs} holds, then by \cref{Z_beta_gamma}  
\begin{align}\label{beta_gamma}
	\beta(\alpha,\epsilon)  = s\cdot \gamma(\alpha,\epsilon)\qquad \Rightarrow \qquad \beta(\alpha) &= s \cdot \gamma(\alpha).
\end{align}

To directly compute the renormalization constants and renormalized solution from a DSE, one inserts \cref{dse_general_un} into \cref{renormalization} and uses \cref{Qs,ZgZs} to obtain
\begin{align}\label{dse_general}
	G (\alpha , x) &= Z_G  \left( 1+ Z_\alpha \alpha  \int \d y \; K(x,y) \left( G_0(Z_\alpha \alpha , y) \right)^{s+1}   \right) \nonumber \\
	&= Z_G + \alpha  \int \d y \; K(x,y)  Q\left( G (\alpha ,y) \right)  G(\alpha,y) .
\end{align}
This equation can be solved iteratively by inserting the solution of order $(m-1)$, $G^{(m-1)} (\alpha ,y)$, into the right hand side to obtain the order-$m$-solution $G^{(m)} (\alpha ,x)$. We assume that no IR-divergences appear. The integrand is finite because it is a power of a renormalized Green function, therefore the integral is only superficially divergent.  The so-obtained divergence is of order $\alpha^m$ and can be absorbed by a suitable summand in $Z_G^{(m)}$, producing a finite $G^{(m)} (\alpha ,x)$.  

\subsection{MOM scheme}
The counter terms introduced in \cref{renormalization} are not unique. To fix them, one needs a renormalization condition. In the MOM scheme, this is done by fixing one particular momentum $\delta \cdot \mu^2$, where $\delta \in \mathbb R$ and $\mu$ is an arbitrary but fixed reference momentum. One then demands the renormalized Green function to take the value unity at that momentum. Introducing $x:=p^2/(\delta \mu^2)$, the MOM renormalization condition is
\begin{align}\label{MOM_def}
	G(\alpha,x=1) &=1 \qquad \text{(MOM scheme)}.
\end{align}
This is achieved order by order in \cref{dse_general} if one includes not only the pole term (in $\epsilon$), but  all finite parts of the amplitude  into the counterterm $Z^{(m)} = Z^{(m-1)}-\renop [\ldots]$. The MOM-scheme operator $\renop$ projects the integral to a fixed scale $x=1$. 

 Using \cref{beta_gamma} and the limit \cref{beta_limit}, the Callan-Symanzik equation \cref{cse} becomes
\begin{align}\label{rge}
	\gamma (\alpha)\left( 1+s\alpha \frac{\partial}{\partial \alpha } \right)  G (\alpha ,x)&=x\frac{\partial}{\partial x} G (\alpha ,x) .
\end{align}
At the renormalization point $x=1$, \cref{rge,MOM_def} lead to
\begin{align}\label{gamma_derivative}
	\gamma(\alpha) &= x \partial_x G(\alpha,x) \big|_{x=1}.
\end{align}

\subsection{Expansion in logarithms}\label{sec_logexpand}

The renormalized solution of a 1-scale DSE in MOM-renormalization with renormalization point $x=1$ can be expanded in logarithms according to
\begin{align}\label{Gren_gammak}
	G(\alpha,x) &= 1 + \sum_{k=1}^\infty \gamma_k(\alpha) \left(\ln x\right)^k.
\end{align}
From \cref{gamma_derivative} we identify the anomalous dimension $\gamma_1(\alpha) = \gamma(\alpha)$. The functions $\gamma_{k>1}(\alpha)$ in \cref{Gren_gammak}, with invariant charge \cref{Qs}, can be computed from \cref{rge} \cite{kreimer_etude_2006}:
\begin{align}\label{rge_MOM}
	\gamma(\alpha) \left(  1+s \alpha \partial_\alpha \right) \gamma_{k-1}(\alpha)&=	k \gamma_k(\alpha).
\end{align}
In perturbation theory, all involved functions will be formal power series in $\alpha$. In fact,
\begin{align}\label{gammak_orders}
	\gamma(\alpha) &\in \mathcal{O}\left( \alpha \right) , \qquad \gamma_k(\alpha) \in \mathcal{O}\left( \alpha^k \right) .
\end{align}

For a linear DSE, the exponent in the invariant charge \cref{Qs} is $s=0$,  and consequently one obtains $\gamma_k(\alpha) = \frac{1}{k!} \gamma^k(\alpha)$.  This corresponds to a scaling solution  of \cref{rge},
\begin{align}\label{scaling_solution}
	G(\alpha,x) &=  x^{\gamma(\alpha)}.
\end{align}

The striking advantage of using MOM renormalization conditions for a Dyson-Schwinger equation is that the anomalous dimension $\gamma(\alpha)$ is the only truly unknown function, and it itself can be computed from a non-linear ODE. This ODE is constructed by inserting the renormalization-group differential  operator \cref{rge_MOM} into the Mellin transform of the primitive kernel\cite{kreimer_etude_2006,yeats_growth_2008}: 
\begin{align}\label{mellin_gamma1}
	\frac{1}{-u \cdot M(u)} \Big|_{u \rightarrow -\gamma ( 1+s \alpha \partial_\alpha)} \gamma(\alpha) &= \alpha.
\end{align}
See \cref{mellin} for the Mellin transforms of the kernels used in this paper. A linear DSE with $s=0$ reduces to the algebraic equation $M\left( -\gamma(\alpha) \right)  = \alpha^{-1}$ \cite{kreimer_etude_2008}.

There is a second expansion of the renormalized Green function $G(\alpha,x)$ in terms of logarithms, the leading-log expansion. It is a reordering of \cref{Gren_gammak} in powers of $(\alpha \ln x)$, 
\begin{align}\label{leadinglog_general}
	G(\alpha,x) &= 1 + \sum_{k=1}^\infty H_k \left( \alpha \ln x \right) \alpha^k.
\end{align}
The function $H_1(z)$ is the leading-log contribution to the Green function, and $H_k(z)$ represents the next-to$^k$ leading log part. These expansions have been studied recently \cite{delage_leading_2016,kruger_filtrations_2015,courtiel_nexttok_2020,kruger_log_2020}, one of the results being that for a  DSE \cref{dse_general} with invariant charge \cref{Qs} where $s\neq 0$ one has \cite{kruger_log_2020}
\begin{align}\label{leadinglog_first}
	H_1 (z) &= \left( 1+s c_1 z \right) ^{-\frac{1}{s}}, \qquad H_2(z) = \frac{\left( 1+s c_1 z \right) ^{-\frac{1}{s}-1}}{-s c_1}c_2  \ln (1+s c_1 z) \\
	H_3(z) &=\frac{\left( 1+s c_1 z \right) ^{-\frac 1 s-2}}{  s^2 c_1^2} \left( s^2 c_1 z  (  c_2^2-c_1 c_3)  - s c_2^2 \ln (1+ s c_1 z) + \frac  { c_2^2}2(1 + s ) \ln^2  (1 + s c_1z)  \right). \nonumber 
\end{align}
Here, $c_j=-[\alpha^j]\gamma(\alpha)$ is the $j^\text{th}$ coefficient of the anomalous dimension. These general results will subsequently be used to cross-check our calculation.

\subsection{MS scheme}

In the Minimal Subtraction scheme, the counter term $\hat Z^{(m)}$ is chosen to contain only the pole terms in $\epsilon$, extracted by the operator $\hat \renop$, of the integral at order $m$:
\begin{align}\label{deltaZ_MS}
	 \hat Z^{(m )}(\epsilon) &= \hat Z^{(m-1)}(\epsilon) - \hatrenopk{\alpha \int \d y \; K(x,y) Q\left(G ^{(m-1)}\right) G ^{(m-1)}  }\qquad \text{(MS scheme)}.
\end{align}

In MS, the anomalous dimension $\hat \gamma(\alpha)$ and the beta function $\hat \beta(\alpha)$ are again defined as derivatives of the $\hat Z$-factors, \cref{beta_def,gamma_def}. This implies, using \cref{Z_beta_gamma},  that they do not depend on $\epsilon$ at all: $\hat \beta(\alpha,\epsilon) = \hat \beta(\alpha)$ and $\hat \gamma(\alpha,\epsilon)=\hat \gamma(\alpha)$. Further, the identity \cref{beta_gamma}, $\hat \beta(\alpha) = s\hat \gamma(\alpha)$, still holds if $\hat Q=\hat G^s$. The so-defined renormalization group functions fulfil once more the Callan-Symanzik equation \ref{cse} and its limit for $\epsilon\rightarrow 0$,  \cref{rge}. Note that for any given Feynman graph, the highest order pole in MS coincides with the one in MOM \cite[Sec. 4]{kreimer_chen_1999}, which is also clear from a an induction over the coradical degree of the graph.

The fact that the counter term in MS consists only of pole terms implies that the residues are closely related to the functions $\gamma_j(\alpha)$ in the log expansion \cref{Gren_gammak}. They satisfy a recursion very similar to \cref{rge_MOM}, namely \cite{connes_renormalization_2001,kreimer_renormalization_2013},\cite[Sec. 7]{connes_noncommutative_2007}
\begin{align}\label{MS_Z_poles}
	\hat Z_\alpha (\alpha,\epsilon) &=: 1+\sum_{j=1}^\infty \frac{1}{\epsilon^j} \hat Z_j(\alpha) \nonumber \\
	s \bar \gamma(\alpha) = \bar \beta(\alpha) &= \alpha \partial_\alpha Z_1(\alpha), \qquad \alpha \partial_\alpha Z_j(\alpha) = \beta \partial_\alpha \left( \alpha Z_{j-1}(\alpha) \right) , \quad j>1.
\end{align}

The MS-bar renormalization scheme is a variant of MS where those finite terms which arise from a series expansion of $\frac{e^{\gamma_E}}{4\pi}$ are also subtracted. All quantities computed in Minimal Subtraction are denoted with hat, like $\hat G$.  We denote MS-bar quantities with a bar, like $\bar G$. The undecorated quantities, like $G$, are in the MOM-scheme.

The MS-scheme involves a scale $\mu$ in the definition of $x=p^2/\mu^2$, but no explicit condition of the Green function is imposed. Intuitively, the MS-renormalized Green function $\hat G (\alpha,x)$ will be unity at some other scale $p^2=\hat \mu^2 = \delta(\alpha) \cdot \mu^2$, where the factor $\delta(\alpha)$ is itself a function of $\alpha$. One can view $\delta(\alpha)$ as the renormalization point to be chosen in MOM in order to reproduce the MS Green function. More mathematically, it was shown in \cite{kreimer_renormalization_2013,panzer_renormalization_2015} that in the Hopf algebra formulation of perturbative quantum field theory, MS and MOM are equivalent up to a scaling of the renormalization point. The objective of the present work is to explicitly find this scaling $\delta(\alpha)$ for various Dyson-Schwinger equations.

\section{MS as a shifted MOM scheme}\label{MS_as_shifted}

\subsection{Shifted MOM renormalization point}\label{shifted}

The formulas of \cref{sec_logexpand} are valid for MOM renormalization at $x=1$. Now choose a kinematic renormalization point $ \delta^{-1} \neq 1$. This is equivalent to choosing a reference momentum $\mu' = \sqrt{\delta} \cdot \mu$ instead of $\mu$ and using a new variable $x':= p^2/{\mu'}^2 = \delta^{-1}\cdot x$ such that $x=1$ equals $x'=\delta^{-1}$. This setup is shown in \cref{pic_def_gamma}. The Green function $G'(x')$ is defined by 
\begin{align}\label{G_Gs_def}
	G(\alpha,x) = G(\alpha,\delta \cdot x') =: G'(\alpha,x') = G'(\alpha, \delta^{-1}\cdot x).
\end{align}
It has a log-expansion (in the original variable $x$, not $x'$) similar to \cref{Gren_gammak},
\begin{align}\label{Gren_gammaks}
	G' (x) &= \gamma'_0(\alpha) + \sum_{k=1}^\infty \gamma'_k(\alpha) (\ln x)^k, \qquad 	\gamma'_k = \sum_{j=k}^\infty \binom{j}{k}\gamma_j (\ln\delta)^{j-k}.
\end{align}
The first two of the new coefficients are, explicitly,
\begin{align}\label{gamma1s}
	\gamma'_0(\alpha) =G'(\alpha,1)=  G(\alpha,\delta^{-1}) , \qquad \text{and} \qquad  \gamma'_1(\alpha) = x  \partial_{x } G'(\alpha,x )\big|_{x =1} = x \partial_x G(\alpha,x)\big|_{x=\delta^{-1}}. 
\end{align}
Assume that $\ln(\delta)$ is a power series in $\alpha$ without pole terms. The shifted functions $\gamma'_k(\alpha)$ therefore start with the same coefficients as $\gamma_k(\alpha)$, using \cref{gammak_orders} we have
\begin{align}\label{gammaks_orders}
	\gamma'_k(\alpha) \in \mathcal O \left( \alpha^k \right) , \qquad \gamma'_k(\alpha) = \gamma_k(\alpha) + \mathcal O \left( \alpha^{k+1} \right) .
\end{align}
This means that the leading log function $H_1(z)$ \cref{leadinglog_first} coincides.

\begin{figure}[htbp]
	\includegraphics[width=.7\linewidth]{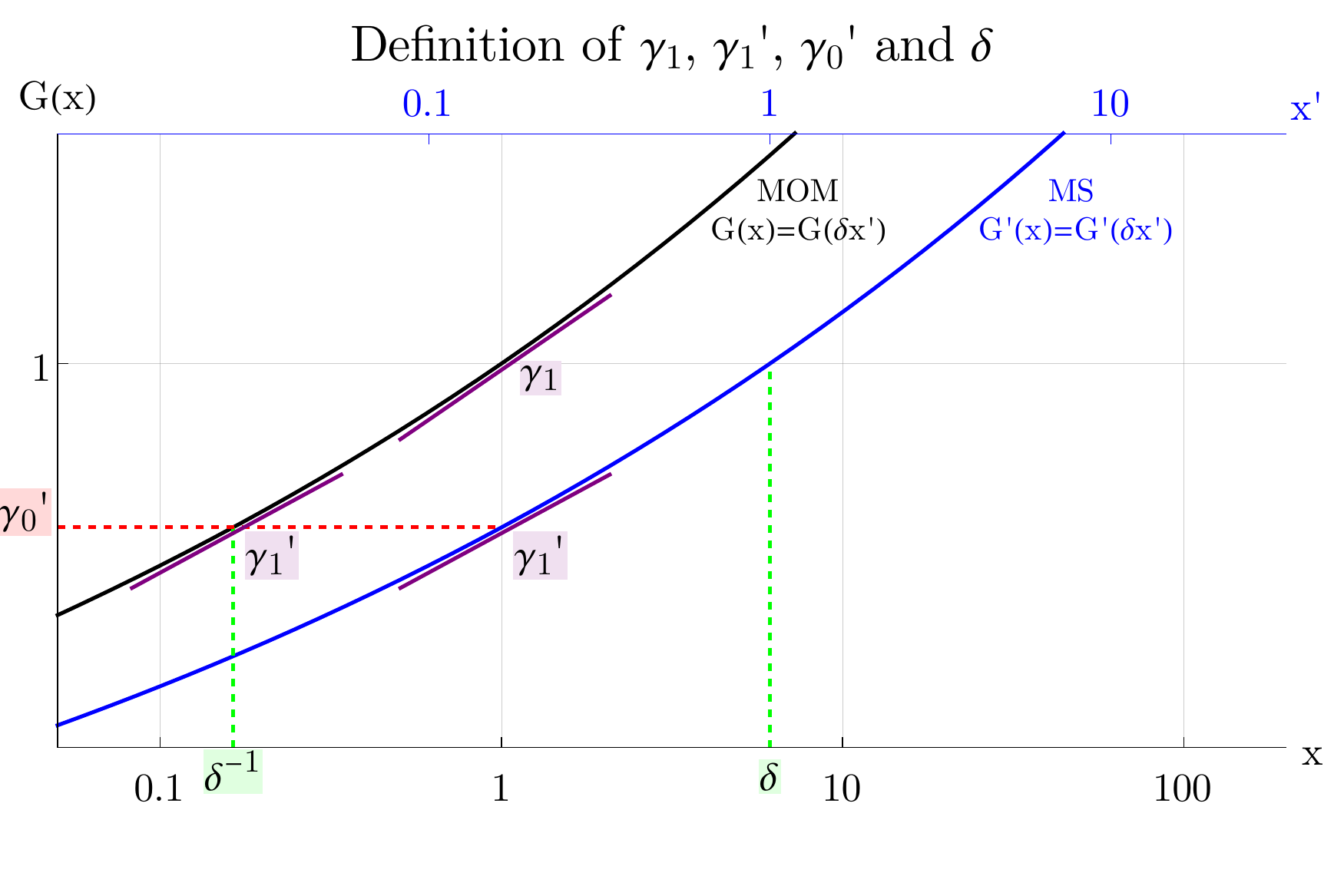}
	\caption{A hypothetical Green function in MOM ($G(x)$, black) and MS ($G'(x)$, blue). Both are initially given as functions of $x$ (lower horizontal axis). We define a rescaled variable $x'=x\cdot \delta^{-1}$ such that $G'(x')=G(x)$. This means that $x=\delta$ (indicated in green) is the would-be kinematic renormalization point $x'=1$. Conversely, $\gamma'_0$ is the value $G'(x=1)$ (red). This equals $G(x=\delta^{-1})$. The derivative $\gamma_1$ of $G(x)$ at $x=1$ (purple) is different from the MS-derivative $\gamma'_1$ at $x=1$, which in turn equals the derivative of $G(x)$ at $x=\delta^{-1}$. It is $\gamma=\gamma_1$, but the new anomalous dimension $\gamma'$ does not have a graphical representation in this plot since it is not defined as a simple derivative of $G'$ with respect to $x$ or $x'$, see \cref{gamma_def}. }
	\label{pic_def_gamma}
\end{figure}

\begin{lem}\label{thm_shifted}
	Assume that the expansion functions $\gamma_k(\alpha)$ from \cref{Gren_gammak} are formal power series and satisfy the Callan-Symanzik equation \cref{rge_MOM}, and the kinematic renormalization point is shifted by a factor  $\delta(\alpha)$ according to \cref{G_Gs_def}, which is a power series in $\alpha$ as well.  Then:
	\begin{enumerate}
		\item The new expansion functions $\gamma'_k(\alpha)$ are given by \cref{Gren_gammaks} and they again fulfil a Callan-Symanzik equation, but with the new anomalous dimension and beta function
		\begin{align}\label{shifted_gammabeta}
			\gamma'(\alpha) &:= \frac{\gamma (\alpha)}{   1+s \gamma(\alpha)\cdot  \alpha  \partial_\alpha \ln\delta(\alpha)   }, \qquad \beta'(\alpha) := s \gamma'(\alpha),
		\end{align}
	\item The shifted anomalous dimension satisfies $\gamma'(\alpha) = \gamma(\alpha) + \mathcal{O} \left( \alpha^3 \right) $.
	\end{enumerate}
\end{lem}
\begin{proof}
	(1) The series representation of $\gamma'_k(\alpha)$ in \cref{Gren_gammaks} follows algebraically from expanding $\ln (x \cdot \delta(\alpha)) = \ln x + \ln \delta(\alpha)$ in the original log expansion \cref{Gren_gammak}. Compute the derivative of this series, using the fact that $\gamma_j(\alpha)$ satisfy \cref{rge_MOM}, and identify the resulting series to obtain
	\begin{align*}
		\alpha \partial_\alpha    \gamma'_k(\alpha)	&=  	-\frac{\gamma}{s\gamma}\gamma'_k  +\frac 1 {s\gamma}\sum_{j=k+1}^\infty \frac{j!(k+1)\gamma_{j} \left( \ln \delta \right) ^{j-1-k} }{(j-1-k)!(k+1)!} + \alpha \partial_\alpha \ln \delta \cdot \sum_{j=k}^\infty \frac{j!(k+1) \gamma_j   \left( \ln \delta \right) ^{j-k-1}  }{(j-k-1)!(k+1)!} \\
		(k+1)  \gamma'_{k+1} &= \frac{\gamma}{1+s\gamma \alpha \partial_\alpha \ln \delta}\cdot   \gamma'_k + \frac{s\gamma}{1+s\gamma \alpha \partial_\alpha \ln \delta}\cdot \alpha \partial_\alpha   \gamma'_k.
	\end{align*}
This is again the Callan-Symanzik equation, but with a different anomalous dimension and beta function as claimed in \cref{shifted_gammabeta}.

(2) Follows from \cref{gammak_orders,shifted_gammabeta} upon noting that $\gamma(\alpha) \cdot \alpha \partial_\alpha \ln \delta(\alpha) \in \mathcal{O}(\alpha^2)$.
\end{proof}

We observe from \cref{shifted_gammabeta} that for a linear Dyson-Schwinger equation ($s=0$), the anomalous dimensions in MS and MOM agree. On an algebraic level, this was remarked in \cite[Ex.  5.12]{panzer_renormalization_2015}. 

Note further that, if $\delta(\alpha)$ depends on $\alpha$, the anomalous dimension $\gamma'(\alpha)$ of the shifted solution is not equal to $\gamma'_1(\alpha)$, the first derivative of the Green function \cref{gamma1s}. These two functions coincide only in the case of a kinematic renormalization scheme with a fixed ($\alpha$-independent) renormalization point.
If $\delta(\alpha)$  depends on $\alpha$  and $s\neq 0$ then there are two distinct effects happening at the same time: First, $\gamma'_1(\alpha)\neq \gamma_1(\alpha)$ because these derivatives are taken at different points, see \cref{pic_def_gamma}, and second, $\gamma'(\alpha)\neq \gamma(\alpha)$ because the moving renormalization point influences the Callan-Symanzik equation and the definition \cref{gamma_def}. 

In the remainder of the paper, we will be concerned with a reverse situation to \cref{thm_shifted}. We are given two sets of functions, $\lbrace \gamma_j\rbrace$ and $ \lbrace\gamma'_j\rbrace$, for example by expanding two known Green functions according to \cref{Gren_gammak}, and want to find $\delta(\alpha)$. We will mostly consider cases where both $\lbrace \gamma_j\rbrace$ and $ \lbrace\gamma'_j\rbrace$ are solutions of the same Dyson-Schwinger equation and hence both fulfil a Callan-Symanzik equation \ref{rge_MOM}. 
But generally, the CSE is neither necessary nor sufficient for $\delta(\alpha)$ to exist. For example, let $\gamma_k(\alpha) = k \alpha^k$ and $\gamma'_k(\alpha) = \alpha^k (\alpha^2 + k) (1 - \alpha^2)^{-2 - k} $, then neither $\gamma_j$ nor $\gamma'_j$ fulfil a CSE but they are related according to \cref{Gren_gammaks} with $\ln \delta (\alpha) = \alpha$. Conversely, the \emph{chain approximation}  is an example where the CSE is not satisfied and $\delta(\alpha)$ does not exist, see \cref{chains}. Finally, two solutions of \emph{different} Dyson-Schwinger equations do both fulfil a CSE and still they are in general not related by a $\delta(\alpha)$ as they fail to satisfy \cref{thm_shifted} (2).

\subsection{Linear case}\label{linear}

The goal of this paper is to determine the shift $\delta(\alpha)$ between kinematic renormalization (with a fixed renormalization point $x=1$) and MS-renormalization.  If we were to know the exact MS-solution $\hat G(\alpha,\hat x)$ then this amounts to finding the point $\hat x=\delta^{-1}(\alpha)$ where $\hat G(\alpha,\delta^{-1})=1$.

\begin{lem}\label{lem_linear_delta}
	Let $G(\alpha,x)$ and $G'(\alpha,x)$ be perturbative solutions of linear Dyson-Schwinger equations ($s=0$ in both cases) with equal anomalous dimensions $\gamma(\alpha)=\gamma'(\alpha)$. Assume that  $\gamma_0(\alpha), \gamma_0'(\alpha) = 1+ \mathcal O (\alpha)$ are formal power series starting with unity. Then the Green functions coincide in the sense of \cref{G_Gs_def} if one chooses the $\alpha$-dependent renormalization point given by the power series
	\begin{align}\label{linear_delta1} 
		\ln \delta(\alpha) &= \frac{1}{\gamma(\alpha)} \ln \frac{\gamma'_0(\alpha)}{\gamma_0(\alpha)}.
	\end{align}
\end{lem}
\begin{proof}
	By \cref{thm_shifted} for $s=0$, a necessary condition for a shift $\delta(\alpha)$ to exist is that the two Green functions must have the same anomalous dimension, which is guaranteed by assumption. It remains to show the reverse, that given the sets of functions $\lbrace \gamma_j \rbrace$ and $\lbrace \gamma'_j\rbrace$, it is always possible to find $\delta(\alpha)$.

	In the linear case, $s=0$, the solution of the Callan-Symanzik equation \cref{rge} is a monomial $G(\alpha,x)=\gamma_0(\alpha) x^{\gamma(\alpha)}$. Using \cref{G_Gs_def}, we demand $\gamma_0'  = \gamma_0  \delta^{\gamma}$, which leads to the claimed formula.  The functions $\gamma(\alpha), \gamma_0(\alpha)$ and $\gamma_0'(\alpha)$ are power series and $\gamma'_0(\alpha)/\gamma(\alpha) = 1 + \mathcal O (\alpha)$  by assumption. Therefore $\ln(\gamma'_0/\gamma_0) \in \mathcal{O}(\alpha)$ and the pole $1/\alpha$ of $1/\gamma(\alpha)$ from \cref{gammak_orders} is cancelled. The right hand side of \cref{linear_delta1} is a formal power series indeed.
\end{proof}

In MOM, we have $\gamma_0(\alpha)=1$ by the renormalization condition \cref{MOM_def}, and hence \cref{scaling_solution}. In MS, the solution will have some $\hat \gamma_0(\alpha)\neq 1$, which is the value of $G'(\alpha,1)$, see \cref{gamma1s}. \Cref{lem_linear_delta} thus specializes to
\begin{align}\label{linear_delta}
	  \qquad \ln \delta(\alpha) &= \frac{\ln  \hat \gamma_0}{\gamma(\alpha)}.
\end{align}

It is also possible to infer $\hat \gamma_0(\alpha)$ from the MOM-solution alone. To this end, note that the MS Green function is proportional to the one in MOM. Going back to the definition \cref{renormalization} of the $Z$-factors, this means that $\hat Z G_0 = \hat \gamma_0 \cdot Z G_0+ \mathcal{O}(\epsilon)$. From the integral representation \cref{Z_beta_gamma}, using that in MS $\hat \gamma(\alpha,\epsilon)=\hat \gamma(\alpha)$, we get 
\begin{align*}
	 e^{ -\int \limits_0^{\alpha } \frac{\d u}{u} \frac{\hat \gamma \left(u \right)  }{\epsilon    } }&=\hat \gamma_0(\alpha)\cdot e^{ -\int \limits_0^{\alpha } \frac{\d u}{u} \frac{\gamma \left(u,\epsilon \right)  }{\epsilon    }} + \mathcal{O}(\epsilon) \qquad \Rightarrow \qquad \gamma(\alpha,\epsilon) = \gamma(\alpha) - \epsilon \alpha \partial_\alpha \ln \hat \gamma_0(\alpha) + \mathcal O(\epsilon^2).
\end{align*}
If we know the MOM anomalous dimension for $\epsilon \neq 0$ then we can compute $\hat \gamma_0$ and hence $\ln \delta$ from
\begin{align}\label{linear_delta2}
	\alpha \partial_\alpha \ln \hat \gamma_0(\alpha) &= -\left[ \epsilon^1 \right] \gamma(\alpha,\epsilon).
\end{align}
In fact, for a linear DSE all coefficients of the MOM counter term $Z(\alpha,\epsilon)$ are directly given by  the $\epsilon$-expansion of the MOM anomalous dimension via  \cref{Z_beta_gamma}:
\begin{align}\label{linear_Z_expansion}
	Z(\alpha,\epsilon) &=: \exp \left( -\sum_{n=-1}^\infty \epsilon^n\cdot z_n (\alpha)  \right)  , \qquad \alpha \partial_\alpha z_n(\alpha) =   [\epsilon^{n+1}]\gamma(\alpha ,\epsilon) , \qquad z_0(\alpha) = \ln \hat \gamma_0(\alpha).
\end{align}
For the counter term $\hat Z(\alpha,\epsilon)$ in MS,  all $\hat z_{n\geq 0}(\alpha)$ vanish, as do the $\epsilon$-dependent parts of $\hat \gamma(\alpha,\epsilon)=\hat \gamma(\alpha)$.

\subsection{Non-linear case}\label{methods_nonlinear}

\begin{lem}\label{lem_nonlinear}
	Assume $\gamma_k(\alpha)$ and $\gamma(\alpha)$ are the expansion coefficients resp. the anomalous dimension in kinematic renormalization, and $\gamma'_k(\alpha), \gamma'(\alpha)$ the corresponding quantities in MS, of the perturbative solution of the same  Dyson-Schwinger equation of type \cref{dse_general}. Then the first two orders of the anomalous dimensions coincide, $\gamma'(\alpha) = \gamma(\alpha) + \mathcal O \left( \alpha^3 \right) $.
 
\end{lem}
\begin{proof}
Let $f^{(k)}_n$ be the coefficients of the $\epsilon$-expansion of the kernel graph according to \cref{series}. Then the first coefficients of an explicit perturbative solution of \cref{dse_general} in MOM resp. MS are
	\begin{align*}
			\gamma(\alpha) &= -f^{(0)}_{-1} \alpha + (s+1)\left( -2 f^{(0)}_{-1} f^{(1)}_0 - 2f^{(0)}_0 f^{(1)}_{-1} + 2 f^{(0)}_{-1} f^{(0)}_0 \right) \alpha^2 + \mathcal O \left( \alpha^3 \right), \\
			\hat \gamma_1(\alpha) &= -f^{(0)}_{-1} \alpha + (s+1)\left( -2 f^{(0)}_{-1} f^{(1)}_0 - 2 f^{(0)}_0 f^{(1)}_{-1} +  f^{(0)}_{-1} f^{(0)}_0 \right) \alpha^2 + \mathcal O \left( \alpha^3 \right),  \\
			\hat \gamma_0(\alpha) &= 1+\alpha f^{(0)}_0 + \mathcal{O}(\alpha^2), \qquad \alpha \partial_\alpha \hat G \big|_{x=1} = \alpha f^{(0)}_0 + \mathcal{O}(\alpha^2).
	\end{align*}
	 Using the Callan-Symanzik equation  \cref{cse} and \cref{gamma1s}, the anomalous dimension in MS is
	\begin{align*}
		\hat \gamma(\alpha) &= \frac{\hat \gamma_1(\alpha)}{\hat \gamma_0(\alpha) + s \alpha \partial_\alpha \hat G |_{x=1}} = \frac{\hat \gamma_1(\alpha)}{ 1+ (s+1)\alpha f^{(0)}_0   + \mathcal{O}\left( \alpha^2 \right)   } = \gamma(\alpha) + \mathcal{O}\left( \alpha^3 \right).  
	\end{align*}
\end{proof}

We now assume that we know the MOM- and the MS-solution and their corresponding counter terms, and hence the renormalization group functions, by explicit calculation. The remaining task is then to extract the shift $\delta(\alpha)$ from this data. 

The constant coefficient of the power series $\ln \delta(\alpha)$ can be inferred from \cref{Gren_gammaks}, $\hat \gamma_0 = 1 + \gamma_1 \ln \delta + \mathcal{O}(\alpha^2)$. We insert the series from the proof of \cref{lem_nonlinear} and read off
\begin{align}\label{lndelta_1}
	\ln\delta(\alpha) &= - f^{(0)}_0  (f^{(0)}_{-1})^{-1} + \mathcal O \left( \alpha \right) .
\end{align}
Note that $f^{(0)}_{-1} \neq 0$ in physically sensible kernels. 
Remarkably, $\delta(0)= e^{- f^{(0)}_0 / f^{(0)}_{-1}}\neq 1$ unless $f^{(0)}_0=0$, so the shift does not necessarily vanish  for vanishing coupling. This result does not depend on the invariant charge in the DSE, or whether it is linear or non-linear, but just on the primitive kernel.

\begin{thm}\label{thm}
	Let $G(\alpha,x)$ and $\hat G(\alpha,x)$ be the perturbative solutions of the same propagator-type Dyson-Schwinger equation \cref{dse_general}, where $G$ uses kinematic renormalization and $\hat G$  Minimal Subtraction. Assume that $\gamma(\alpha),\hat \gamma(\alpha)$ are power series with a non-vanishing term $\propto \alpha$. Then there is a unique power series $\delta(\alpha)$ such that $G(\alpha, \delta(\alpha) \cdot x) = \hat G(\alpha,x)$ for all $x$, given by \cref{lndelta_1} and
		\begin{align}\label{rge_dDelta2}
		\frac{\partial}{\partial \alpha} \ln \delta(\alpha) &=  \frac{1}{s \alpha} \left( \frac{1}{\hat \gamma(\alpha)}-\frac{1}{  \gamma(\alpha)}\right) = \frac{\gamma(\alpha)-\hat \gamma(\alpha)}{s \alpha \hat \gamma(\alpha) \gamma(\alpha)} .
	\end{align}
\end{thm}
\begin{proof}
	The function $\delta(\alpha)$ for the linear case $s=0$ was constructed explicitly in \cref{lem_linear_delta}. It remains to consider $s\neq 0$.
	
	The fact \emph{that} MS and MOM are related via a change in renormalization point, or equivalently, via a change in the value of the renormalized coupling, is known from a Hopf-algebraic analysis, see \cite{kreimer_chen_1999,kreimer_renormalization_2013,panzer_renormalization_2015}. It remains to show that, in our setup, the shift $\delta(\alpha)$ is a  well defined power series.
	
	From \cref{thm_shifted} we know how shifting the kinematic renormalization point induces a change in the anomalous dimension. Solving  \cref{shifted_gammabeta} for $\delta(\alpha)$ produces \cref{rge_dDelta2}.

By \cref{gammak_orders,gammaks_orders}, and assumption, the denominator of the last fraction is proportional to $\alpha^3$. But, since $\hat \gamma(\alpha)$ is the anomalous dimension in MS, the numerator is $\gamma(\alpha)-\hat \gamma(\alpha) \in \mathcal{O} (\alpha^3)$ by \cref{lem_nonlinear}. Therefore the right hand side of \cref{rge_dDelta2} is a well defined power series in $\alpha$. It uniquely defines the power series $\delta(\alpha)$ up to a constant summand, which is fixed by \cref{lndelta_1}. 

Since we know that MS and MOM are related via a shifted renormalization point, the so-constructed shift $\delta (\alpha)$ necessarily also gives rise to the correct $\hat \gamma_0(\alpha)$, uniquely defined via
\begin{align*}
	\hat \gamma_0(\alpha) &= \sum_{j=0}^\infty \gamma_j(\alpha) \left( \ln \delta(\alpha) \right) ^j.
\end{align*}
By assumption, the Green function $\hat G(\alpha,x)$ satisfies the Callan-Symanzik equation with anomalous dimension $\hat \gamma(\alpha)$, hence by \cref{thm_shifted} from $\hat \gamma_0(\alpha)$ also all other functions $\hat \gamma_j(\alpha)$ are reproduced correctly.
\end{proof}

Given the solutions of a DSE in MOM and MS, there are at least  three approaches to calculate $\delta(\alpha)$ from this data. 
The first approach directly uses \cref{rge_dDelta2}, where the anomalous dimensions $\gamma(\alpha),\hat \gamma(\alpha)$ can be extracted from the corresponding $Z$-factors with the help of \cref{Z_beta_gamma,MS_Z_poles}.

The second approach utilizes the renormalization group equation in MS derived in \cref{thm_shifted},
\begin{align}\label{rge_MS}
	(k+1)  \hat \gamma _{k+1}(\alpha) &= \frac{\gamma(\alpha)}{1+s\gamma(\alpha) \alpha \partial_\alpha \ln \delta (\alpha)}\cdot  \left( 1+s \alpha \partial_\alpha \right)   \hat \gamma_k(\alpha).
\end{align}
If any two of the MS functions $\hat \gamma_k(\alpha)$, together with the MOM anomalous dimension $\gamma(\alpha)$,  are known, then $\delta(\alpha)$ can be computed. For example, using $\hat \gamma_0$ and $\hat \gamma_1$, one has
\begin{align}\label{rge_dDelta}
	  \frac{\partial}{\partial \alpha} \ln  \delta(\alpha ) &= \frac{ \gamma \cdot  \hat \gamma_0 -  \hat \gamma_1 }{ s \alpha \cdot  \hat \gamma_1 }+\frac{  1  }{  \hat \gamma_1  } \frac{\partial}{\partial \alpha} \hat \gamma_0 \qquad (\text{for }s \neq 0) .
\end{align}

The third, and computationally most efficient, approach  is to compute all MS functions $\hat \gamma_j(\alpha)$ up to some desired maximum $j$ and additionally all MOM functions $\gamma_j(\alpha)$. Next, one writes a power series ansatz for $\ln \delta(\alpha)$ and uses this to formally compute the powers $(\ln \delta(\alpha))^k$. Then the right side of \cref{Gren_gammaks} is a linear system for the unknown coefficients of $\ln \delta(\alpha)$ which can be solved.

All three approaches are different ways to solve the same equations, therefore the resulting $\ln \delta(\alpha)$ agree. The third approach does not involve a derivative and therefore it produces one order higher in $\alpha$ compared to the first two, for the same order of input data.

We want to stress that, despite the above considerations, it is in general not possible to recover the MS Green function from the MOM one or vice versa when only the limit $\epsilon\rightarrow 0$ is known.  The shift $\delta(\alpha)$ is a truly unknown function which depends on the particular DSE and on the integral kernel. 
In the remainder of the paper, we will explicitly compute $\delta(\alpha)$ for three different Dyson-Schwinger equations.

\section[Linear DSE in D=4 dimensions]{Linear DSE in $D=4-2\epsilon$ dimensions}\label{D4}

We first consider  a linear Dyson-Schwinger equation of iterated one-loop Feynman graphs, namely
\begin{align}\label{dse_linear_ms}
	\hat G  ( q^2) &= 1 + \lambda \int \frac{\d^D  k}{(2\pi)^D} \frac{\hat G  ( k^2)}{\left(  k + q \right) ^2  k^2} - \hatrenopk{\lambda \int \frac{\d^D  k}{(2\pi)^D} \frac{\hat G ( k^2)}{\left(  k + q \right) ^2  k^2}}.
\end{align}
Here, $\lambda$ is a coupling constant and $D=4-2\epsilon$ is the spacetime dimension. \Cref{dse_linear_ms} is is the linear DSE (rainbow approximation) for the fermion propagator in Yukawa theory, but scaled and projected onto suitable tensors such that the order zero (tree level) solution is
\begin{align}\label{linear_4d_G0}
	\hat G ^{(0)} \left( q^2 \right) &:= 1.
\end{align}

\subsection{Computation of the coefficients}\label{linear_computation}

One can solve the DSE \cref{dse_linear_ms} to first order by   inserting \cref{linear_4d_G0} into it and computing the integrals according to \cref{series}:
\begin{align}\label{linear_4d_gren1}
	\hat G ^{(1)}(  q^2) 
	&= 1 +  \frac{\lambda}{(4\pi)^2}\left( 1-\hatrenop  \right)  \left[ (4 \pi)^{\epsilon} e^{-\gamma_E \epsilon}\left( \frac{q^2}{m^2} \right) ^{ - \epsilon}\sum_{w=-1}^\infty f_w^{(0)} \epsilon^w \right].
\end{align}
Here, $m^2$ is an arbitrary mass scale introduced for dimensional reasons and the coefficients $f^{(k)}_w$ are given in \cref{linear_4d_fkn}. The operator $\hatrenop$ projects the pole part of this equation, which is $f^{(0)}_{-1} \epsilon^{-1}$. This pole part is independent of momenta as expected in a locally renormalizable theory. 

For a systematic treatment of higher orders, it will be advantageous to include the counter term as a summand $\propto 1=(q^2)^0$ into the solution. Define
\begin{align*}
	\bar g_{1,w}^{(1)} &:= f_w^{(0)}\qquad &	\bar g_{0,w}^{(1)} &:= \begin{cases}
		-f_{-1}^{(0)} & w=-1\\
		0 & \text{else}
	\end{cases} 
\end{align*}
then the renormalized solution to order one, \cref{linear_4d_gren1}, is 
\begin{align}\label{linear_4d_gren1_2}
	\hat G^{(1)}(  q^2) 	&= 1 + \frac{\lambda}{(4\pi)^2} \sum_{t=0}^1 (4\pi)^{t\epsilon} e^{-\gamma_E t\epsilon} \left(  \frac{ q^2}{m^2} \right) ^{-t\epsilon} \sum_{w=-1}^\infty \bar g_{t,w}^{(1)} \epsilon^w.
\end{align}
The factor $(4\pi)^{t\epsilon} e^{-\gamma_E t\epsilon}$ eventually produces finite contributions $\propto \gamma_E$ and $\propto \ln(4\pi)$. In MS renormalization, these terms are present in the finite part of $\hat G $ while in MS-bar renormalization they are assigned to the counter term $\bar Z$ and thus absent from $\bar G $. To facilitate computations, we will absorb them into the momentum variable. This way, we effectively obtain the MS-bar Green function of the new momentum variable. On the other hand, our counter term will be  $\hat Z$ in MS as it does not contain the $\ln(4\pi),\gamma_E$ contributions either. This way, we can skip in every intermediate step two additional series expansions, the ones of $e^{-\gamma_E t\epsilon}$ and $e^{\ln(4\pi)t\epsilon}$. Let therefore
\begin{align}\label{MS_MSbar}
	\hat x := \frac{q^2}{m^2},\qquad  \bar x &:= \frac{e^{\gamma_E} q^2  }{4\pi m^2}   \equiv \frac{q^2}{  \bar m^2} , \qquad \hat G \left( \hat x \right)  \equiv \bar G  (\bar x). 
\end{align}
The transition $\hat x \leftrightarrow \bar x $ is a rescaling of the momentum which can also be understood as choosing a suitable value of the reference momentum $m$, as evidenced by the second equal sign where $ \bar m^2 = 4\pi e^{-\gamma_E}m^2$. This is not kinematic renormalization: $m$ is merely the unit used for the momenta, the MS-renormalized Green function $\hat G (q^2)$ does not fulfil any particular condition at the point $q^2=m^2$. Also, $m^2$ is not the mass of a particle, the field is still massless. Finally, let $\alpha := \lambda (4\pi)^{-2}$.

Higher orders of the renormalized Green function are computed iteratively. Assume that we know the solution to order $(m-1)$,
\begin{align*}
	\bar G ^{(m-1)}(\alpha,\bar x) &= \bar G ^{(m-2)}(\alpha, \bar x) +  \alpha^{m-1}  \sum_{t=0}^{m-1}  \left( \bar x \right) ^{-t\epsilon} \sum_{w=-(m-1)}^\infty  \bar g^{(m-1)}_{t,w}\epsilon^w.
\end{align*}
Inserting this into \cref{dse_linear_ms}, the order $m$ solution is
\begin{align}\label{linear_4d_Gren_MSbar_m}
	\bar G ^{(m)} (\alpha,\bar x) 
	 &= \bar G ^{(m-1)} (\alpha,\bar x)   + \alpha^m   \sum_{t=0}^{m}   \left( \bar x \right) ^{-t\epsilon} \sum_{w=-m}^\infty  \bar g_{t,w}^{(m)} \epsilon^w ,
\end{align}
where the new coefficients are determined by
\begin{align}\label{linear_4d_recursion_MSbar}
	\bar g^{(m)}_{u,w} &:= \sum_{n=-m+1}^{w+1} \bar g^{(m-1)}_{u-1,n} f^{(u-1)}_{w-n} \qquad \forall u>0.
\end{align}
The counter term is included in this sum as the $t=0$ summand. 
The coefficient of order $\epsilon^w$ in the MS-counter term at $m$ loops is
\begin{align}\label{linear_4d_recursion_g0_MSbar}
	\bar g_{0,w}^{(m)} &= 	-\sum_{u=1}^m \bar g_{u,w}^{(m)}\quad  \text{ if } w\in \left \lbrace -m,\ldots,-1 \right \rbrace, \text{ and }\bar g_{0,w}^{(m)} =	0 \text{ for all other }w.
\end{align}

The all-order perturbative solution $\bar G(\bar x)$ of \cref{dse_linear_ms} is defined as the limit $m\rightarrow \infty$ in \cref{linear_4d_Gren_MSbar_m}, effectively it is an infinite sum over the orders $\alpha^m$ in the coupling. We exchange the sums to expose the  expansion \cref{Gren_gammak}   and the counter term $\hat Z$:
\begin{align}\label{linear_4d_Gren_MSbar_series}
	\bar G  \left(\alpha,\bar x \right) &= \hat Z(\alpha,\epsilon) + \sum_{k=0}^\infty \ln\left( \bar x \right)^k \sum_{t=1}^{\infty }  \frac{(-t  )^k}{k!}  \sum_{m=t  }^\infty  \alpha^m  \sum_{w=-m}^\infty  \bar g_{t,w}^{(m)} \epsilon^{k+w}  .
\end{align} 
As explained above \cref{MS_MSbar}, we have introduced the  MS (not MS-bar) counter term,
\begin{align}\label{linear_4d_counterterm_MS}
	\hat Z(\alpha,\epsilon) &:= 1+ \sum_{m=1}^\infty \alpha^m \sum_{w=-m}^{-1} \bar g^{(m)}_{0,w} \epsilon^w = 1-\sum_{m=1}^\infty \alpha^m \sum_{w=-m}^{-1} \sum_{t=1}^m \bar g_{t,w}^{(m)} \epsilon^w.
\end{align}
The \cref{linear_4d_Gren_MSbar_series,linear_4d_counterterm_MS} together with the recursion relations \cref{linear_4d_recursion_MSbar,linear_4d_recursion_g0_MSbar} allow us to compute the solution of the DSE \ref{dse_linear_ms} to arbitrary order.

A similar procedure can   be used to obtain the MOM-renormalized finite Green function $G(\bar x)$. One merely has to extend \cref{linear_4d_recursion_MSbar} to include all orders $w$ of $\epsilon$. If we are only interested in the finite (as $\epsilon \rightarrow 0$) part, then, in a linear DSE, it is sufficient to include the $w=0$ term into the counter term. To this end,  instead of \cref{linear_4d_recursion_MSbar,linear_4d_recursion_g0_MSbar} one uses
\begin{align}\label{linear_4d_recursion_MOM}
	g^{(m)}_{u,w} &:= \begin{cases}     \sum_{n=-m+1}^{w+1}   g^{(m-1)}_{u-1,n} f^{(u-1)}_{w-n}  \quad  &\forall u>0 \\
	 	-\sum_{u=1}^m  g_{u,w}^{(m)} &\text{if } u=0 \text{ and } w\in \left \lbrace -m,\ldots,0 \right \rbrace, \text{ and }	0 \text{ else }. 
	 	\end{cases}
\end{align}
We will call this prescription \enquote{pseudo-MOM-scheme}, since the counter term computed by \cref{linear_4d_recursion_MOM} is not the true counter term of kinematic renormalization, it misses higher orders in $\epsilon$. But this prescription produces  finite Green function
\begin{align}\label{linear_4d_Gren_MOM_series}
	 G(\alpha,\bar x) &= Z(\alpha,\epsilon)+   \sum_{k=0}^\infty \ln \left(  \bar x \right)^k \frac{1}{k!} \sum_{m=1  }^\infty  \alpha^m \sum_{t=1}^{m}   (-t  )^k  \  \sum_{w=-m}^\infty   g_{t,w}^{(m)} \epsilon^{k+w}.    
\end{align}
which for $\epsilon \rightarrow 0$ is conventionally MOM-renormalized. It by construction takes the value unity at the renormalization point $\bar x=1$, i.e. at the momentum $q^2 = \mu^2$. Now the interpretation of $\mu^2$ has changed: In the MS- or MS-bar-solutions, it was an arbitrary momentum scale without  particular meaning for the Green function while in MOM it is the momentum where $G(\alpha,q^2/\mu^2)=1$.

\subsection{Renormalized correlation function}

In MOM renormalization, the analytic solution of the this linear DSE has long  been known  \cite{delbourgo_dimensional_1996}. As always for a linear DSE, it is a pure scaling solution where the anomalous dimension $\gamma(\alpha)$ is determinded from  the Mellin transform of the primitive, see \cref{mellin_gamma1,mellin}. With the notation of \cref{Gren_gammak}, the MOM-solution at $\epsilon=0$ reads
\begin{align}\label{linear_4d_Gren_MOM_theory}
	 G(\alpha,\bar x) &=  \bar  x ^{   \gamma(\alpha)}\\
	 \text{ where} \quad \gamma_0(\alpha) =1, \quad \gamma(\alpha) &= \frac{ \sqrt{1-4\alpha}-1}{2} = -\sum_{n=1}^\infty C_{n-1} \alpha^n ,\qquad 	   \gamma_k(\alpha) = \frac{ 1 }{k!}  \left(\gamma(\alpha)\right)^k. \nonumber
\end{align}
Here, $C_n$ are the Catalan numbers. It has been verified symbolically up to order $\alpha^{25}$ that the series  \cref{linear_4d_Gren_MOM_series} indeed coincides with \cref{linear_4d_Gren_MOM_theory} in the limit $\epsilon \rightarrow 0$.

In MS-bar-renormalization, there is a finite remainder term in the $\epsilon^0$-coefficient of each order in $\alpha$, which would have been subtracted in kinematic renormalization. Therefore, the MS-bar-coefficients $\bar \gamma_k(\alpha)$ in \cref{Gren_gammak} are generally different from the MOM-coefficients $\gamma_k(\alpha)$, see \cref{shifted}.
We compute $\delta(\alpha)$ from $\bar \gamma_0(\alpha)$ via \cref{linear_delta}.
From \cref{linear_4d_Gren_MSbar_series} one reads off
\begin{align}\label{linear_4d_deltabar}
	 \bar \gamma_0(\alpha) &=1+\sum_{m=1}^\infty \alpha^m \sum_{t=1}^m   \bar g^{(m)}_{t,0} 
	  =1+2 \alpha + \frac{11}2 \alpha^2 + \frac{51-2\zeta(3)}{3} \alpha^3  + \frac{1341 -80 \zeta(3)}{24} \alpha^4  +\mathcal O \left( \alpha^5 \right)   .
\end{align}
We have used that the $\hat Z$-factor in MS has, at finite order, only terms singular in $\epsilon$ and therefore does not contribute to $\bar \gamma_k$ and the summand $t=0$ can be left out. For the interpretation as a change of renormalization point, the remaining functions $\bar \gamma_k(\alpha)$ have to be consistent with \cref{Gren_gammaks,linear_4d_Gren_MOM_theory}. It has been verified to order $\alpha^{25}$ and for $k \leq 15$ that indeed $\bar \gamma_k(\alpha) = \bar \gamma_0(\alpha) \cdot \gamma_k(\alpha)$.  

Remarkably, for the linear DSE \cref{dse_linear_ms} considered here, it is possible to find a closed formula for $\bar \gamma_0(\alpha)$. To do this, one computes the logarithm of the series \cref{linear_4d_deltabar} and repeatedly uses the OEIS \cite{oeis}. One can first identify the rational coefficients and their generating function. Subtracting that part one is left with
\begin{align*}
	  \ln\bar \gamma_0 - \ln \left( \frac{1-\sqrt{1-4\alpha}}{2\alpha (1-4\alpha)^{\frac 1 4}} \right)   
	  &= -\zeta(3) \left( 2 \frac{\alpha^3}{3}  + 8 \frac{\alpha^4}{4} +30\frac{\alpha^5}{5} + 112 \frac{\alpha^6}{6} + 420 \frac{\alpha^7}{7}  + 1584 \frac{\alpha^8}{8} + \ldots  \right) \\
	&\quad - \zeta (5) \left( 2\frac{\alpha^5}{5}+ 12 \frac{\alpha^6}{6} + 56 \frac{\alpha^7}{7} + 240 \frac{\alpha^8}{8} + 990 \frac{\alpha^9}{9}   +\ldots \right)-\zeta(7)\ldots .
\end{align*}
At least up to $\zeta(11)$ and $\alpha^{25}$, the coefficients of $\frac{\alpha^{j+m-1}}{j+m-1}$ in the term proportional to $\zeta(m)$ are given by the binomial coefficient $2 \binom{2j +m-3}{j-1}$. Assuming again that this holds universally, all series over $\alpha$ and then the remaining series in $\zeta(m)$ can be summed and yield known functions. The result is
\begin{align}\label{linear_4d_gammabar0}
	\bar \gamma_0 (\alpha) &= e^{\gamma_E(1-\sqrt{1-4\alpha})}   \frac{1-\sqrt{1-4\alpha}}{2\alpha\left( 1-4\alpha \right) ^{\frac 14}}    \frac{\Gamma \left( \frac 3 2 -\frac 12 \sqrt{1-4\alpha}    \right)  } {\Gamma \left(  \frac 12 + \frac 12  \sqrt{1-4\alpha} \right)  }
	= \frac{-\gamma}{\alpha }\sqrt{\frac{\d \; (-\gamma) }{\d \alpha}  }    e^{-2\gamma \gamma_E}    \frac{\Gamma \left( 1-\gamma   \right)  } {\Gamma \left(  1+\gamma \right)  }   .
\end{align}
In the latter form, $\gamma \equiv \gamma(\alpha)$ is the anomalous dimension from \cref{linear_4d_Gren_MOM_theory}. We remark that such  fraction of Euler gamma functions is not uncommon in the computation of multiedge Feynman graphs, compare for example \cite{bonisch_analytic_2021}.

With this function $\bar \gamma_0(\alpha)$ and \cref{linear_4d_Gren_MOM_theory}, the MS-bar-renormalized solution $\bar G(\bar x)$ in the limit $\epsilon \rightarrow 0$ reads explicitly
\begin{align}\label{linear_4d_Gren_MSbar}
	\bar G(\alpha,\bar x) &= \bar \gamma_0(\alpha) \cdot G(\alpha,\bar x) = \bar \gamma_0(\alpha)  \cdot \bar x   ^{\gamma(\alpha)}.
\end{align}
From \cref{MS_MSbar} one can reconstruct the MS-renormalized function $\hat G(\hat p^2)$:
\begin{align}\label{linear_4d_Gren_MS}
	\hat G(\alpha,\hat x) &= \bar G \left( \alpha,\frac{\hat x}{4\pi e^{-\gamma_E}} \right)  = \hat \gamma_0(\alpha)\cdot  \hat x  ^{\gamma(\alpha)}   \quad 
	\text{where} \quad \hat \gamma_0(\alpha) :=   (4\pi e^{-\gamma_E})^{-\gamma(\alpha)}   \cdot \bar \gamma_0(\alpha) 
\end{align}

Following \cref{shifted}, the correspondence between MS-, MS-bar- and MOM-renormalized Green functions can equivalently be expressed by their respective renormalization points. Let $\mu$ be the mass scale in MS- and MS-bar-renormalization. Then the Green function is unity at $x=\hat \delta^{-1}$ and $x=\bar \delta^{-1}$, respectively. Equivalently, $\hat \mu^2  = \hat \delta \cdot \mu^2$ and $\bar \mu^2=\bar \delta \cdot \mu^2$ are the  mass scales one needs to choose for MS and MS-bar, in order to reproduce kinematic renormalization at the scale $\mu^2$:
\begin{align*}
	G \left( \alpha,\frac{q^2}{\mu^2} \right) \equiv \bar G  \left(\alpha, \frac{q^2}{\bar \mu^2} \right) \equiv \hat G  \left(\alpha, \frac{q^2}{\hat \mu^2} \right) .
\end{align*}
By \cref{linear_delta,linear_4d_gammabar0,linear_4d_Gren_MOM_theory}, these scales are related via
\begin{align}\label{linear_4d_barmu_hatmu}
	\bar \delta (\alpha) &=  \bar \gamma_0^{\frac 1 {\gamma}} =   e^{-2} \left( 1- \frac 32 \alpha - \left( \frac{49}{24} -\frac 2 3 \zeta(3)\right) \alpha^2 + \mathcal O \left( \alpha^3 \right)   \right),\qquad   
	\hat \delta (\alpha)=   \hat \gamma_0^{\frac{1}{\gamma}} = \frac{\bar \delta(\alpha)}{\sqrt{4\pi e^{-\gamma_E}}}  .
\end{align}
The latter of course reproduces the transformation $m \leftrightarrow \bar m$ in \cref{MS_MSbar}. The zeroth order coefficient is $e^{-2}$ as expected from \cref{lndelta_1} where $f_{-1} = f^{(0)}_{-1}=1, f_0=f^{(0)}_0=2$. 
In any case, the shift between renormalization schemes  is a finite  function as long as $\alpha <\frac 1 4$, shown in \cref{linear_4d_functions_fig}.

\begin{figure}[htbp]
	\centering
	 \begin{subfigure}{.48\textwidth}
	 	\includegraphics[width=\linewidth]{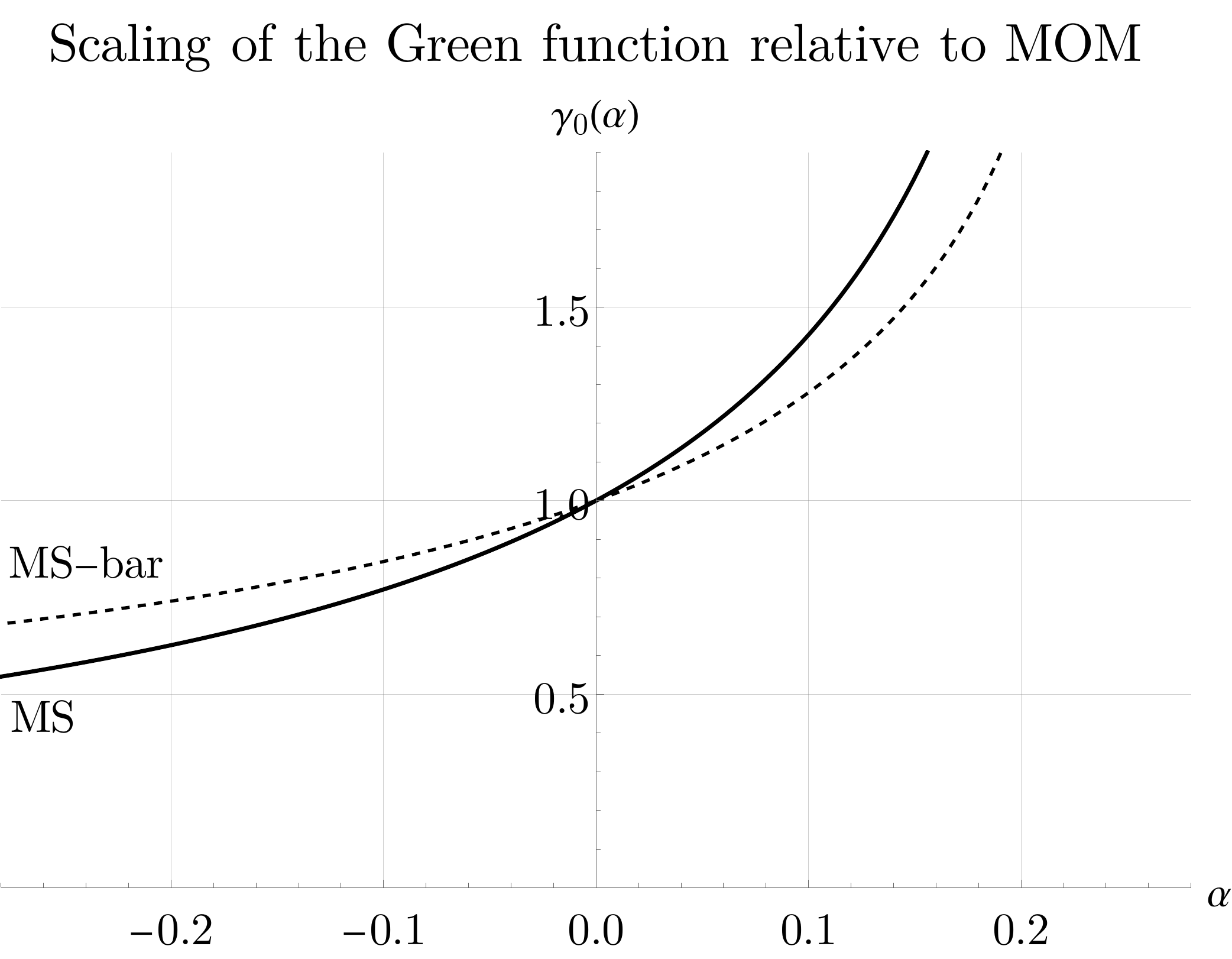}
	 	\caption{}
	 \end{subfigure}
 \hfill 
 \begin{subfigure}{.48\textwidth}
 	\includegraphics[width=\linewidth]{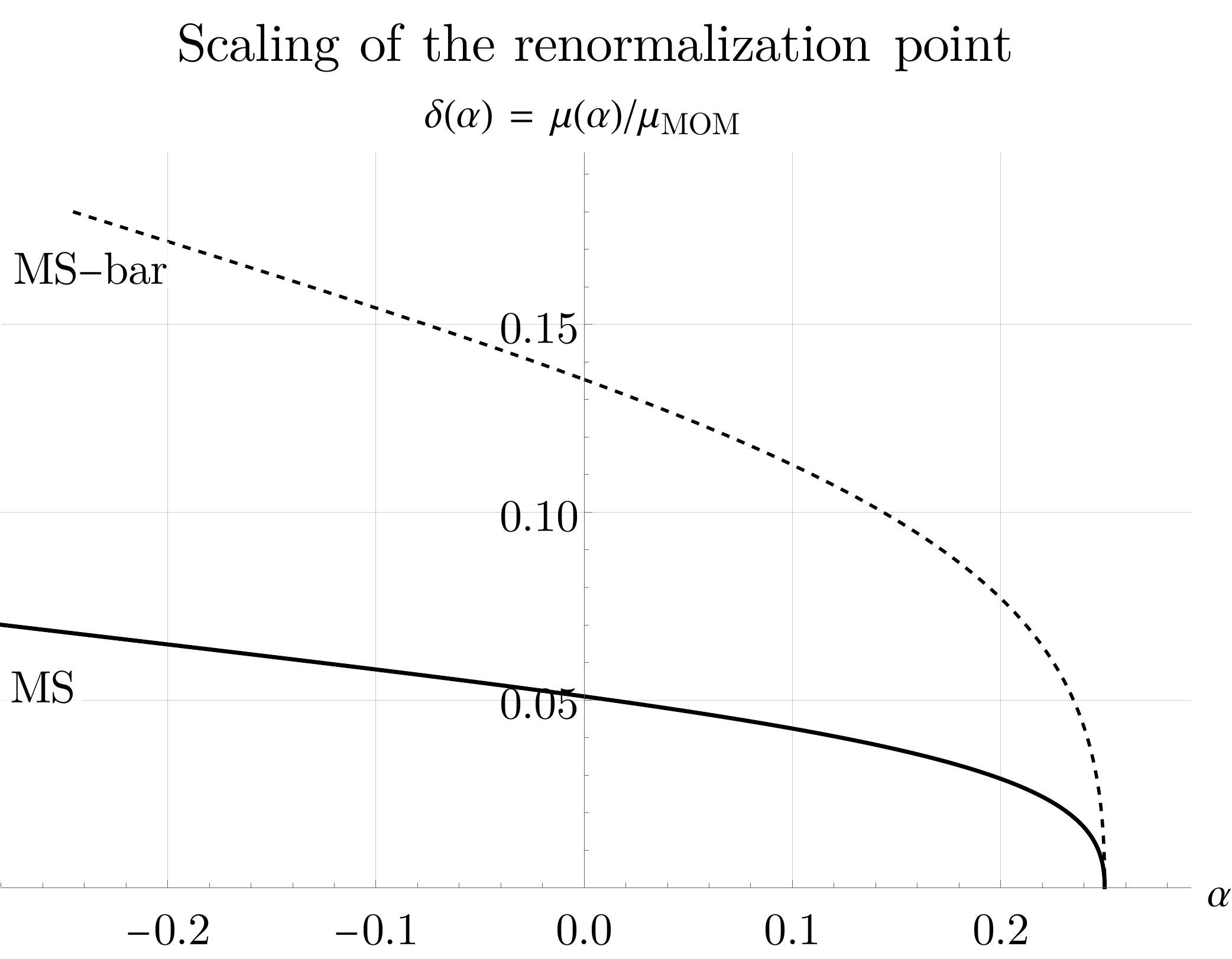}
 	\caption{}
 \end{subfigure}
\caption{(A) Behaviour of the MS scaling-function $\hat \gamma_0(\alpha)$ (thick) and the MS-bar scaling function $\bar \gamma_0(\alpha)$ (dashed) as functions of the renormalized coupling $\alpha$.  Both are unity at $\alpha=0$ and diverge at $\alpha = \frac 1 4$. (B) Rescaling factor $\delta(\alpha)$ according to \cref{shifted} of the reference momentum in MS-renormalization $\hat \mu$ (thick) and MS-bar-renormalization $\bar \mu$ (dashed)  relative to the MOM-renormalization-point $\mu$. The  functions are equal up to the factor $\sqrt{4\pi e^{-\gamma_E}}$. They are not unity for $\alpha=0$. }
\label{linear_4d_functions_fig}
\end{figure}

\subsection[Counter term and epsilon-dependence]{Counter term and $\epsilon$-dependence}
Our calculation (up to order $\alpha^{25}$) delivers  for the counter term  in MS the series coefficients
\begin{align*}
	\ln \left(  \hat Z(\alpha,\epsilon) \right)  &= -\frac{1}{\epsilon} \left( \alpha + \frac{\alpha^2}{2} + 2 \frac{\alpha^3}{3} + 5 \frac{\alpha^4}{4} + 14 \frac{\alpha^5}{5} + 42 \frac{\alpha^6}{6}+\ldots \right) .
\end{align*}
Once more we recognize the Catalan numbers and introduce $\gamma = \gamma(\alpha)$ from \cref{linear_4d_Gren_MOM_theory},
\begin{align}\label{linear_4d_Zhat}
	\hat Z(\alpha,\epsilon) &= \exp \left( -\frac 1 \epsilon \sum_{m=1}^\infty C_{n-1}\frac{\alpha^n}{n} \right)   = e^{-\frac 1 \epsilon \left( 1-\sqrt{1-4\alpha} + \ln \left( 1-\frac{1-\sqrt{1-4\alpha}}{2} \right)   \right) }   = e^{ \frac 1 \epsilon \left(  2\gamma  - \ln (1+\gamma ) \right)  }  .
\end{align}
This expression is the integral of $\gamma(\alpha)$, as expected from \cref{Z_beta_gamma} for a linear DSE. 
As long as $\alpha$ and $\epsilon$ have the same sign, this function has the limit $\hat Z(\alpha, 0^+) = 0$ when $\epsilon \rightarrow 0$, see \cref{linear_4d_Z_fig}. With \cref{linear_4d_Zhat}, the counter term turns out to be a remarkably well-behaved function of $\epsilon$, compared to its perturbative expansion, where every single term diverges as $\epsilon \rightarrow 0$. 
This is in line with \cite{mack_conformalinvariant_1973} and a comment made in \cite{kreimer_etude_2008}: The all-order-solution \cref{linear_4d_Gren_MOM_theory} \enquote{regulates itself} by its anomalous dimension. The integral in the DSE \cref{dse_linear_ms} is not divergent and the remaining finite counter term  is set to zero in MOM by choice of the renormalization point.  \Cref{linear_4d_Z_fig} shows how $\hat Z$ approaches zero as the scaling solution $\epsilon=0$ is reached.

Using the expansion \cref{linear_Z_expansion}, it has been verified to order $\alpha^{20}$ that $z_{-1}(\alpha) = \bar z_{-1}(\alpha)$ is given by $\gamma(\alpha)=\bar \gamma(\alpha)$ and that the MOM-coefficient $z_0(\alpha)$ fulfils $z_0(\alpha) = \ln \bar \gamma_0(\alpha)$ as expected. 
In our simplified MOM-scheme \cref{linear_4d_recursion_MOM}, all other $z_{n>0}$ vanish. In true kinematic renormalization, they are present. The author computed the coefficients $z_{n\leq 9}$ up to order   $\alpha^{10}$  but did not succeed in finding generating functions. However it turned out that all $z_{n\leq 9}(\alpha)$ for $0<\alpha< \frac 1 4$ are, within the computed order, strictly positive. This implies that the exponent is strictly negative for all $\epsilon>0$ and hence $Z(\alpha,\epsilon)\in [0,1]$. The classical interpretation of the $Z$-factor as a probability requires these bounds. Compare \cite[Sec 8]{lutz} for the various interpretations of $Z$ and their relations.

\begin{figure}[htbp]
	\centering
	\begin{subfigure}{.48\textwidth}
		\includegraphics[width=\linewidth]{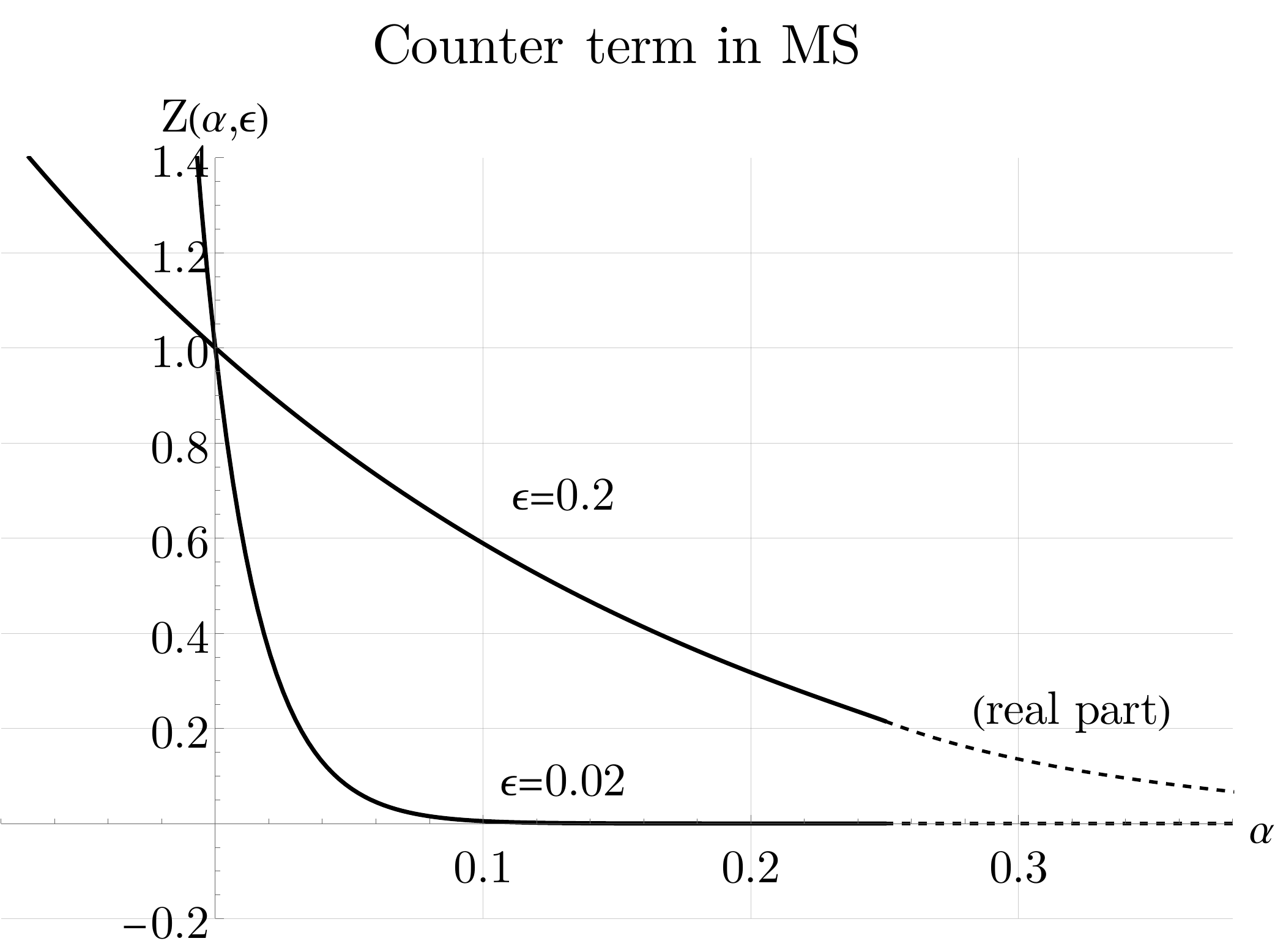}
		\caption{}
	\end{subfigure}
	\hfill 
	\begin{subfigure}{.48\textwidth}
		\includegraphics[width=\linewidth]{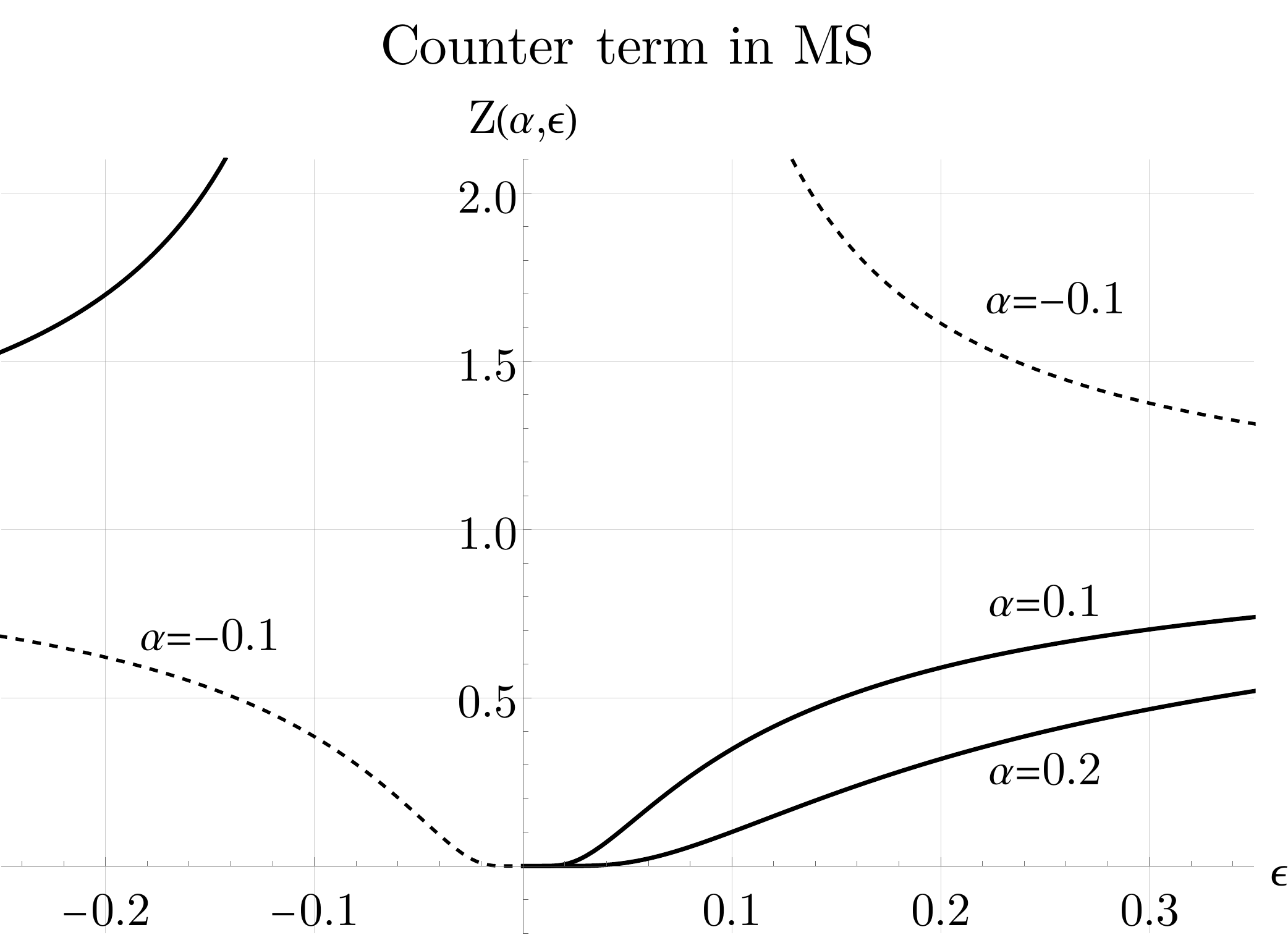}
		\caption{}
	\end{subfigure}
	\caption{(A) Counter term $\hat Z(\alpha,\epsilon)$ in MS as a function of the renormalized coupling $\alpha$ for different values of $\epsilon$. As $\epsilon \rightarrow 0^+$, the function approaches zero for   $\alpha>0$ and diverges for   $\alpha<0$. For $\alpha >\frac 1 4$, the counter term acquires an imaginary part which is not shown. (B) The same counter term as a function of $\epsilon$ for fixed values of $\alpha$. For positive $\alpha$, the counter term smoothly approaches the value zero as $\epsilon \rightarrow 0^+$.  }
	
	\label{linear_4d_Z_fig}
\end{figure}

\section[Linear DSE in D=6 dimensions]{Linear DSE in $D=6-2\epsilon$ dimensions}\label{D6}

For the 6-dimensional case, the procedure is completely analogous to the one described in the previous section. We will use the same symbols as in the $D=4$ case in order to not clutter notation.

The Dyson-Schwinger equation is once more \cref{dse_linear_ms}, only now with $D=6-2\epsilon$.
This time, we define $\alpha := \lambda (4\pi)^{-3}$ to account for the additional factor $4\pi$ produced by the integration. Moreover, the 1-loop-integral is now proportional to $q^2$ but this factor is absorbed by projection onto the basis tensor  such that again the tree level solution is $G^{(0)} (q^2)=1$. Then, the coefficients of the renormalized Green function and the counter term are given by the same recursion relations as above, namely \cref{linear_4d_recursion_MSbar,linear_4d_recursion_g0_MSbar,linear_4d_recursion_MOM}. The crucial difference is that for $f^{(k)}_n$ one now takes the value of the 6-dimensional primitive integral, as given in \cref{linear_6d_fkn}. 

Once more, the anomalous dimension can be computed analytically from the Mellin transform of the 1-loop integral \cref{mellin}. The Green function in kinematic renormalization is \cite{delbourgo_dimensional_1997}
\begin{align*}
	G (\bar x) &=  \bar x  ^{\gamma(\alpha)} , \qquad \gamma(\alpha) = \frac{\sqrt{5+4\sqrt{1+\alpha}} -3}{2}.
\end{align*}

The functions  $\bar \gamma_k(\alpha)$ are again computed from \cref{linear_4d_Gren_MSbar_series} using the appropriate $\bar g^{(m)}_{t,s}$. Like above, all $\bar \gamma_k(\alpha)$ are proportional to the corresponding $\gamma_k(\alpha)$  at least up to $\alpha^{25}$ and $k=20$. 

The series coefficients of $\bar \gamma_0(\alpha)$ can no longer be identified from tables right away but the result \cref{linear_4d_gammabar0}, expressed in terms of $\gamma(\alpha)$, is a helpful starting point. Eventually, one arrives at
\begin{align}\label{linear_6d_gammabar0}
	\bar \gamma_0(\alpha) &=   3\sqrt 3\frac{ e^{\frac 12 \left( \sqrt{5+4\sqrt{1+\alpha}}-3 \right)(1-2\gamma_E)  } \left( \sqrt{5+4\sqrt{1+\alpha}}-3 \right) \Gamma \left( \frac 5 2 -\frac 12 \sqrt{5+4\sqrt{1+\alpha}} \right)   }{\alpha \left( (1+\alpha) \left(5+4\sqrt{1+\alpha} \right) \right)^{\frac 1 4} \Gamma \left( -\frac 12 + \frac 12 \sqrt{5+4\sqrt{1+\alpha}} \right)  }  \\
	&=  \frac{6 \gamma}{\alpha} \sqrt{\frac{\d \; (6 \gamma)}{\d \alpha}}  e^{\gamma(1-2\gamma_E)  }  \frac{     \Gamma \left( 1-\gamma \right)   }{ \  \Gamma \left( 1+\gamma \right)  } . \nonumber 
\end{align}
This has been verified symbolically to order $\alpha^{25}$. Knowing $\bar \gamma_0(\alpha)$,   the shifts between MS-bar- resp. MS-  and MOM-renormalization are $\bar \delta(\alpha)=\bar \gamma_0^{\frac{1}{\gamma}}$ and $\hat \delta(\alpha)  = (4\pi e^{-\gamma})^{-\frac 12} \bar \delta(\alpha)$. 

The counter term in MS for the 6-dimensional theory is
\begin{align*}
	\hat Z(\alpha,\epsilon) &= e^{\frac 1 \epsilon \left( 2 \sqrt{5+4\sqrt{1+\alpha}} -6 -\frac{3}{2}\ln\alpha  + \ln (108) + \frac 12 \ln \frac{\sqrt{5 + 4\sqrt{1+\alpha}} -1}{\sqrt{5+4\sqrt{1+\alpha}} +1} + \frac 3 2 \ln \frac{  \sqrt{5+4\sqrt{1+\alpha}}  -3}{   \sqrt{5+4\sqrt{1+\alpha}} +3}\right)}.
\end{align*}
This function fulfills \cref{Z_beta_gamma}. Furthermore, it has been checked to order $\alpha^{25}$ resp. $\alpha^{10}$ that $Z(\alpha,\epsilon)$ reproduces $\bar \gamma_0$ via \cref{linear_Z_expansion} for the pseudo-MOM and the true MOM scheme, respectively.

\section{Kreimer's linear toy model}\label{toy}
Both Dyson-Schwinger equations considered above were based on the same 1-loop primitive Feynman graph as an integral kernel. Our formalism is not restricted to that particular integral, for comparison we here examine the linear Dyson-Schwinger equation in a  toy model of renormalization proposed by Dirk Kreimer \cite{panzer_hopf_2012}.  It reads
\begin{align}\label{toy_dse}
	\hat G  (\alpha, x) &= 1+\left( 1-\hat \renop \right)  \alpha \int \limits_0^\infty \frac{\d y \; (xy)^{-\epsilon}}{1+y} \hat G (\alpha,xy),
\end{align}
where again $\epsilon$ is a regularization parameter. 
There is no distinction between MS- and MS-bar schemes in the toy model. We retreat to a comparison between MS and pseudo-MOM.

The anomalous dimension, computed from the Mellin transform \cref{mellin_gamma1,mellin}, is
\begin{align*}
	\gamma(\alpha) &= -\frac 1 \pi \arcsin \left( \pi \alpha \right)  = -\alpha-\frac{\pi^2 }{6}\alpha^3- \frac{3 \pi^4}{40} \alpha^5-\frac{5\pi^6 }{112}\alpha^7-\ldots
\end{align*}
It has been verified to order $\alpha^{30}$ that the symbolic calculation of $G(\alpha, x)$ in pseudo-MOM renormalization produces, in the limit $\epsilon \rightarrow 0$, the expansion coefficients of $x^{\gamma(\alpha)}$ from \cref{scaling_solution}.

For MS, using \cite[A034255]{oeis}, the scaling factor is
\begin{align}\label{toy_gamma0}
	  \hat \gamma_0(\alpha) &= 1+ \sum_{m=1}^\infty \alpha^m \sum_{t=1}^m \bar g^{(m)}_{t,0} = 1 +  \left(\frac{\pi^2 \alpha^2}{4}\right)  + \frac 5 2 \left( \frac{\pi^2 \alpha^2}{4} \right) ^2  +  \ldots=   \left( 1-\pi^2 \alpha^2 \right) ^{-\frac 1 4} = \sqrt{\frac{\d \; \gamma}{\d \alpha}}.
\end{align}
From this, the change of the renormalization point can be computed using \cref{linear_delta}. One finds 
\begin{align*}
	\ln \hat \delta (\alpha) &= \frac{\pi\ln \left( 1-\alpha^2 \pi^2 \right) }{4\arcsin(\alpha \pi)}=-\frac{\pi^2}{4}\alpha -\frac{\pi^4}{12}\alpha^3 -\frac{73 \pi^6}{1440}\alpha^5 -\mathcal O(\alpha^7),
\end{align*}
where the constant coefficient vanishes since for the toy model $f_0 = f^{(0)}_0=0$.

The series coefficients of the toy model are somewhat easier than for the physical theory. This entails that  in the expansion  
\begin{align}\label{toymodel_Z}
	Z(\alpha,\epsilon) &\equiv \exp \left( -\sum_{n=-1}^\infty z_n (\alpha) \cdot \epsilon^n \right)  \equiv \exp \left( \sum_{n=1}^\infty z'_n(\epsilon) \cdot \alpha^n \right) .
\end{align}
the first functions $z_n(\alpha)$ and $z'_n(\alpha)$ can be found symbolically. Some of them are quoted in \cref{app_toymodel} for the interested reader. They suggest that $Z(\alpha,\epsilon)\in [0,1]$.  As expected from \cref{linear_Z_expansion}, also the linear toy model fulfils 
\begin{align*}
	z_{-1} &= -	\left[ \epsilon^{-1} \right] \ln Z =- \int \limits_0^\alpha \frac{\d u}{u}\gamma(u)  ,\qquad 
	z_0  = -	\left[ \epsilon^0 \right] \ln Z = -\frac 1 4 \ln \left( 1-\alpha^2\pi^2 \right) = \ln \bar \gamma_0(\alpha).
\end{align*}

The upshot from the three linear Dyson-Schwinger equations is that the series coefficients of the shift  $\delta(\alpha)$ between MS and MOM are sufficiently tame that one can recognize the functional form. These functions are convergent power series for small couplings $\alpha$.

\section{The Chain approximation in D=4}\label{chains}
As an example of a situation where MS and MOM can not be related by a shift $\delta(\alpha)$, we consider the chain approximation. In it, the only allowed graphs consist of a chain of one-loop subgraphs inserted into one single primitive, but no further nestings. 
It is sometimes viewed as an intermediate step between the linear  and the full recursive DSE, see for example \cite{borinsky_nonperturbative_2020}. We restrict here to the $D=4, \epsilon=0$ case. The first function of the log-expansion in MOM is known to be
\begin{align}
\gamma_1(\alpha) &= -\sum_{n=1}^\infty (n-1)!\alpha^n = e^{-\frac 1 \alpha} \int_{-\frac 1 \alpha}^\infty \frac{\d t}{t}e^{-t},
\end{align}
where the resummed series is the incomplete Euler gamma function. Explicit computation of the higher $\gamma_t (\alpha)$ in MOM produces coefficients which can again be identified, 
\begin{align}\label{chain_rge}
\gamma_{t \geq 1} &= (-1)^t \frac{1}{t!} \sum_{n=t}^\infty (n-1)! \alpha^n, \qquad 
k \gamma_k(\alpha) = -\alpha \cdot \alpha \partial_\alpha \gamma_{k-1}(\alpha) \nonumber \\
\Rightarrow \qquad x \partial_x G(\alpha,x) &= \gamma_1(\alpha)+ (-\alpha) \alpha \partial_\alpha G(\alpha,x).
\end{align}
Although the last equation reminds us of Callan-Symanzik equation \cref{rge} for a beta function $\beta(\alpha)=-\alpha$, it is structurally different. The function $\gamma_1(\alpha)$ is not the anomalous dimension of this model in the conventional physical sense because it does not multiply $G(\alpha,x)$. 

In Minimal Subtraction we find
\begin{align*}
	\bar\gamma_0(\alpha) &= 1 + 2 a + \frac {11}2 \alpha^2 + \left(\frac{37}3+ \frac 2 3 \zeta(3)\right) a^3 + \left(\frac{169}4 - \frac{1}{120}\pi^4 +\frac 12 \zeta(3)\right) a^4+ \ldots =: \sum_{r=0}^\infty r_k \alpha^k.
\end{align*}
The coefficients grow approximately $r_k \sim (k-1)!$. The higher expansion functions $\bar \gamma_{t>0}(\alpha)$ are purely rational. Empirically, the coefficients agree with [A010842] \cite{oeis}, $	c_n = (n-1)![x^{n-1}] \frac{e^{2x}}{x-1}$,
\begin{align}
	\bar\gamma_1(\alpha)&= -\alpha - 3 \alpha^2 - 10 \alpha^3 - 38\alpha^4 - 168 \alpha^5 - 872 \alpha^6 - 5296 \alpha^7 - 
	37200 \alpha^8  -\ldots = \sum_{n=1}^\infty c_n \alpha^n.
\end{align}

The higher $\bar \gamma_j$, but not $\bar \gamma_0$, satisfy the recursion $	k \bar \gamma_k(\alpha) = -\alpha \cdot \alpha \partial_\alpha \bar \gamma_{k-1}(\alpha).$ 
As remarked below \cref{thm_shifted}, it is possible for $\delta(\alpha)$ to exist even if both $\left \lbrace \gamma_j \right \rbrace $ and $\left \lbrace \bar \gamma_j \right \rbrace $ do not fulfil a Callan-Symanzik equation. But in the present case, explicit calculation using \cref{Gren_gammaks} shows that no $\delta(\alpha)$ exists.  The fact that MOM and MS are not related by any $\delta(\alpha)$ means that the chain-approximation is unphysical in the sense that, for different renormalization schemes, it gives rise to measurably different Green functions.

\section[Non-linear DSE in D=4]{Non-linear DSE in $D=4$}\label{nonlinear_4d}

In the remainder of the paper we repeat the above analysis of the two physical models and the toy model for the case that the Dyson-Schwinger equation is non-linear. Namely, instead of $Q\equiv 1$ we insert the invariant charge \cref{Qs},
\begin{align}\label{Qs2}
	Q(G(\alpha, x)) = \left(G(\alpha,x)\right)^{s},
\end{align}
where $s\in \left \lbrace -5,\ldots,+5 \right \rbrace $. 
It seems that the literature so far has mostly concentrated on the physically most relevant cases $s=0$ (linear approximation, e.g. \cite{delbourgo_dimensional_1996,delbourgo_dimensional_1997,kreimer_etude_2008}) and $s=-2$ (one inverse Green function inserted into the kernel, e.g. \cite{broadhurst_combinatoric_2000,broadhurst_exact_2001,bellon_renormalization_2008,bellon_efficient_2010,bellon_approximate_2010,borinsky_nonperturbative_2020,borinsky_semiclassical_2021}). 
In our case, the corresponding power of $G$ is inserted into \emph{only one} edge of the primitive. The setup discussed in \cite{bellon_renormalization_2008,bellon_efficient_2010,bellon_approximate_2010} is conceptually different from our $s=-3$ since it amounts to inserting one $G^{-1}$ into each of the two internal edges. Compare our result \cref{nonlinear_4d_gamma1_sm3} with \cite[Table 1]{bellon_renormalization_2008}.  Also see \cite[Sec. 5]{yeats_growth_2008} for a discussion how insertion into only a subset of the available edges is equivalent to including additional primitive kernels. 

\subsection{Computation of the coefficients}\label{nonlinear_computation}

The MS-renormalized Dyson-Schwinger equation for the $D=4$ model reads
\begin{align}\label{nonlinear_4d_dse}
	\hat G(\alpha, q^2/\mu^2) &= 1 +  \alpha (4\pi)^2 \left( 1-\hatrenop \right)   \int \frac{\d^D k}{(2\pi)^D} \frac{\left( \hat G (\alpha, k^2/\mu^2) \right) ^{s+1}}{(k+q)^2 k^2}.
\end{align}
Note that as above we choose a  sign $+\alpha$ in front of the integral. 
For $s=-2$, this is the model examined in \cite{broadhurst_combinatoric_2000,broadhurst_exact_2001,borinsky_nonperturbative_2020}, up to a factor -2 in the definition of $\alpha$, see \cite[eq. (12)]{broadhurst_exact_2001}.

For a recursive computation of the coefficients, we once more introduce $\hat x, \bar x$ according to \cref{MS_MSbar} and thereby switch from MS- to MS-bar-renormalization.
The first order solution coincides with the one of the linear DSE, \cref{linear_4d_gren1_2},
\begin{align*}
	\bar G^{(1)} (\alpha, \bar x) &= 1+ \alpha \sum_{t=0}^1 \bar x^{-t\epsilon} \sum_{w=-1}^\infty \bar g^{(1)}_{t,w} \epsilon^w.
\end{align*}
The index $^{(m)}$ in the linear case counts both the coradical degree (= number of recursive iterations of the solution) and the order in $\alpha$ (= loop number of the involved graphs). In the non-linear DSE, we truncate the series expansion of $ (\hat G (\alpha,\bar x))^{s+1}$ at order $\alpha^{m-1}$ so that $\hat G ^{(m)}(\alpha,\bar x)$ again involves graphs with at most $m$ loops. This choice is arbitrary but convenient because it saves one index.

The recursion formula for the next order is more complicated than in the linear case \cref{linear_computation}. This is because the non-linear DSE \cref{nonlinear_4d_dse} involves a non-trivial power of the Green function inside the integral which needs to be expanded both in $\alpha$ and in $\bar x^{-\epsilon}$.  Assume we know the order-$m$-solution in the form
\begin{align*}
	\bar G^{(m)} (\alpha, \bar x) &= 1+\sum_{n=1}^m \alpha^n \sum_{u=0}^n \bar x^{-u\epsilon} \sum_{s=-n}^\infty \bar g^{(n)}_{u,w} \epsilon^w =: 1+ \sum_{n=1}^m \alpha^n \bar G_n (\bar x, \epsilon),
\end{align*}
where we defined functions $\bar G_n(\bar x, \epsilon)$. They are universal for all $m$. Next, we write a generic expansion of the invariant charge \cref{Qs2}  according to 
\begin{align}\label{nonlinear_4d_Q_expansion} 
	\bar G^{(m)}(\alpha,\bar x) \cdot Q(\bar G^{(m)}(\alpha, \bar x)) &\equiv \left(\bar G^{(m)} (\alpha, \bar x)\right)^{s+1}  =: 1 + \sum_{n=1}^m \alpha^n \sum_{t=0}^n \bar x^{-t\epsilon} \bar h^{(n)}_t(\epsilon) .
\end{align}
The helper functions $\bar h^{(n)}_t(\epsilon)$ are Laurent series in $\epsilon$ with the highest pole order $\epsilon^{-n}$. They are given by Faa di Bruno's formula and the Binomial theorem,   $B_{n,k}$ are  Bell polynomials\cite{bell_exponential_1934} \cite[p 134]{comtet_advanced_1974}:
\begin{align}\label{faadibruno}
	\frac{1}{\left( \bar G^{(m)}(\bar x) \right) ^{-s-1}} &= \frac{1}{(-s-2)!}\sum_{n=0}^\infty \alpha^n \frac{1}{n!} \sum_{k=1}^n  (-s-2+k)!  B_{n,k} \left( 1! \bar G _1, 2!\bar G_2, \ldots \right), \quad  s <-1 \nonumber \\
	\left( \bar G^{(m)}(\bar  x) \right) ^{s+1} &= (s+1)! \sum_{n=0}^{\infty } \alpha ^n \frac 1 {n!}  \sum_{k=0}^{s+1}  \frac{1}{(s+1-k)!}     B_{n,k} \left( 1! \bar G_1, 2! \bar G_2, \ldots  \right) , \quad s>-1 .
\end{align}
Knowing the functions $\bar h^{(n)}_t(\epsilon)$, one integrates the sum \cref{nonlinear_4d_Q_expansion} term-wise like in the linear case \cref{linear_computation} and obtains the next-order coefficients 
\begin{align*}
	\bar g^{(1)}_{1,w} &= f^{(0)}_w, \qquad \bar g^{(n)}_{u,w} = \sum_{r=-1}^{n+w-1} \left( [\epsilon^{w-r}] \bar h^{(n-1)}_{u-1} \right)   f^{(u-1)}_r.
\end{align*}
Finally, one obtains the next order solution of the DSE,
\begin{align}\label{nonlinear_G_coefficients}
	\bar G^{(m+1)} ( \alpha, \bar x) &= \bar G^{(1)} (\alpha , \bar x) + (1-\hat \renop) \sum_{n=2}^{m+1} \alpha^n \sum_{u=1}^n \bar x^{-u \epsilon} \sum_{w=-n}^\infty \bar g^{(n)}_{u,w} \epsilon^w.
\end{align}
The  MS-counter term is included via coefficients $g^{(n)}_{0,s}$ as in the linear case \cref{linear_4d_recursion_g0_MSbar}. For the non-linear DSE, there is no simple pseudo-MOM scheme. In practice, it is sufficient to include terms $\propto \epsilon^m$ if one is interested in the finite part of the Green function $G ^{(m)}(\alpha, \bar x)$ since every iteration potentially multiplies the result with $\epsilon^{-1}$.

Of course, the established methods  \cite{broadhurst_combinatoric_2000,kreimer_etude_2006,bellon_efficient_2010} are tremendously more efficient in computing the anomalous dimension in MOM. A power-series solution of the ODE  \cref{nonlinear_4d_rge} to order $\alpha^{100}$ can be obtained within seconds while the brute-force algorithm merely reaches $\alpha^{10}$ symbolically after several hours. But the computation of $\gamma(\alpha)$ is only a side effect of our  algorithm since we are actually interested in  $\ln \delta(\alpha)$. 

All computations, also the extraction of series coefficients in the computation of $h^{(n)}_t$ and $g^{(n)}_{u,w}$, have been done with the computer algebra system Mathematica 12.3. 
In practice, the computation is entirely limited by CPU time due to an explosion of series coefficients: In order to reach $\epsilon^0$ at order $\alpha^{10}$, we have to include terms up to $\epsilon^{10}$ in the intermediate steps. Further we produce pole terms up to $\epsilon^{-10}$ and contributions up to $x^{-10\epsilon}$,   each $g^{(10)}_{t,s}$ requires series reversion and -multiplication. If we were to go to $\alpha^{20}$, we would have to include $\epsilon^{20}$ from the start, dramatically slowing down every intermediate step. 

The higher the order, the higher powers of $\pi^2$ and the more different zeta values appear. Algebraic operations with these expressions are increasingly slow.
This second problem can be circumvented by working with floating point numbers, but it turns out that each iteration loses several decimal digits of precision. We computed the first orders symbolically and then continued numerically. 

Thirdly, the expansions \cref{faadibruno} are, for large $n$, much harder for negative $s$ than for small positive $s$ due to the summation boundaries.  Therefore we  reach higher order   for the positive $s$.

\subsection{Results}\label{nonlinear_4d_results}

The coefficients were computed symbolically up to $\alpha^{10}$ for an invariant charge \cref{Qs2} where $s \in \left \lbrace -5,  \ldots, +5 \right \rbrace $. The results are extended up to at least $\alpha^{20}$ numerically with at least 30 valid decimal digits.
It was verified in all cases that the first three orders of the leading-log expansion fulfil \cref{leadinglog_first} and that \cref{rge_MS} holds (for some function $\bar \gamma(\alpha)$).  

The anomalous dimensions up to order $\alpha^8$ are reported in \cref{d4_s_gamma1} in \cref{tables}.
Let the anomalous dimension be $\gamma(\alpha) = \sum_{n=0}^\infty c_n \alpha^n$ then the empirical values of \cref{d4_s_gamma1} suggest $c_0=0, c_1=-1, c_2=-(s+1), c_3=-(1+s)(2+3s), c_4 = -(s+1)(2s+1)(7s+5)$. 
Further, for $s=1$ the sequence is \cite[A177384]{oeis}. The case $s=-3$ produces
\begin{align}\label{nonlinear_4d_gamma1_sm3}
	\gamma(\alpha)=-\alpha + 2 \alpha^2 - 14 \alpha^3 + 160 \alpha^4 - 2444 \alpha^5 + 45792 \alpha^6 - 1005480 \alpha^7 + 		25169760 \alpha^8\mp \ldots.
\end{align}
Compare this to  \cite[Table 1]{bellon_renormalization_2008}, in which $G(\alpha,x)$ is inserted into both the internal edges of the primitive. Our result \cref{nonlinear_4d_gamma1_sm3} reproduces the purely rational part of the latter, but not the terms proportional to $\zeta(j)$. 

The anomalous dimension considered so far,  $\gamma(\alpha) =:\gamma^{\text{pert}}(\alpha)$, is the perturbative solution to the differential equation \ref{mellin_gamma1},
\begin{align}\label{nonlinear_4d_rge}
	-\left(1 + \gamma(\alpha) \left(  s \alpha \partial_\alpha  +1 \right) \right)\gamma(\alpha)  &=\alpha.
\end{align} 
This ODE has also non-perturbative solutions \cite{borinsky_semiclassical_2021,borinsky_nonperturbative_2020} of the form
\begin{align}\label{nonlinear_4d_nonpert}
	\gamma^{\text{non-pert}}(\alpha) &= \alpha^\beta \exp \left(-\frac{\lambda}{\alpha}\right) \left( 1+ b^{(1)} \alpha + b^{(2)} \alpha^2 + \ldots \right) . 
\end{align}
We use the method of \cite[V. A.]{borinsky_semiclassical_2021} to determine the unknown coefficients\footnote{The author thanks Gerald Dunne for suggesting the method.}. The ansatz $\gamma (\alpha) = \gamma ^{\text{pert}} (\alpha)+ \gamma ^{\text{non-pert}}(\alpha)$ is inserted into \cref{nonlinear_4d_rge} and the above coefficients $c_j$ are used for $\gamma^{\text{pert}}$. The  equation is  then linearized in $\gamma^{\text{non-pert}}$. The resulting series in $\alpha$  has to vanish, this leads to the expressions listed in \cref{nonlinear_4d_nonpert_parameters}, especially $\lambda(s)=1/s, \beta(s) = -(3+2s)/s$.
For $s=-2$ we reproduce \cite[(25)]{borinsky_nonperturbative_2020} up to different sign conventions regarding $\alpha$, mentioned below \cref{nonlinear_4d_dse}.

The  coefficients $c_n$ of the perturbative solution of the non-linear DSE  \cref{nonlinear_4d_rge}, grow factorially, which has been studied repeatedly \cite{borinsky_generating_2018,borinsky_nonperturbative_2020,broadhurst_exact_2001,bellon_approximate_2010}. The asymptotic behaviour of $c_n$ is dictated by the non-perturbative solution \cref{nonlinear_4d_nonpert}, \cite{aniceto_primer_2019} (alternatively, use the methods of \cite{bellon_approximate_2010}),  namely for $n \rightarrow \infty$
\begin{align}\label{nonlinear_4d_asymptotic}
	c_n \sim S(s)\cdot \frac 1 {\lambda(s) ^n} \cdot \Gamma \left( n-\beta(s) \right) \left( 1 + \frac { \lambda(s) \cdot  b^{(1)}(s)} {(n-\beta(s) -1)} + \frac{\lambda(s)^2\cdot  b^{(2)}(s)}{(n-\beta(s)-1)(n-\beta(s)-2)}+ \ldots  \right)   .
\end{align}
We computed 500 series coefficients of $\gamma^{\text{pert}}(\alpha)$ and extracted their asymptotic behaviour using order-70 Richardson extrapolation. This produced at least 50 significant digits and confirmed the expressions $\lambda(s), \beta(s), b^{(1)}(s) ,b^{(2)}(s)$ and $b^{(3)}(s)$ listed in \cref{nonlinear_4d_nonpert_parameters}. The Stokes constant $S(s)$ is reported in \cref{d4_s_stokes} in \cref{tables}. One recognizes \cite{borinsky_nonperturbative_2020}, $S(-2) = 1/(\sqrt \pi e)$ and also $S(-3) = 3/(\pi e^2)$, all other Stokes constants appear unfamiliar\footnote{The Stokes constant $S(s)$ was also computed for non-integer $s$. It appears to be a fairly smooth function of $s$, with zeros, as expected, at the points $s=-1$ and $s=0$. Remarkably, inside the interval $(-1,0)$ the   function is oscillating  and has additional zeros, accumulating near $s=0$. This could be worth further study. }.

To visualize the asymptotic behaviour, we  consider the ratio  
\begin{align}\label{other_cn_ratio}
	\frac{c_{n+1}/\Gamma(n+1-\beta(s))}{c_n /\Gamma(n-\beta(s))} \equiv \frac{c_{n+1}}{(n+\frac{3+2s}{s})c_n} = s-b^{(1)}(s) \frac{1}{n^2}+\mathcal O \left( \frac 1 {n^3} \right)  \qquad (\text{for }s \neq 0,-1).
\end{align} 
There is no $1/n$ correction to this quantity, hence it converges quickly, as shown in \Cref{nonlinear_4d_fig} (A).

\begin{figure}[htbp]
	\centering
	\begin{subfigure}{.48\textwidth}
		\includegraphics[width=\linewidth]{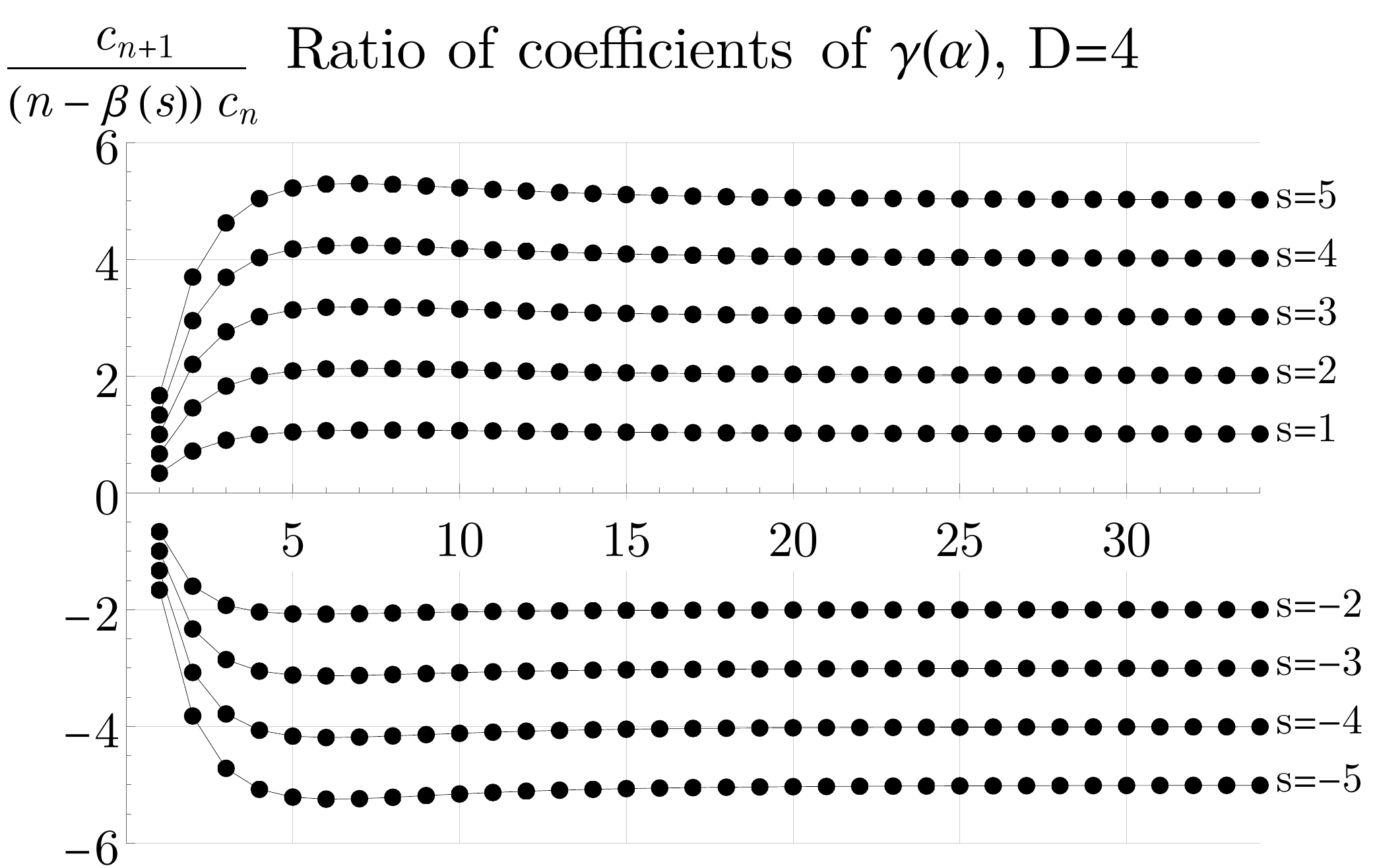}
		\caption{}
	\end{subfigure}
	\hfill 
	\begin{subfigure}{.48\textwidth}
		\includegraphics[width=\linewidth]{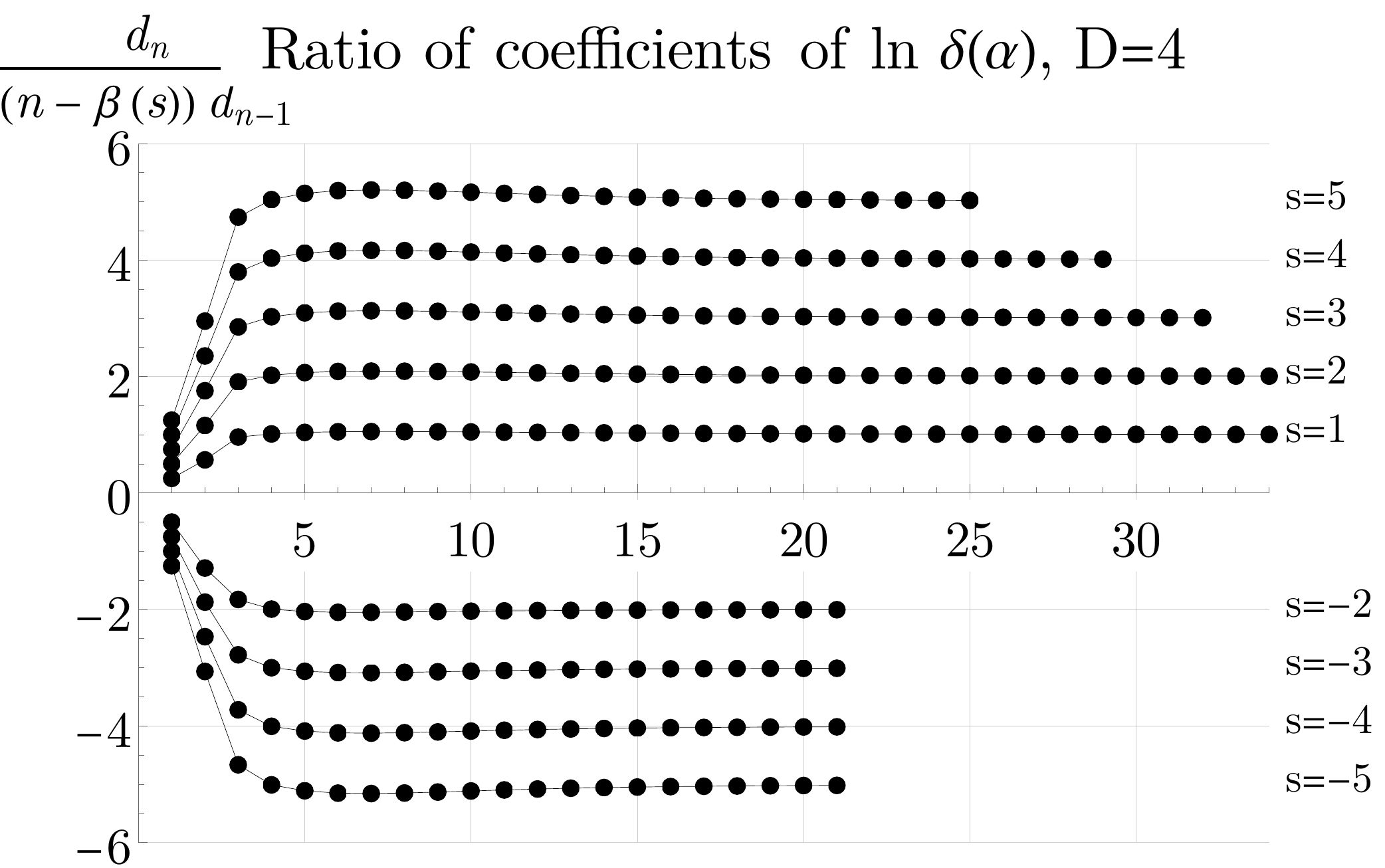}
		\caption{}
	\end{subfigure}
	\caption{(A) Ratio of successive coefficients $c_n$ of $\gamma(\alpha)=\sum  c_n  \alpha^n$   for the physical model in $D=4$ dimensions. The denominator $(n-\beta(s))$ is chosen to match the known asymptotics \cref{nonlinear_4d_asymptotic}. The ratio quickly converges towards the limit $s$, see \cref{other_cn_ratio}.
	(B) Ratio of successive coefficients $d_n$ of $\ln\delta(\alpha)=\sum_n d_n\alpha^n $. Seemingly, it converges to the same limit $s$ as the ratio in (A). The computation is much harder for negative $s$, therefore only a lower order $n$ is available.  }
	
	\label{nonlinear_4d_fig}
\end{figure}

The shift from MOM- to MS-bar-renormalization is computed as discussed in \cref{shifted}. The first coefficients are reported in \cref{d4_s_ldelta}, for example, for $s=-2$ one obtains
\begin{align*}
	\bar \gamma_0(\alpha) &= 1+2\alpha -\frac{11}{2}\alpha^2 + \frac{88}{3}\alpha^3 - \left( \frac{1781}{8} + \frac{\zeta(3)}{3}  \right) \alpha^4  + \left( \frac{42613}{20} + \frac{\pi^2}{150} -\frac{12 \zeta(3)}{5} \right) \alpha^5 - \mathcal O(\alpha^6)\\
	\ln\bar \delta(\alpha) &= -2 + \frac 3 2 \alpha -\frac{29}{6}\alpha^2 + \frac{94-\zeta(3)}{3} \alpha^3 -\left( \frac{5573}{20} + \frac{\pi^4}{150}-\frac{7\zeta(3)}{5} \right) \alpha^4 \mp \ldots. 
\end{align*}

We write $\ln\bar \delta(\alpha) = \sum  d_n \alpha^n$, where the coefficients $d_n$ were computed up to order $\alpha^{10}$ symbolically and to at least order $\alpha^{20}$ numerically. 
As expected from \cref{lndelta_1}, the constant coefficient   is $d_0=-2$ for all $s$.  Similarly to \cref{other_cn_ratio}, we examine the ratio of successive $d_n$, the result is shown in \cref{nonlinear_4d_fig} (B). We extract numerical estimates for the growth parameters in the ansatz
\begin{align}\label{dn_asymptotic_ansatz}
	d_n \sim \tilde S(s) \cdot \tilde F(s)^n \cdot \Gamma \left( n-\tilde \beta(s) \right) \left( 1+ \frac{\tilde b^{(1)}(s)}{(n-\tilde \beta(s) -1)} + \ldots  \right)  
\end{align}
by the following method: We use Richardson extrapolation\cite{richardson_ix_1911,aniceto_primer_2019} of orders 2,3,4 and 5 and take their mean as the estimation and the largest absolute difference between any of these as uncertainty. Experiments with the  coefficients $c_n$ of $\gamma(\alpha)$ show that this procedure likely overestimates the uncertainties.

\begin{table}[htbp]
	\begin{tabular}{|c|c|c|c|c|c|}
		\hline
		$s$ & $n_\text{max}$& {\qquad  $ \tilde S(s) / s$} & {\qquad $\tilde F(s)$}&{\qquad $\tilde \beta(s)$} & \qquad $\tilde b^{(1)}(s)$ \\
		\hline
		5 & 24 & $-0.02532 \pm 0.00037$ & $ 4.987 \pm 0.062$ & $-3.59 \pm 0.12$ & $-2.61 \pm 0.18$\\
		\hline
		4 & 27 & $-0.02709 \pm 0.00019$ & $ 3.993 \pm 0.036$ & $-3.74 \pm 0.08$ & $-2.79 \pm 0.11$\\
		\hline
		3 & 32 & $-0.02749 \pm 0.00011$ & $ 2.997 \pm 0.017$ & $-3.99 \pm 0.04$ & $-3.10 \pm 0.06$\\
		\hline
		2 & 38 & $-0.02272 \pm 0.00010$ & $ 1.999 \pm 0.009$ & $-4.50 \pm 0.03$ & $-3.74 \pm 0.05$\\
		\hline
		1 & 38 & $-0.00541 \pm 0.00009$ & $ 0.999 \pm 0.007$ & $-6.00 \pm 0.04$ & $-5.97 \pm 0.12$ \\
		\hline
		-2 & 21 & $ 0.2080 \pm 0.0018$ & $-1.998 \pm 0.012$ & $-1.49 \pm 0.05$ & $-0.74 \pm 0.08$\\
		\hline
		-3 & 21 & $ 0.1295 \pm 0.0014$ & $-2.995 \pm 0.026$ & $-1.99 \pm 0.07$ & $-1.10 \pm 0.11$ \\
		\hline
		-4 & 21 & $ 0.0882 \pm 0.0011$ & $-3.993 \pm 0.040$ & $-2.24 \pm 0.09$ & $-1.30 \pm 0.14$\\
		\hline
		-5 & 21 & $ 0.0655 \pm 0.0009$ & $-4.991 \pm 0.054$ & $-2.40 \pm 0.10$ & $-1.43 \pm 0.15$\\
		\hline 
	\end{tabular}
	\caption{Numerical findings of the growth parameters of $\ln\bar \delta(\alpha)$ in the $D=4$ model \cref{nonlinear_4d}, according to \cref{dn_asymptotic_ansatz}. They are consistent with   \cref{d4_s_stokes,nonlinear_4d_nonpert_parameters} in the appendix. $\ln\bar \delta(\alpha)$ was computed including order $\alpha^{n_\text{max}}$.  }
	\label{d4_s_parameters}
\end{table}

The results are reported in \cref{d4_s_parameters}. They are consistent with $\tilde S(s) = s\cdot S(s), \tilde F(s)=s$   and $\tilde \beta(s) = \beta(s)-1$  within around $1 \%$ relative uncertainty. Unlike the above analysis of $\gamma$, at this level of uncertainty our findings of \enquote{rational numbers} are to be understood as educated guesswork rather than numerical proofs. The estimates obtained for $\tilde b^{(1)}(s)$ are too imprecise to deduce a formula at this point.

It turns out to be very useful to compare the coefficients of $\ln \delta(\alpha)$ to those of  $\gamma(\alpha)$. To this end we compute the ratio $d_{n}/( s \lambda  c_{n+1})$, where $s \cdot  \lambda(s) =1$ in our case. Using \cref{dn_asymptotic_ansatz} with \cref{d4_s_parameters}, we expect  that $d_n/c_{n+1} \rightarrow 1$. As before, we use Richardson extrapolation of order 2,3,4 and 5 to determine the parameters in
\begin{align}\label{4d_ratio}
	\frac{d_{n}}{ s\lambda c_{n+1}} &\sim r(s) + r_1(s) \frac{1}{n} +r_2(s) \frac{1}{n^2}+ \ldots  , \qquad n \rightarrow\infty .
\end{align}

\begin{table}[htbp]
	\begin{tabular}{|c|c|c|c|c|c|}
		\hline
		$s$ & $r_1(s)$ & $r_2(s)$ & $r_3(s)$ & $r_4(s)$ & $r_5(s)$\\
		\hline
		5 & $1.20002 \pm 0.00012$ & $0.0003 \pm 0.0019$ & $0.005 \pm 0.031$ & $0.09 \pm 0.50$ & $1.3 \pm 7.9$\\
		\hline
		4 & $1.25000 \pm 0.00001$ & $0.0000 \pm 0.0002$ & $0.000 \pm 0.003$ & $0.01 \pm 0.05$ & $0.18 \pm 0.95$\\
		\hline
		3 & $1.33333 \pm 0.00001$ & $0.0000 \pm 0.0001$ & $0.000 \pm 0.001$ & $0.01 \pm 0.01$ & $0.00 \pm 0.02$\\
		\hline
		2 & $1.50000 \pm 0.00001$ & $0.0000 \pm 0.0001$ & $0.000 \pm 0.001$ & $0.00 \pm 0.01$ & $0.00 \pm 0.01$\\
		\hline
		1 &  $2.00000 \pm 0.00001$ & $0.0000 \pm 0.0001$ & $0.000 \pm 0.001$ & $0.00 \pm 0.01$ & $0.00 \pm 0.01$\\
		\hline
		-2 & $0.50000 \pm 0.00001$ & $0.0000 \pm 0.0001$ & $0.000 \pm 0.001$ & $0.00 \pm 0.01$ & $0.00 \pm 0.02$ \\
		\hline
		-3 & $0.66667 \pm 0.00001$ & $0.0000 \pm 0.0002$ & $0.000 \pm 0.002$ & $0.01 \pm 0.03$ & $0.07 \pm 0.39$\\
		\hline
		-4 & $0.75001 \pm 0.00004$ & $0.0001 \pm 0.0005$ & $0.001 \pm 0.007$ & $0.02 \pm 0.09$ & $0.20 \pm 1.24$\\
		\hline
		-5 & $0.80001 \pm 0.00006$ & $0.0002 \pm 0.0009$ & $0.002\pm 0.012$ & $0.03 \pm 0.17$ & $0.4 \pm 2.3$\\
		\hline 
	\end{tabular}
	\caption{Parameters of the ratio $d_n/c_{n+1}$ for $D=4$  from \cref{4d_ratio}.  $r_{\geq 2}$ is consistent with zero.}
	\label{d4_s_ratio}
\end{table}

We find $r(s)=1$ with uncertainty smaller $10^{-5}$. 
The numerical results for the corrections $r_j(s)$ are reported in \cref{d4_s_ratio}. The relative uncertainties are much smaller than for the above parameters of \cref{dn_asymptotic_ansatz}. They suggest the simple formula $r(s) =1$ and $r_1(s) = (s+1)/s$. Inserting this, surprisingly, the higher order corrections $r_{j\geq 2}(s)$ seem to vanish. 

\begin{figure}[htbp]
	\centering
	\centering
	\begin{subfigure}{.48\textwidth}
		\includegraphics[width=\linewidth]{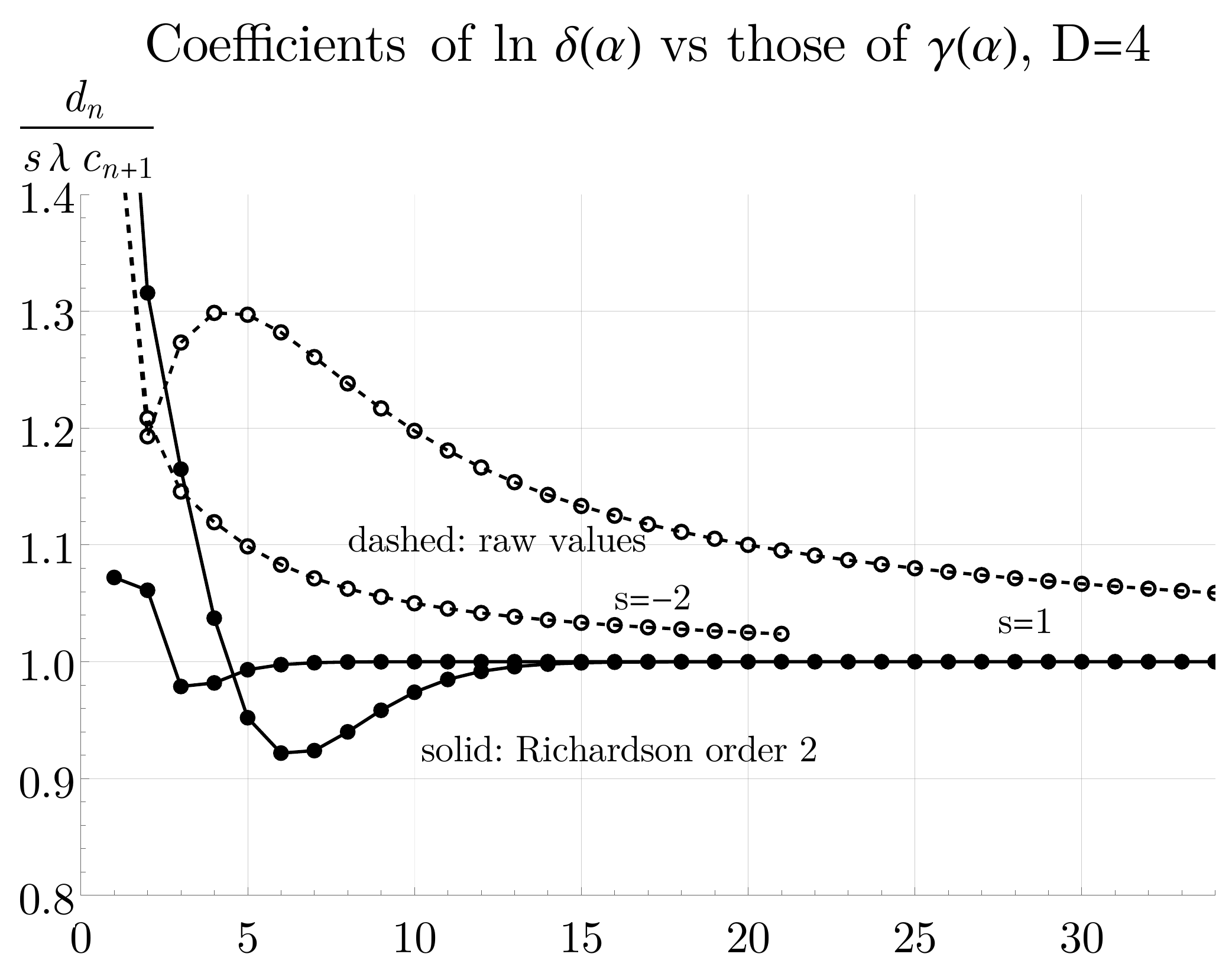}
		\caption{}
	\end{subfigure}
	\hfill 
	\begin{subfigure}{.48\textwidth}
		\includegraphics[width=\linewidth]{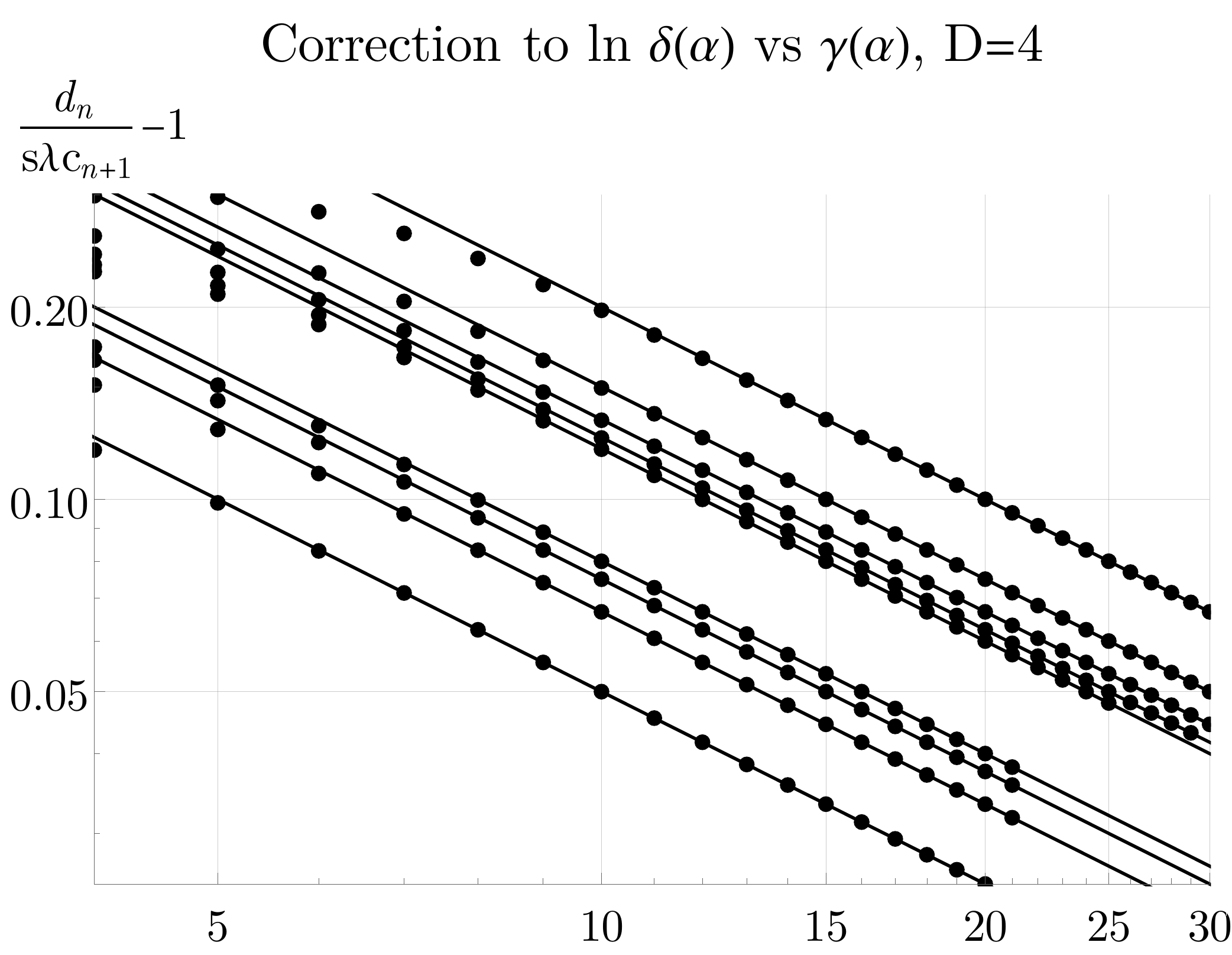}
		\caption{}
	\end{subfigure}
	\caption{(A) Ratio of the coefficients of $\ln\delta(\alpha) = \sum d_n \alpha^n$ and $\gamma(\alpha) = \sum c_n \alpha^n$. Shown are  two representative sequences, namely $s=-2$ and $s=1$. For each of them, the dashed line indicates the raw values while the solid line is the order-2 Richardson extrapolation. The ratio $d_n/c_{n+1}$ approaches the limit unity. (B) Correction to this ratio. The points are the values of $\frac{d_{n}}{c_{n+1}}-1$ for the different $s$.   Solid lines are the functions $\frac{s+1}{sn}$, they match the points surprisingly well even for low orders. The remaining difference falls of faster than $1/n^2$. }
	
	\label{nonlinear_4d_fig_ratio}
\end{figure}

Together with the known behaviour of $c_{n+1}$, \cref{nonlinear_4d_asymptotic} and $\beta(s) = -(3+2s)/s$, we conclude
\begin{align}\label{nonlinear_4d_dn}
	d_n \sim  S(s)  s^{n+1}  \Gamma \left( n-\beta(s)+1 \right) \left( 1 + \frac { -\frac{1+3s+2s^2}{s^2}} {n-\beta(s)} + \frac{\frac{ 1+4s+4s^2-6s^3-7 s^4}{2 s^4}}{(n-\beta(s) )(n-\beta(s)-1)}+ \mathcal O \left( \frac{1}{n^3} \right)  \right) .  
\end{align}
The subleading coefficient is consistent with the value $\tilde b^{(1)}(s)$ which we found in \cref{d4_s_parameters}.

For the higher order corrections in \cref{d4_s_ratio}, the uncertainties are increasing. If we nonetheless speculate that their vanishing is a general pattern, then we obtain  
\begin{align}\label{4d_dn_cn}
	d_n =   \left(1+\frac{s+1}{sn}\right)\cdot c_{n+1} + e_n.
\end{align} 
The numerical values suggest that the remainder $e_n$ falls off faster than geometrically, see \cref{nonlinear_4d_f_plot} (A). Using the ring of factorially divergent power series\cite{borinsky_generating_2018}, this asymptotic statement can be translated to a relation between the corresponding generating functions\footnote{The author thanks Michael Borinsky for pointing out this implication of \cref{4d_dn_cn}.}:
\begin{align}\label{4d_lndelta_gamma1}
	\ln \bar \delta(\alpha) &= \frac{\gamma (\alpha)}{\alpha} + \frac{s+1}{s}\int^\alpha \frac{\d a}{a }\frac{\gamma (a)}{a} + f(\alpha), \qquad s \neq 0.
\end{align}
\Cref{nonlinear_4d_f_plot} (B) shows the coefficients of the function $f(\alpha):=\sum_{n=1}^\infty f_n \alpha^n$, they seemingly grow geometrically, not factorially. This indicates that $f(\alpha)$ is  a convergent power series around $\alpha=0$. The growth rate is reported in \cref{d4_s_ldelta} in \cref{tables}. The coefficients of the function $f(\alpha)$ contain zeta values, which appear in $\ln\bar \delta(\alpha)$ but not in $\gamma(\alpha)$. The author has not succeeded in finding a closed formula.

A conclusion of \cref{nonlinear_4d} is that the factorial growth, or, equivalently, the leading non-perturbative contribution, of the shift function $\delta(\alpha)$ is surprisingly similar to the behaviour of the anomalous dimension $\gamma(\alpha)$.

\begin{figure}[htbp]
	\centering
	\centering
	\begin{subfigure}{.48\textwidth}
		\includegraphics[width=\linewidth]{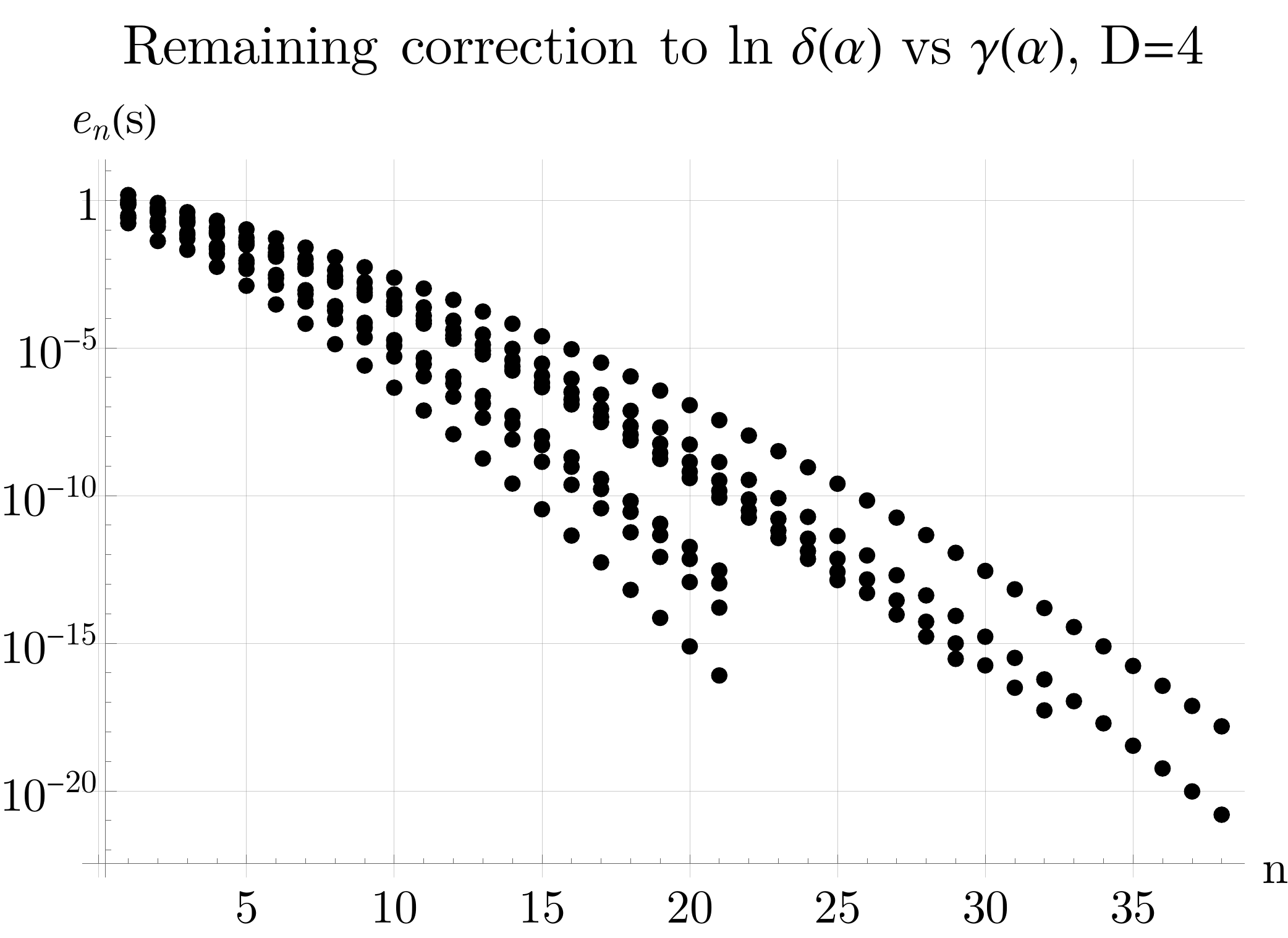}
		\caption{}
	\end{subfigure}
	\hfill 
	\begin{subfigure}{.48\textwidth}
		\includegraphics[width=\linewidth]{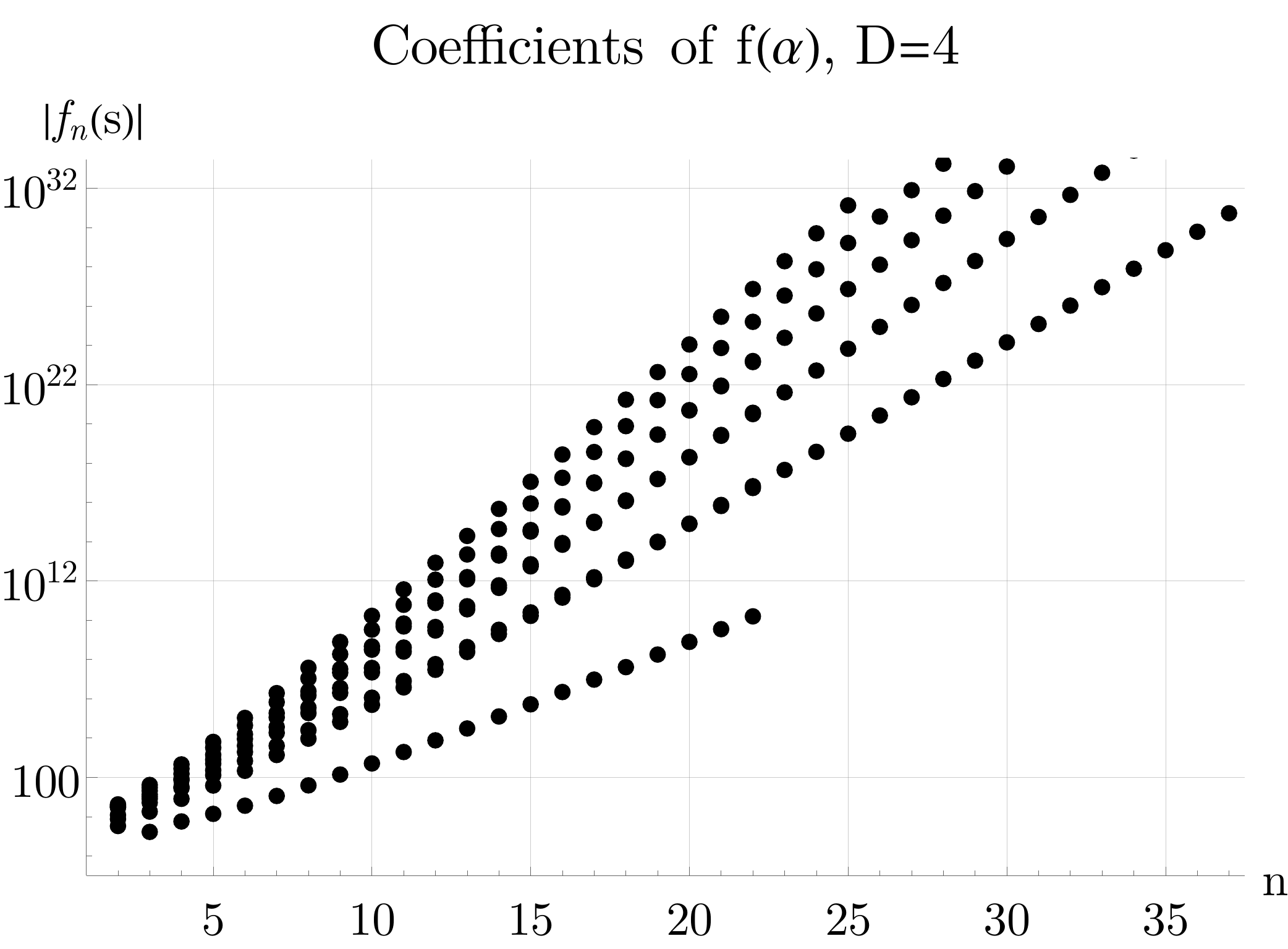}
		\caption{}
	\end{subfigure}
	\caption{(A) Remainder coefficients $e_n$ from \cref{4d_dn_cn}, for the different values of $s$. This is a logarithmic plot,  they decay faster than exponentially. (B) Coefficients of the function $f(\alpha)$ in \cref{4d_lndelta_gamma1}. The data points overlap for different $s$. The coefficients grow exponentially, which suggests that $f(\alpha)$ is an analytic function. }
	
	\label{nonlinear_4d_f_plot}
\end{figure}

\section[Non-linear DSE in D=6]{Non-linear DSE in $D=6$}\label{nonlinear_6d}

Following the procedure of \cref{nonlinear_computation}, we also evaluate the 6-dimensional model of \cref{D6} for various powers $s\in \left \lbrace -5, \ldots, +5 \right \rbrace $ in the invariant charge \cref{Qs2}. 

The leading-log functions $H_1,H_2,H_3$ are as expected from \cref{leadinglog_first}. The anomalous dimension $\gamma(\alpha)$  contains only rational coefficients for all values of $s$. Up to the computed symbolic precision $\alpha^{12}$, they fulfil  the ODE \cref{mellin_gamma1}, which here takes the form
\begin{align}\label{nonlinear_6d_rge}
	\left( 3  + \gamma \left( s \alpha \partial_\alpha  +1 \right)  \right) \left( 2 + \gamma \left(  s \alpha \partial_\alpha  +1 \right)   \right) \left( 1+\gamma \left(  s \alpha \partial_\alpha+1  \right)   \right) \gamma  &=\alpha.
\end{align}
The perturbative solution $\gamma^{\text{pert}}(\alpha) = \sum_{j=0}^\infty c_j\alpha^j$ can be computed to high order from this ODE, the first coefficients are
\begin{align*}
	c_1 = \frac{1}{6}, \quad c_2 = -\frac{11(s+1)}{216}, \quad c_3 = \frac{(s+1)(206+291s)}{7776}, \quad c_4 = -\frac{(s+1)(4711+14887s+11326 s^2)}{279936}.
\end{align*}
We use these coefficients, insert the non-perturbative ansatz \cref{nonlinear_4d_nonpert} into \cref{nonlinear_6d_rge}, linarize,  and solve for the parameters $\beta,\lambda,b^{(1)}, b^{(2)}, b^{(3)}$. \Cref{nonlinear_6d_rge} is of third order, unlike in the case $D=4$, we find three linearly independent solutions \cref{nonlinear_6d_nonpert_parameters}.
In the ansatz \cref{nonlinear_4d_nonpert}, the solution with smallest absolute $\lambda$ is dominant, this is the first entry of the vectors \cref{nonlinear_6d_nonpert_parameters}. 

With these parameters, the series coefficients $c_n$ of the perturbative solution grow according to \cref{nonlinear_4d_asymptotic}. 
We have confirmed this behaviour numerically from the first 500 coefficients $c_n$ for $s\in \left \lbrace -5,\ldots,+5 \right \rbrace $. The Stokes constant $S(s)$ is reported in \cref{d6_s_stokes}, we reproduce the value \cite[(15)]{borinsky_semiclassical_2021} for $s=-2$.
Like  \cref{other_cn_ratio}, the ratio of successive coefficients $c_{n+1}/((n-\beta_1(s))c_n)$ is constant up to quadratic corrections.

The power series coefficients of the shift $\ln\hat \delta(\alpha)$ have been computed symbolically up to order $\alpha^{10}$, the leading ones are reported in \cref{d6_s_ldelta} in \cref{tables}. The numerical computation extends further, depending on $s$. The ratio of  successive coefficients, with the same normalization as for $\gamma(\alpha)$, is shown in \cref{nonlinear_6d_fig}(A). The plot suggests that $c_n$ grow at a similar rate as $d_n$.

\begin{figure}[htbp]
	\centering
	\begin{subfigure}{.48\textwidth}
		\includegraphics[width=\linewidth]{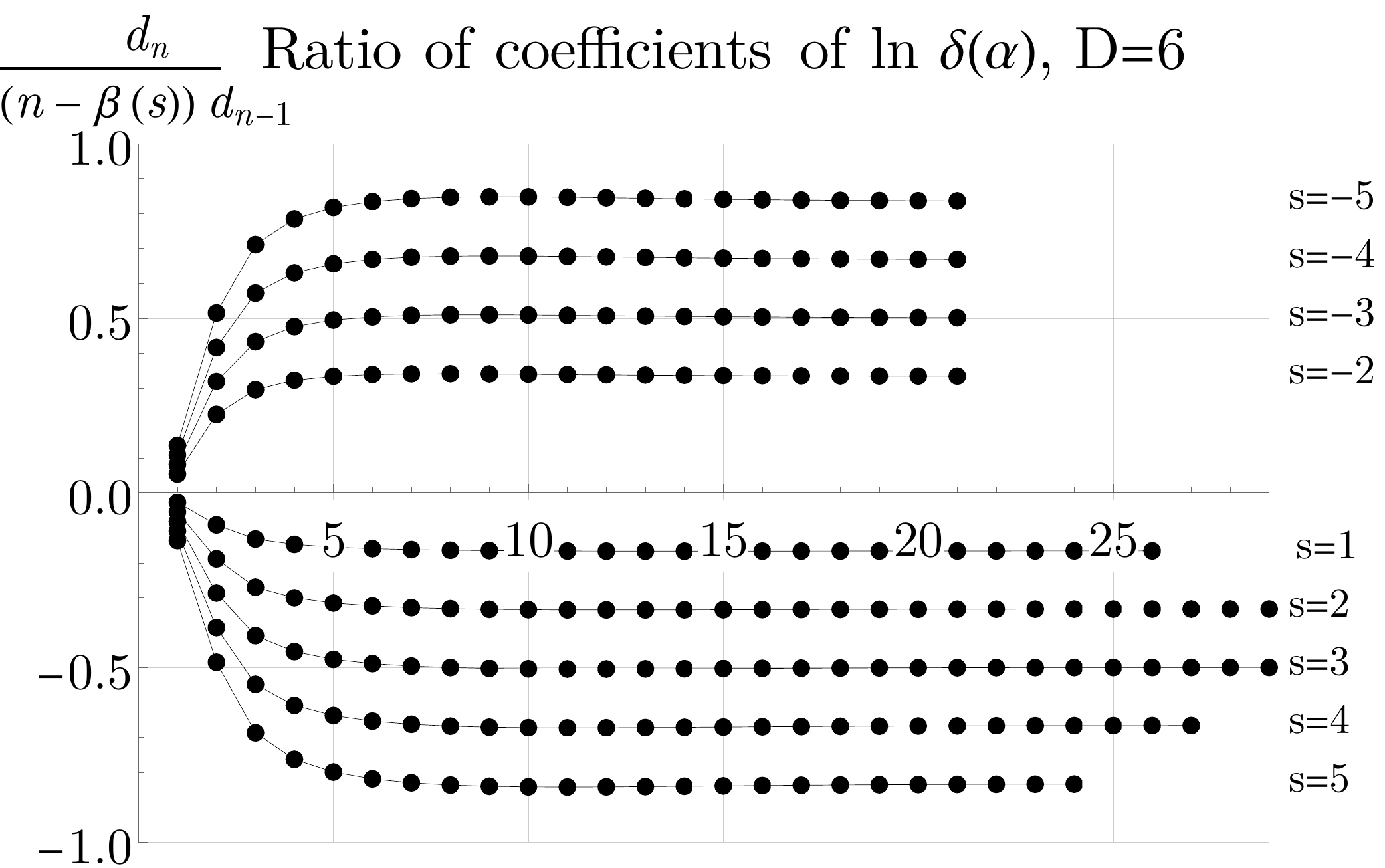}
		\caption{}
	\end{subfigure}
	\hfill 
	\begin{subfigure}{.48\textwidth}
		\includegraphics[width=\linewidth]{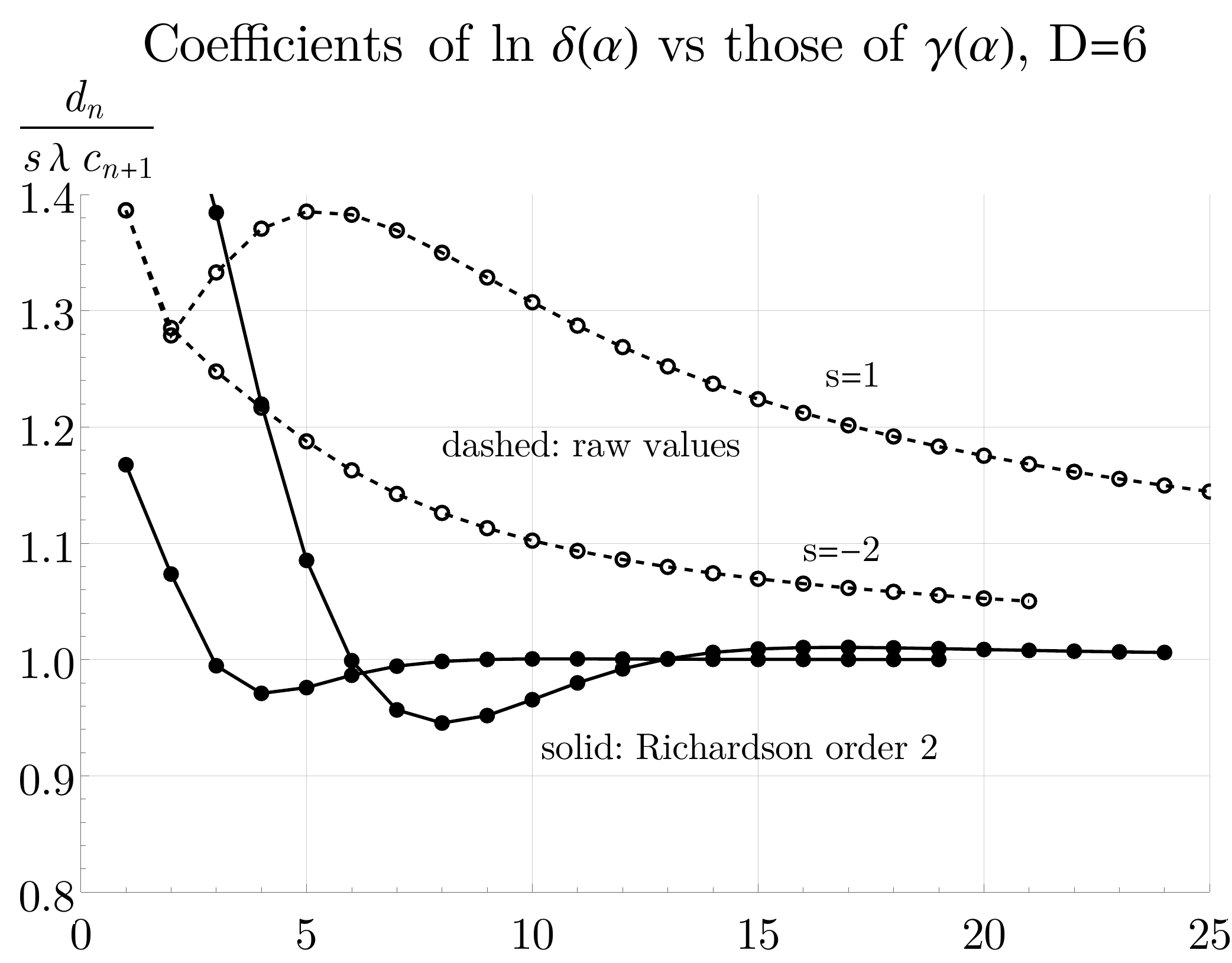}
		\caption{}
	\end{subfigure}
	\caption{(A) Ratio   of successive coefficients $d_n$ of $\ln\delta(\alpha)$, divided by the assumed leading asymptotic behaviour from \cref{nonlinear_6d_nonpert_parameters}, for $D=6$. This quantity quickly approaches the limit $-s/6$.
	(B) Ratio between the coefficients $d_n$ of $\ln\delta(\alpha)$ and the coefficients $c_{n+1}$ of $\gamma(\alpha)$. Compared to \cref{nonlinear_4d_fig_ratio} (A), the Richardson extrapolation converges slower, indicating a significant $1/n^2$-correction, see \cref{d6_s_ratio}. }
	
	\label{nonlinear_6d_fig}
\end{figure}

As explained in \cref{nonlinear_4d_results}, we extracted numerical estimates of the parameters in \cref{dn_asymptotic_ansatz}. They are given in \cref{d6_s_parameters} and are consistent with   $\tilde S(s) = s\cdot S(s), \tilde F(s) = -s/6$ and $\tilde \beta(s) = \beta_1(s)-1$.   Once more, the relative uncertainty of about 1\% of these values is too large to rigorously identify the rational values.

\begin{table}[htbp]
	\begin{tabular}{|c|c|c|c|c|}
		\hline
		$s$ & $n_\text{max}$& {\qquad  $10^6 \cdot \tilde S(s) / s$} & {\qquad $\tilde F(s)$}&{\qquad $\tilde \beta(s)$} \\
		\hline
		5 &24& $-50.1 \pm 4.5$ & $-0.833 \pm 0.018$ & $-7.00 \pm 0.37$ \\
		\hline
		4 &27& $-34.5 \pm 2.4$ & $-0.666 \pm 0.012$ & $-7.25 \pm 0.23$\\
		\hline
		3 &29& $-16.6 \pm 1.2$ & $-0.500 \pm 0.007$ & $-7.72 \pm 0.19$\\
		\hline
		2 &29& $-2.97 \pm 0.28$ & $-0.333 \pm 0.006$ & $-8.68 \pm 0.26$\\
		\hline
		1 &26& $ -0.0054 \pm 0.0015$ & $-0.167 \pm 0.006$ & $-11.64 \pm 0.75$\\
		\hline
		-2 &21& $ 87900 \pm 1600$ & $0.333 \pm 0.005$ & $-2.92 \pm 0.12$\\
		\hline
		-3 &21& $18000 \pm 560$ & $0.500 \pm 0.009$ & $-3.90 \pm 0.21$\\
		\hline
		-4 &21& $6690 \pm 290$ & $0.666 \pm 0.013$ & $-4.39 \pm 0.26$\\
		\hline
		-5 &21& $3410 \pm 180$ & $0.833\pm 0.018$ & $-4.69 \pm 0.29$ \\
		\hline 
	\end{tabular}
	\caption{Numerical findings of the growth parameters of $\ln\bar \delta(\alpha)$ in the $D=6$ model \cref{nonlinear_6d}, according to \cref{dn_asymptotic_ansatz}.  They are consistent with  \cref{nonlinear_6d_nonpert_parameters,d6_s_stokes} in the appendix. }
	\label{d6_s_parameters}
\end{table}

\begin{table}[htbp]
	\begin{tabular}{|c|c|c|c|}
		\hline
		$s$ &   $r(s)$ & $r_1(s)$ & $r_2(s)$ \\
		\hline
		5 &  $1.0010 \pm 0.0017$  & $2.573 \pm 0.072$ & $-6.91 \pm 0.84$ \\
		\hline
		4 &  $1.0007 \pm 0.0019$  & $2.685 \pm 0.072$ & $-7.3 \pm 1.4$ \\
		\hline
		3 &  $1.0007 \pm 0.0024$  & $2.863 \pm 0.078$ & $-8.4 \pm 1.8$ \\
		\hline
		2 &  $1.0010 \pm 0.0026$  & $3.21 \pm 0.11$ & $-11.2 \pm 1.8$ \\
		\hline
		1 &  $1.0032 \pm 0.0048$  & $4.22 \pm 0.23$ & $-22.3\pm 1.3$ \\
		\hline
		-2 & $1.0000 \pm 0.0002$  & $1.083 \pm 0.003$ & $-1.10 \pm 0.34$ \\
		\hline
		-3 & $1.0002 \pm 0.0004$  & $1.441 \pm 0.012$ & $-1.80 \pm 0.07$ \\
		\hline
		-4 & $1.0003 \pm 0.0007$  & $1.619 \pm 0.018$ & $-2.33\pm 0.16$ \\
		\hline
		-5 & $1.0004 \pm 0.0008$  & $1.726 \pm 0.022$ & $-2.71 \pm 0.22$ \\
		\hline 
	\end{tabular}
	\caption{Parameters of the ratio $d_n/(-6c_{n+1})$ for $D=6$  from \cref{4d_ratio}.  }
	\label{d6_s_ratio}
\end{table}

The ratio $d_n / (s \lambda c_{n+1}) = d_n / (-6 c_{n+1})$ is depicted in \cref{nonlinear_6d_fig} (B) for two values of $s$. The asymptotic parameters, according to \cref{4d_ratio}, are given in \cref{d6_s_ratio}. Unlike for $D=4$, this ratio does not converge particularly quickly. A fit suggests that $r_1(s) = (2.12 + 2.15 s)/s$, but the uncertainties are too large to identify the  numbers as rational. This is reflected by the large absolute values we obtain for the $1/n^2$-correction $r_2(s)$, see \cref{d6_s_ratio}. In $D=4$, this correction vanished and therefore allowed us to  extract $r_1(s)$ precisely. 

All in all, we can not clearly identify the subleading corrections of $d_n$ in the $D=6$ model, but the findings at least  suggest that the leading growth coincides with the one of $c_{n+1}$, that is
\begin{align}\label{nonlinear_6d_dn}
	d_n \sim S(s) s  \left( -\frac{s}{6} \right) ^{n} \Gamma \left( n+1+\frac{35+29s}{6s} \right) .
\end{align}

\FloatBarrier

\section{Non-linear toy model}\label{nonlinear_toy}

We solved the non-linear toy model DSE for $s\in \left \lbrace -5,\ldots,+4 \right \rbrace $ symbolically to order $\alpha^{16}$. The leading-log functions $H_1, H_2$ and $H_3$ agree with the general formula \cref{leadinglog_first} of \cite{kruger_log_2020} for the appropriate choice $c_1=1, c_2=0, c_3= \pi^2/2$ and for all values of $s$. Especially, for $s=-2$, we confirm $H_1$ from \cite[Cor. 3.6.4]{panzer_hopf_2012} and $H_2=H_4=H_6=0$.

The symbolic results for the anomalous dimension fulfil \cref{mellin_gamma1},
\begin{align}\label{toy_rge}
	-\frac{\sin(u)}{u}\Big|_{u \rightarrow -\gamma  (1+s \alpha \partial_\alpha)} \gamma(\alpha)  &=\alpha.
\end{align}
Unlike the ODEs \cref{nonlinear_4d_rge,nonlinear_6d_rge}, \cref{toy_rge} contains a pseudo differential operator.   

We computed a symbolic perturbative power series solution $\gamma(\alpha) = \sum c_n \alpha^n$ of \cref{toy_rge} to order 450 and extracted the asymptotic behaviour. The result has the form \cref{nonlinear_4d_asymptotic} for $n$ odd. We find $\beta(s) = -(2+s)/s$, numerical values of the constants $S(s), b^{(1)}(s)$ and $b^{(2)}(s)$ are given in \cref{toy_s_stokes} in \cref{tables}. We did not recognize these numbers apart from the Stokes constant $S(-2)=2/\pi$.

By \cref{lndelta_1}, the shift $\ln\delta(\alpha)=\sum d_n \alpha^n$ does not have a constant term in the toy model.
The first coefficients for the shift are reported in \cref{toy_s_ldelta} in \cref{tables}, while \cref{toy_s_parameters} contains the numerical estimates for their growth parameters.

 \begin{table}[htbp]
	\begin{tabular}{|c|c|c|c|c|}
		\hline
		$s$ & $n_\text{max}$& {\qquad  $ \tilde S(s) / s$} & {\qquad $\tilde F(s)$}&{\qquad $\tilde \beta(s)$}   \\
		\hline
		
		4 & 23 & $-0.389 \pm 0.010$ & $ 3.985 \pm 0.081$ & $-2.50 \pm 0.21$  \\
		\hline
		3 & 23 & $-0.485 \pm 0.015$ & $ 2.988 \pm 0.068$ & $-2.68 \pm 0.24$  \\
		\hline
		2 & 23 & $-0.612 \pm 0.023$ & $ 1.991 \pm 0.059$ & $-3.02 \pm 0.31$  \\
		\hline
		1 & 23 & $-0.572 \pm 0.037$ & $ 0.996 \pm 0.056$ & $-4.09 \pm 0.56$   \\
		\hline
		-2 & 23 & $ 0.6382 \pm 0.0067$ & $1.997 \pm 0.012$ & $-0.991 \pm 0.050$  \\
		\hline
		-3 & 23 & $ 0.5275 \pm 0.0076$ & $2.994 \pm 0.024$ & $-1.322 \pm 0.071$  \\
		\hline
		-4 & 23 & $ 0.4202 \pm 0.0069$ & $3.991 \pm 0.037$ & $-1.488 \pm 0.082$  \\
		\hline
		-5 & 23 & $ 0.3443 \pm 0.0060$ & $4.988 \pm 0.051$ & $-1.588 \pm 0.088$ \\
		\hline 
	\end{tabular}
	\caption{Numerical findings of the growth parameters of $\ln\bar \delta(\alpha)$ in the toy model, according to \cref{dn_asymptotic_ansatz}. $\tilde S(s)$ is  consistent with   \cref{toy_s_stokes}.  }
	\label{toy_s_parameters}
\end{table}

The toy model has the property that both $c_n$ and $d_n$ vanish for even $n$.  This is problematic for two reasons: Firstly, although we computed numerically the order $\alpha^{23}$, we only get 12 non-vanishing coefficients of $\ln\delta(\alpha)$. Secondly, we can not compute the ratio $d_n/c_{n+1}$ and therefore not use the trick which allowed us to extract the behaviour of $d_n$ in the $D=4$ physical model \cref{nonlinear_4d}.

To visualize the coefficients $d_n$ of $\ln\delta(\alpha)$, we consider  the following two ratios for odd $n$:
\begin{align}\label{toy_R}
  R^{(\delta)} &:= \sqrt{\frac{d_{n+2}}{(n-\beta(s)+1)(n-\beta(s)+2) d_n}}, \qquad R^{(\delta/\gamma)}:= \frac{d_n}{s\cdot (n-\beta(s)+1) c_n}.
\end{align} 
\Cref{nonlinear_toy_fig} (A) shows that $R^{(\delta)}$ approaches the limit $\abs s$, which suggests that  $d_n$ scale asymptotically  $ \sim s^n \Gamma(n-\beta(s)+1)$, with (approximately) the same $\beta(s)$ as the coefficients $c_n$ of $\gamma(\alpha)$. The quantity $R^{(\delta/\gamma)}$ allows us to fix the Stokes constant, it is shown in  \cref{nonlinear_toy_fig} (B). The limit of $R^{(\delta/\gamma)}$  is $1.00\pm 0.02$, suggesting that the Stokes constant agrees with the one of $\gamma(\alpha)$. In \cref{toy_s_parameters}, the direct estimates of the asymptotic growth are reported. All in all, it seems that the coefficients of $\ln\delta(\alpha)$ grow factorially according to 
\begin{align}\label{nonlinear_toy_dn}
	d_n \sim S(s) s^{n+1} \Gamma(n+(2+2s)/s).
\end{align}
 We did not try to determine subleading corrections.

\begin{figure}[htbp]
	\centering
	\begin{subfigure}{.48\textwidth}
		\includegraphics[width=\linewidth]{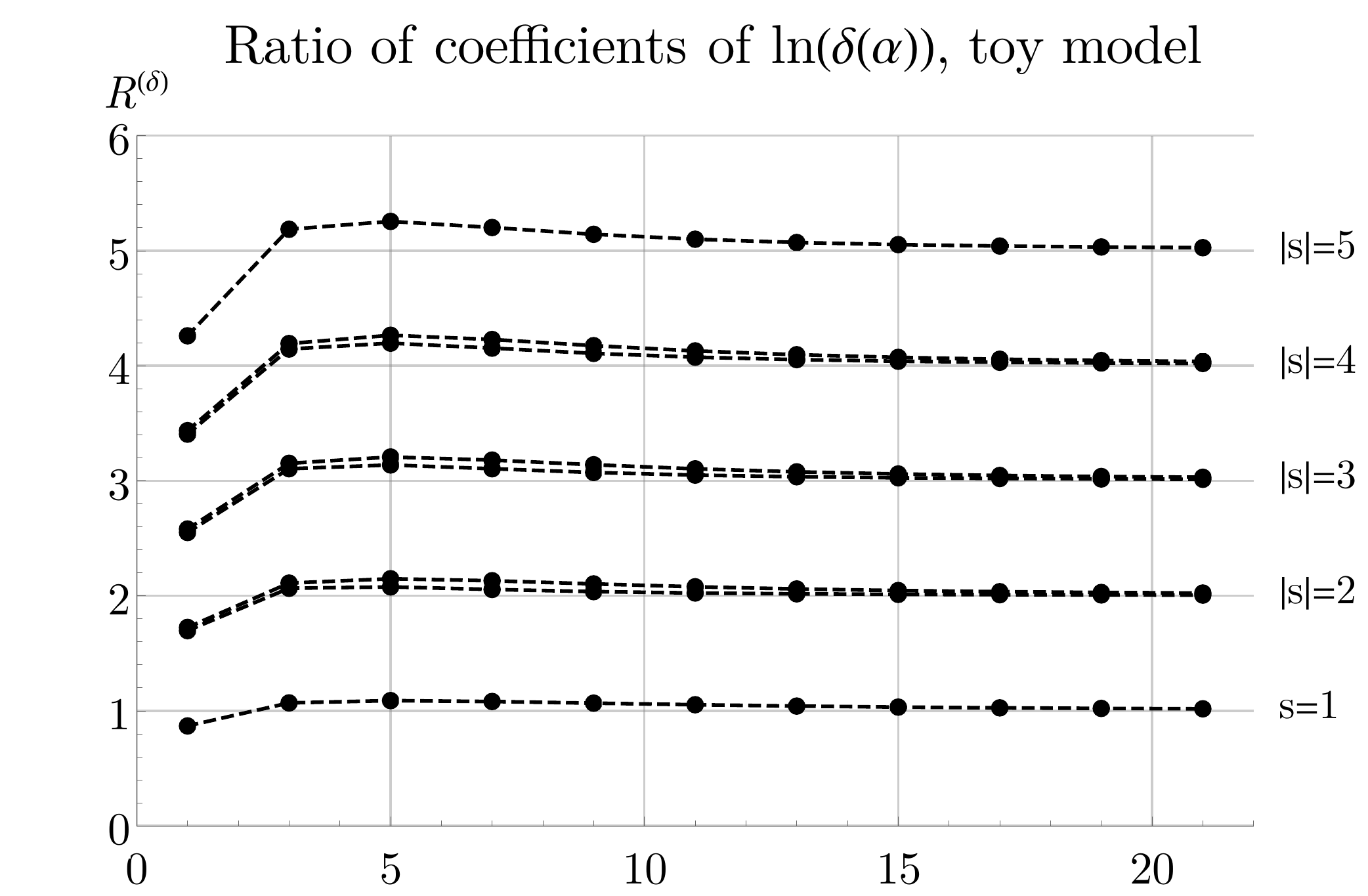}
		\caption{}
	\end{subfigure}
	\hfill 
	\begin{subfigure}{.48\textwidth}
		\includegraphics[width=\linewidth]{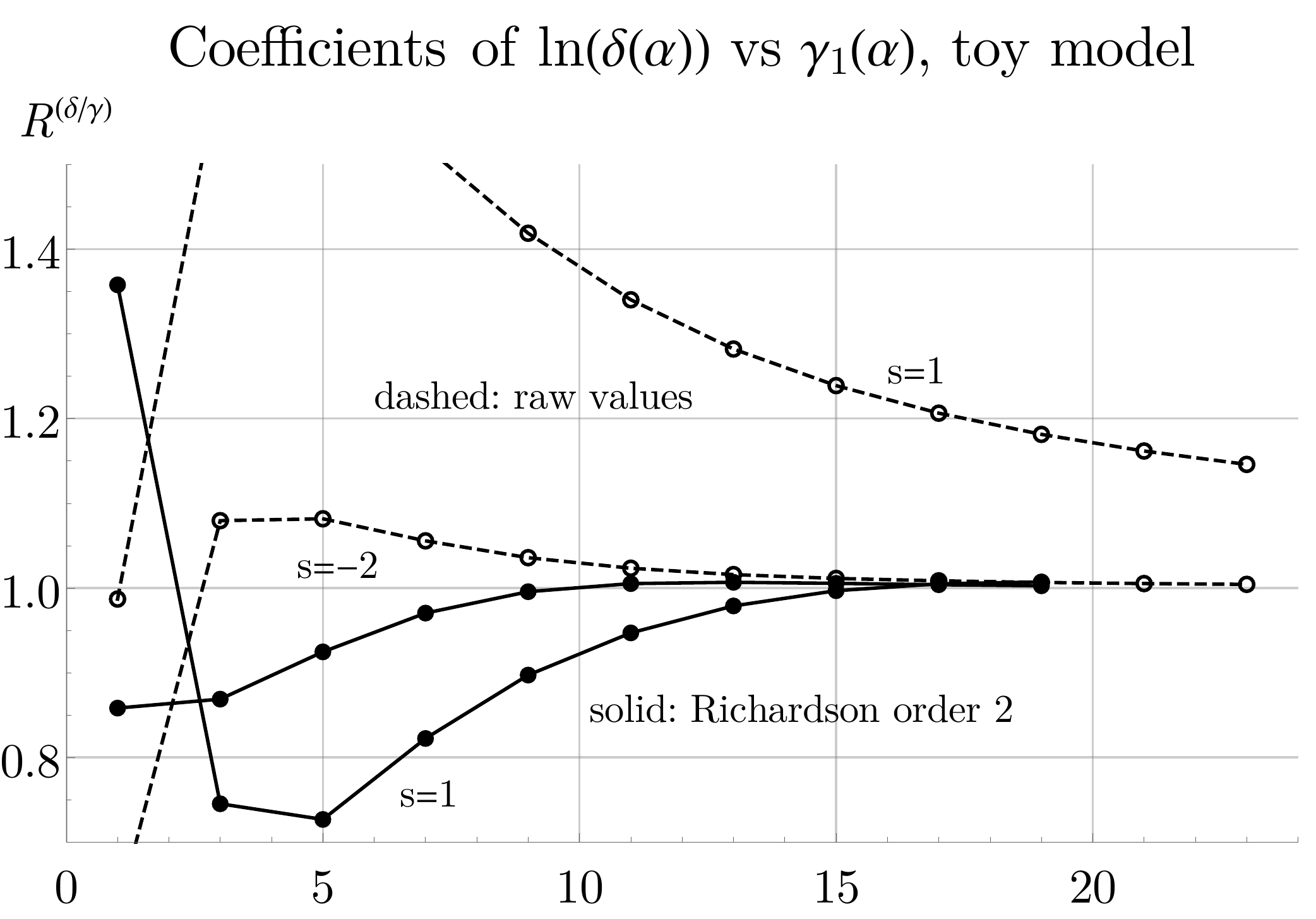}
		\caption{}
	\end{subfigure}
	\caption{(A) Ratio \cref{toy_R} of successive coefficients  of $\ln\delta(\alpha)$.  The ratio visibly approaches $\abs s$, Richardson extrapolation (not shown) confirms this. 
	(B) Ratio $R^{(\delta/\gamma)}$ between coefficients of $\ln\delta(\alpha)$ and $\gamma(\alpha)$ and its order-2-Richardson extrapolation. The limit seems to be unity, but knowing only 12 terms the uncertainty is large. Compare the first 12 terms of \cref{nonlinear_4d_fig_ratio} (A).   }
	
	\label{nonlinear_toy_fig}
\end{figure}

\subsection{Exact solutions}\label{toy_exact}

We end this paper with a curious empirical observation. First, for $s=-\frac 12$, the perturbative anomalous dimension in MOM for $\epsilon \rightarrow 0$ turns out to be $\gamma(\alpha) = -\alpha$, which was checked up to $\mathcal O(\alpha^{500})$. Moreover, in MS for $s=-2$, we find $\bar \gamma(\alpha)=-\alpha$ at least up to order $\alpha^{18}$. By construction, the latter is true even for $\epsilon \neq 0$. If we assume that there are indeed no higher order terms in $\alpha$, then  a particularly simple Callan-Symanzik equation \ref{rge} follows. 

Firstly, consider the case $s=-1$ in MOM which clearly has $\gamma(\alpha)=-\alpha$ since the DSE is not even recursive. The beta function is $\beta(\alpha) = s\gamma(\alpha) = +\alpha$ and the  Callan-Symanzik equation becomes
\begin{align*}
\partial_{\ln x} G(\alpha,x) &= -\alpha G(\alpha,x) + \alpha^2 \partial_\alpha G(\alpha,x).
\end{align*}
The general solution of this partial differential equation is 
\begin{align*}
G(\alpha,x) &= \alpha F_{-1} \left( \ln(x)-\frac 1 \alpha \right) ,
\end{align*}
where $F_{-1}$ is an arbitrary function. The requirement $\gamma(\alpha) = -\alpha$ together with the boundary condition $G(\alpha,1) = 1$ fixes $F_{-1}(u) = -u$. 

The Callan-Symanzik-equation for  $s=-\frac 12$ has the general solution
\begin{align*}
G(\alpha,x) &= \alpha^2 F_{-\frac 12} \left( \ln (x)-\frac 2 \alpha \right) .
\end{align*}
The condition $\gamma(\alpha) =-\alpha$ translates to $\partial_u F_{-\frac 12 }(u) = \frac 12 u$ and we find
\begin{align*}
G(\alpha,x) &= \frac 1 4 \alpha^2 \ln(x)^2 -\alpha \ln(x) +1.
\end{align*}
Both MOM-results are consistent with \cref{rge_MOM}.

The case $s=-2$ in MS leads to a similar general solution,
\begin{align*}
\bar G(\alpha,x) &= \sqrt{\alpha} \bar F_{-2} \left( \ln(x) - \frac{1}{2\alpha}. \right) , 
\end{align*}
This time we can not fix the function $\bar F_{-2}$ because the anomalous dimension $\bar \gamma(\alpha)$ is not simply the derivative of $\bar G(\alpha,x)$, see \cref{shifted_gammabeta}. The shift $\delta(\alpha)$ between MS and MOM is given by the inverse function,
\begin{align*}
\ln \delta(\alpha) &= -(\bar F_{-2})^{-1} \left( \frac{1}{\sqrt \alpha} \right) -\frac{1}{2\alpha}.
\end{align*}
These two non-linear DSEs illustrate that it can be worth trying to solve a DSE both in MOM and in MS, but also that going from one scheme to another requires a truly new, independent calculation and is not trivial even if one happens to know an exact solution in one of the schemes.

\FloatBarrier

\section{Conclusion}\label{conclusion}

We have discussed how a Green function in Minimal Subtraction (MS) is related to its corresponding Green function in kinematic renormalization (MOM). To this end, we have examined and used various relations between the renormalization group functions and $Z$-factors to find the shift of the renormalization point (\cref{MS_as_shifted}).

We have computed the series coefficients of  $\delta(\alpha)$ symbolically and numerically for propagator-type Dyson-Schwinger equations with three different kernels. In some cases, we identified the coefficients' algebraic formulas. Whenever there is an overlap with earlier literature, our results agree with the ones reported.  The key outcomes of the present work are:
\begin{enumerate}
	\item We have shown for single, propagator-type Dyson-Schwinger equations, that  their solutions in MS, MS-bar and MOM schemes agree to all orders in perturbation theory if one chooses a suitable kinematic renormalization point $\delta(\alpha)$, which is a power series in $\alpha$.  (\cref{thm}).
	\item In the linear examples, the factors $\delta(\alpha)$ between the renormalization points have been deduced in closed form as $  \delta(\alpha) =   \gamma_0^{1/\gamma}$, see  \cref{linear_4d_gammabar0,linear_6d_gammabar0,toy_gamma0}. The result is finite in perturbation theory and proportional to $\sqrt{\partial_\alpha \gamma(\alpha)}$. It also encodes information about the MOM-solution  for $\epsilon\neq 0$, see \cref{linear_delta2}. 
	\item For non-linear DSEs, series coefficients for $\gamma(\alpha)$ and $\ln\delta(\alpha)$ have been computed for several different exponents $s$ in the invariant charge $Q=G^s$. The results are highly regular in $s$.  The first symbolic  coefficients of $\ln\delta(\alpha)$ are collected in \cref{tables}. 
	\item  The coefficients $d_n$ of $\ln \delta(\alpha)$ seem to grow factorially. In all cases, we find $d_n \sim  S\cdot s \cdot s^n\lambda^{-n}  \cdot  \Gamma(n-\beta+1)$, where $S,\lambda$ and $\beta$ are the growth parameters of the corresponding anomalous dimension $\gamma(\alpha)$,  \cref{nonlinear_4d_dn,nonlinear_6d_dn,nonlinear_toy_dn}.	This suggests that  $\delta(\alpha)$ receives non-perturbative contributions of the same type as does the anomalous dimension $\gamma (\alpha)$  \cite{borinsky_nonperturbative_2020,borinsky_semiclassical_2021}. The resemblance is particularly striking in the $D=4$ model, \cref{4d_lndelta_gamma1}.
		\item The chain approximation \cref{chains} is an example of a Green function which does not originate from a DSE and can not be transformed between MS and MOM. This calls into question the physical validity of this approximation since different renormalization schemes will produce truly different renormalized Green functions.
	\item Our numerical data suggests some new tentative exact results:  The Stokes constant for $s=-3$ in the $D=4$ model seems to be $S(-3) = 3/(\pi e^2)$, see \cref{nonlinear_4d_results}. And for the non-linear toy model at $s=-2$ in MS and for $s=-\frac 12$ in MOM, the anomalous dimension appears to be $ \gamma(\alpha)=-\alpha$. This allows one to solve the Callan-Symanzik equation up to one unknown function, which in the MOM-case can be determined uniquely (\cref{toy_exact}).
\end{enumerate}
All examples indicate that there is a significant shift  factor $\delta(\alpha)$ between the mass scale $\mu$ of MS-renormalization and the corresponding kinematic renormalization point. By \cref{lndelta_1}, the shift does not vanish in the limit of vanishing coupling.  Consequently one should be careful not to confuse the mass scale $\mu$ of MS-renormalization with a kinematic renormalization point, even in the most well-behaved cases and even for \enquote{small coupling}.

For the linear DSEs, we have found an explicit map $\gamma(\alpha) \mapsto \ln\delta(\alpha)$ in \cref{linear_4d_gammabar0,linear_6d_gammabar0,toy_gamma0}. For the non-linear DSEs, this connection is not quite so simple. Intuitively, the identical asymptotic growth hints at the possibility to find an explicit map as well. Indeed, for the $D=4$ case, assuming that our empirical findings hold to all orders, \cref{4d_lndelta_gamma1} reproduces the factorial growth of $d_n$ and therefore the non-perturbative behaviour. The analytic remainder function $f(\alpha)$ in this case remains to be identified.

\appendix

\section{Mellin transforms}\label{mellin}
The Mellin transform is by definition the value of the primitive integral where one of the propagators is raised to a power $1+\rho$, evaluated at unity external momentum. Factors of $4\pi$ from the Fourier transform are implicitly absorbed into the mass scale in the main text, therefore they are left out from the Mellin transform. For the 4-dimensional resp. 6-dimensional propagator,
\begin{align*}
	M(\rho) = \int \frac{\d^4 k  \;(k^2)^{-\rho}}{(k+q)^2 k^2} \Big|_{q^2=1} = \frac{1}{\rho(1-\rho)},\qquad 	M(\rho)  &= \int  \frac{\d^6 k \;  (k^2)^{-\rho}}{(k+q)^2 k^2} \Big|_{q^2=1} =-\frac{ 1}{ \rho (1-\rho)(2-\rho) (3-\rho) }.
\end{align*}
The Mellin transform in the toy model is
\begin{align*}
	M(\rho) &= \int \limits_0^\infty \frac{\d y \; y^{-\rho}}{x+y} \Big|_{x=1} = \frac{\pi}{\sin (\pi \rho)}.
\end{align*}

\section{Series expansion of the primitive graphs}\label{series}
We are interested in the series expansion in $\epsilon$ of the integral
\begin{align*}
	I_D^{(k)} (q) &:= \int \frac{\d^D  p }{(2\pi)^D} \frac 1 {\left(  p + q \right) ^2 \left( p^2 \right) ^{1+k\epsilon}} \\
	&= (4\pi)^{-\frac D 2} (q^2)^{\frac D 2-2-k\epsilon}\frac{\Gamma\left( -\frac D 2+2+k\epsilon  \right) \Gamma \left( \frac D 2 -1 \right) \Gamma \left( \frac D 2 -1-k\epsilon  \right)  }{\Gamma \left( 1+k\epsilon  \right) \Gamma \left( D-2-k\epsilon  \right)  } \\
	&=: (4\pi)^{-\frac D 2} (q^2)^{\frac D 2-2-k\epsilon} e^{-\gamma_E \epsilon}\sum_n f_n^{(k)} \epsilon^n .
\end{align*}
The factors of $q^2$ must not be expanded into logarithms, in order to be integrated in the next iteration. Furthermore, we do not expand the $(4\pi)$   and factor out $e^{-\gamma_E \epsilon}$ because both can conveniently be absorbed into the momentum. It remains to expand the gamma functions using
\begin{align*}
	\Gamma(x+1) = x \Gamma(x) ,\qquad \Gamma\left( 1+\epsilon \right) &= \exp \left( -\gamma_E \epsilon + \sum_{m=2}^\infty \frac{(-\epsilon)^m}{m}\zeta(m) \right) .
\end{align*}
In $D=4-2\epsilon$ dimensions one obtains
\begin{align}\label{Tkn}
	\Gamma:=& \frac{\Gamma((k+1)\epsilon) \Gamma(1-(k+1)\epsilon) \Gamma(1-\epsilon)}{\Gamma(1+k\epsilon) \Gamma(2-(k+2)\epsilon)}
	= \frac{1}{(k+1) \left( 1-(k+2)\epsilon \right)\epsilon  }\exp \left( -\gamma_E \epsilon + \sum_{m=2}^\infty T_m^{(k)}  \epsilon^m \right) , \nonumber \\
	&\text{where} \qquad 	T^{(k)}_m := (m-1)!\left( (-1)^m (k+1)^m + (k+1)^m +1- (-k)^m-(k+2)^m \right)\zeta(m).
\end{align}
Expanding the prefactor in a geometric series and leaving out  $e^{-\gamma_E \epsilon}$,
\begin{align}\label{linear_4d_fkn}
	f_n^{(k)}&= \sum_{t=-1}^n \frac{(k+2)^{t+1}}{k+1} \frac{1}{(n-t)!}\sum_{m=0}^{n-t} B_{n-t,m} \left( 0, T^{(k)}_2, T^{(k)}_3, \ldots, T^{(k)}_{n-t+1-m} \right) .
\end{align}
Here $B_{n,k}$ are incomplete Bell polynomials \cite{bell_exponential_1934} \cite[p 134]{comtet_advanced_1974}. 
For $D=6-2\epsilon$ dimensions, observe
\begin{align*}
	\frac{\Gamma\left( -1+(k+1)\epsilon  \right) \Gamma \left( 2-\epsilon \right) \Gamma \left( 2-(k+1)\epsilon  \right)  }{\Gamma \left( 1+k\epsilon  \right) \Gamma \left( 4-(k+2)\epsilon  \right)  }  &= \frac{\epsilon-1 }{ (3-(k+2)\epsilon)(2-(k+2)\epsilon) }\cdot \Gamma.
\end{align*}
The gamma functions on the right hand side are the same as in $D=4-2\epsilon$, consequently their series expansion is again given by the polynomials $T^{(k)}_m$ from \cref{Tkn}.  
\begin{align}\label{linear_6d_fkn}
	f^{(k)}_n &=   \sum_{t=-1}^n  \left( -  (k+1)   +\frac{k}{2^{t+1}} -\frac{k-1}{  3^{t+2}} \right)    \frac{(k+2)^t}{2(k+1)}  \frac{1}{(n-t)!}\sum_{m=0}^{n-t} B_{n-t,m} \left( 0, T^{(k)}_2, \ldots, T^{(k)}_{n-t+1-m} \right).
\end{align}

In the toy model, the relevant integral and its series expansion are
\begin{align*}
	&\int \limits_0^\infty \frac{\d y \; y^{-(k+1)\epsilon } }{1+y} = \frac{\pi}{\sin (\pi (k+1)\epsilon)}= \Gamma((k+1)\epsilon)\Gamma(1-(k+1)\epsilon)  =: \sum_{n=-1}^\infty f^{(k)}_n \epsilon^n.
\end{align*}
The Bernoulli numbers $B_n$ vanish when $n>1$ is odd, therefore we can write
\begin{align}\label{linear_toymodel_fkn}
	f^{(k)}_n &:= \frac{1}{(k+1)}\sum_{m=0}^{n+1} \frac{1}{(n+1)!}B_{n+1,m} \left( 0, T^{(k)}_2, \ldots, T^{(k)}_{n+2-m}  \right) ,\qquad  T^{(k)}_n  :=		 \Big(2\pi (k+1)\Big)^n \frac{ \abs{ B_{n}}}{n}.	 
\end{align}

\section{Kinematic counter Term of the linear toy model}\label{app_toymodel}

These are the first coefficients of \cref{toymodel_Z} for the toy model. Define $A := \alpha^2 \pi^2$. 
 \begin{align*}
	z_1  &= \frac{ \alpha \pi^2 \left( 4+  A  \right)  }{24 \left( 1- A \right) ^{\frac 3 2}}, \qquad  
	z_2   = \frac{\alpha^2 \pi^4 \left(3+2A \right)  }{16 \left( 1 - A \right) ^3},\qquad 
	z_3  = \frac{\alpha \pi^4 \left( 112 +2240 A + 2919 A^2 + 254 A^3 \right)  }{ 5760 \left( 1-A\right) ^{\frac 9 2}},  \\
	z_4  &= \frac{\alpha^2 \pi^6 \left( 36 + 515A + 900 A^2 +240 A^3+ 4 A^4 \right)  }{384 \left( 1-A \right)^6 },\\
	z_5  &= \frac{\alpha \pi^6 \left( 1984+522152 A + 6074220 A^2 + 12882535 A^3 + 6095260 A^4 + 511956 A^5 + 1768A^6 \right)  }{967680 \left( 1-A\right) ^{\frac {15} 2}},\\
	z_6  &= \frac{\alpha^2 \pi^8 \left( 471 + 42058 A + 428661 A^2 + 1041030 A^3 + 715270 A^4 + 129414 A^5 + 4092 A^6 + 4 A^7 \right)  }{11520 \left( 1-A\right) ^{9}}.
\end{align*}
All explicitly determined functions $z_n(\alpha)$ for $n>0$ behave qualitatively similar: They diverge like $(1-A)^{-\frac 3 2 n}$ as $\alpha \pi \rightarrow 1$ and are positive for $0\leq \alpha <\frac 1 \pi$. In this interval, they give rise to a finite $Z$-factor. The other coefficients in \cref{toymodel_Z} are
\begin{align*}
	z'_1  &=  \frac{1}{\epsilon} +\frac{\pi^2}{6}\epsilon +\frac{7\pi^4}{360} \epsilon^3  +\ldots
	= -\frac 1 \epsilon \sum_{n=0}^\infty \frac{(-1)^n(4^n-2)B_{2n}}{(2n)!}\pi^{2n} \epsilon^{2n}= \pi \left( \cot \left( \frac \pi 2 \epsilon \right) -\cot (\epsilon \pi)   \right), \\
	z'_2 &= \frac {\pi^2}4 \frac{1}{\cos^2\left( \frac \pi 2 \epsilon  \right) \cos(\epsilon \pi)},\qquad 	z'_3 = \frac{\pi^3}{12}\frac{ 1-2 \cos(2 \pi \epsilon) +2 \cos(\pi \epsilon)  }{  \left( 2 \cos(\pi \epsilon) + 3 \cos(3 \pi \epsilon) \right) \sin\left( \frac   \pi2 \epsilon  \right) \cos^3 \left( \frac   \pi 2 \epsilon \right)  },\\
	z'_4 &= -\frac{\pi^4}{4} \frac{\cos(\pi \epsilon) + 2 \sin^2 (\pi \epsilon)}{\cos^4(\pi \epsilon) \left( \cos(\pi \epsilon) -4 \cos^2 (\pi \epsilon) \cos(2\pi \epsilon)  + \cos(3\pi \epsilon) \right)  }.
\end{align*}
The functions $z'_n(\epsilon)$ are positive for small positive $\epsilon$. They change sign at their poles but probably, there are other continuations of the series expansion around $\epsilon=0$ beyond the poles, which stay positive. For example
\begin{align*}
	z'_2(\epsilon) = \frac{\pi^2}{2} \left( \frac{1}{\sin^2( \pi \epsilon)} -\frac{\left|\cot(\pi \epsilon)-\cot \left( \frac \pi 2 \epsilon \right) \right|^3 }{|\cot (\pi \epsilon)|} \right)  
\end{align*}
is always positive and reduces to the above form of $z'_2$ for $|\epsilon|<0.5$. If this holds for all $z'_n$ then $Z(\alpha,\epsilon)\in [0,1]$, which allows us to interpret $Z$ as a probability.

\section{Asymptotic growth of the anomalous dimension}
For the 4-dimensional physical model, the ansatz \cref{nonlinear_4d_nonpert} delivers the growth parameters
\begin{align}\label{nonlinear_4d_nonpert_parameters}
	\lambda &= \frac 1 s ,\quad \beta(s) = -\frac{3+2s}{s}, \quad b^{(1)}(s) = -\frac{1+4s+3 s^2}{s}, \quad b^{(2)}(s) = \frac{1+6s+8s^2-2s^3-5s^4}{2s^2}\\
	&\qquad b^{(3)}(s) = \frac{-1-6s-s^2+24s^3-25s^4-126 s^5 -81 s^6}{6s^3}. \nonumber
\end{align}
They match \cite[(14)]{borinsky_nonperturbative_2020} for $s=-2$. 
For the 6-dimensional physical model, there are three solutions:
\begin{align}\label{nonlinear_6d_nonpert_parameters}
	\vec \lambda (s) &= \left(  -\frac{6}{s} , -\frac{12}{s}, -\frac{18}{s}\right),\quad 	\vec \beta (s) = \left(  -\frac{35+29s}{6s} ,  -\frac{5+2s}{3s},-\frac{15+13 s}{2s} \right) \\
	\vec b^{(1)}(s) &= \left(  \frac{-275-267s+8s^2}{216s} , \frac{265 + 624 s + 359s^2}{108s},\frac{85+241 s + 156 s^2}{72 s}\right)  \nonumber \\
	\vec b^{(2)}(s) &=  \left( {\scriptstyle \frac{75625 + 83790 s - 101849 s^2 - 177828 s^3 - 67814 s^4 }{93312 s^2}, } \right.  \nonumber \\
	&\qquad \left. {\scriptstyle  \frac{ 70225 + 339690 s + 602764 s^2 + 465258 s^3 + 131959 s^4}{23328 s^2}, \frac{ 7225 + 37950 s + 69779 s^2 + 51628 s^3 + 12574 s^4}{10368 s^2}  } \right) \nonumber 
\end{align}
\begin{align*}
	\vec b^{(3)}(s) &= \left(   {\scriptstyle \frac{ -20796875 - 8551125 s + 107422197 s^2 + 206297091 s^3 + 
		177713418 s^4 + 90251478 s^5 + 23658704 s^6}{60466176 s^3}}, \right. \\
	&\qquad {\scriptstyle  \frac{18609625 + 138592350 s + 424432473 s^2 + 687305592 s^3 + 
		624311121 s^4 + 303609366 s^5 + 62154089 s^6}{7558272 s^3} },\\
	& \qquad \left.  {\scriptstyle \frac{614125 + 4453575 s + 12499453 s^2 + 16989843 s^3 + 11830354 s^4 + 
			4259034 s^5 + 758520 s^6}{2239488 s^3}} \right)
\end{align*}
Including  order $1/n^3$, the large-order growth of $c_n$ is determined entirely by the first component of these vectors.
In order to match \cite[eqs. (41)-(43)]{borinsky_semiclassical_2021}, $b^{(1)}$ has to be multiplied with 3, $b^{(2)}$ with 9 and $b^{(3)}$ with 27. Parameters for the toy model are given in \cref{toy_s_stokes} in \cref{tables}.

\section{Tables}\label{tables}

\begin{table}[htbp]
	\begin{tabular}{|c|l|}
		\hline
		$s$ & \qquad  $\gamma(\alpha)$ \\
		\hline
		5 & $ -\alpha - 6 \alpha^2 - 102 \alpha^3 - 2640 \alpha^4 - 87804 \alpha^5 - 3483072 \alpha^6 - 		158329512 \alpha^7 - 8050087584 \alpha^8$ \\
		\hline
		4 & $-\alpha - 5 \alpha^2 - 70 \alpha^3 - 1485 \alpha^4 - 40370 \alpha^5 - 1306370 \alpha^6 - 		48365100 \alpha^7 - 2000065725 \alpha^8 $ \\
		\hline
		3 & $-\alpha - 4 \alpha^2 - 44 \alpha^3 - 728 \alpha^4 - 15368 \alpha^5 - 384960 \alpha^6 - 		11004672 \alpha^7 - 350628096 \alpha^8$ \\
		\hline
		2 & $-\alpha - 3 \alpha^2 - 24 \alpha^3 - 285 \alpha^4 - 4284 \alpha^5 - 75978 \alpha^6 - 1530720 \alpha^7 - 		34237485 \alpha^8$ \\
		\hline
		1 & $-\alpha - 2 \alpha^2 - 10 \alpha^3 - 72 \alpha^4 - 644 \alpha^5 - 6704 \alpha^6 - 78408 \alpha^7 - 		1008480 a^8$ \\
		\hline
		0 & $-\alpha - \alpha^2 - 2 \alpha^3 - 5 \alpha^4 - 14 \alpha^5 - 42 \alpha^6 - 132 \alpha^7 - 429 \alpha^8$ \\
		\hline
		-1 & $-\alpha$ \\
		\hline
		-2 & $-\alpha + \alpha^2 - 4 \alpha^3 + 27 \alpha^4 - 248 \alpha^5 + 2830 \alpha^6 - 38232 \alpha^7 + 		593859 \alpha^8$ \\
		\hline
		-3 & $-\alpha + 2 \alpha^2 - 14 \alpha^3 + 160 \alpha^4 - 2444 \alpha^5 + 45792 \alpha^6 - 1005480 \alpha^7 + 		25169760 \alpha^8$ \\
		\hline
		-4 & $-\alpha + 3 \alpha^2 -30 \alpha^3 + 483 \alpha^4 -10314 \alpha^5 + 268686 \alpha^6 -8167068 \alpha^7+281975715 \alpha^8$  \\
		\hline
		-5 & $-\alpha + 4 \alpha^2 - 52 \alpha^3 + 1080 \alpha^4 - 29624 \alpha^5 + 988288 \alpha^6 - 		38377152 \alpha^7 + 1689250176 \alpha^8$\\
		\hline 
	\end{tabular}
	\caption{Non-linear DSE in $D=4$ dimensions, see \cref{nonlinear_4d_results}. Series expansion of the anomalous dimension in MOM as a function of  the renormalized coupling $\alpha$ up to order $\alpha^8$ for various powers $s$ of the invariant charge $Q = G^s$. Only insertions into a single internal edge were performed in all cases.}
	\label{d4_s_gamma1}
\end{table}

\begin{table}[htbp]
	\begin{tabular}{|c|S[table-format=2.53]|}
		\hline
		$s$ & {\qquad  $S(s)$} \\
		\hline
		5 &  -0.025296711447842155554062589810922604262477942805771 \\
		\hline
		4 & -0.027093755285804302538145834438779321901953254099492 \\
		\hline
		3 & -0.027514268695235967509951466619196206136028416088749 \\
		\hline
		2 & -0.022754314527304604570864961094569471756231077114904 \\
		\hline
		1 & -0.0054283179932662026367480341381320752861015892636883 \\
		\hline
		-2 & 0.20755374871029735167013412472066868268445351496963   \\
		\hline
		-3 & 0.12923567581109177871522936685966399491429288708430 \\
		\hline
		-4 & 0.087977369959821254076048394021324447743442962588612  \\
		\hline
		-5 & 0.065314016354658749144010387750377100215558556707446 \\
		\hline 
	\end{tabular}
	\caption{First 50 digits of the Stokes constant $S(s)$   for the non-linear DSE in $D=4$, see \cref{nonlinear_4d_asymptotic}. One finds $S(-2)= (\sqrt \pi e)^{-1}$ and $S(-3)=3 (\sqrt \pi e)^{-2}$}
	\label{d4_s_stokes}
\end{table}

\begin{table}[htbp]
	\begin{tabular}{|c|l|l|}
		\hline
		$s$ & \qquad  $\ln \bar \delta(\alpha)$ & $f_{n+1}/f_n$\\
		\hline
		5 & $ -2 -9 \alpha + \left( -139 + 14 \zeta(3) \right) \alpha^2 + \left( -3464 -\frac{7\pi^4 }{12} + 233 \zeta(3) \right) \alpha^3 $ &$30.22 \pm 0.09$ \\
		\hline
		4 & $-2 -\frac{15}{2}\alpha + \left( -\frac{575}{6} + 10 \zeta(3) \right) \alpha^2 + \left( -\frac{23525}{12} -\frac{\pi^4}{3} + \frac{410}{3}\zeta(3) \right) \alpha^3 $ &$ 25.09 \pm 0.06$ \\
		\hline
		3 & $ -2 -6 \alpha + \left( -\frac{182}{3}+ \frac{20}{3}\zeta(3) \right) \alpha^2 + \left( -\frac{2911}{3} -\frac{\pi^4}{6} + \frac{214}{3}\zeta(3) \right) \alpha^3$ & $19.96 \pm 0.04$ \\
		\hline
		2 & $ -2 -\frac 9 2 \alpha + \left( \frac{-67}{2} + 4 \zeta(3) \right) \alpha^2 + \left( -\frac{773}{2}-\frac{\pi^4}{15}+ 31 \zeta(3) \right) \alpha^3$ & $14.80 \pm 0.02$\\
		\hline
		1 & $ -2 -3 \alpha +\left( -\frac{43}{3}+2\zeta(3) \right) \alpha^2  + \left( -\frac{305}{3} -\frac{\pi^4}{60} + \frac{29}{3}\zeta(3) \right) \alpha^3 $ & $9.60 \pm 0.01$ \\
		\hline
		0 & $ -2 -\frac 3 2 \alpha + \left( -\frac{19}{6}+ \frac 2 3 \zeta(3) \right) \alpha^2 + \left( -\frac{103}{12} + \frac 4 3 \zeta(3) \right) \alpha^3$ & \\
		\hline
		-1 & $ -2 $ &  \\
		\hline
		-2 & $ -2 + \frac 3 2 \alpha -\frac{29}{6 } \alpha^2 + \left( \frac{94}{3}-\frac 1 3 \zeta(3) \right) \alpha^3$ & $5.8\pm 1.8$ \\
		\hline
		-3 & $ -2 + 3 \alpha + \left( -\frac{53}{3} + \frac 2 3 \zeta(3) \right) \alpha^2 + \left( \frac{578}{3} + \frac{\pi^4}{60}-\frac{17}{3}\zeta(3) \right) \alpha^3$ &$ 10.50 \pm 0.11$ \\
		\hline
		-4 & $ -2+ \frac 92\alpha  + \left(-\frac{77}{2}  + 2 \zeta(3)\right)\alpha^2 + \left( \frac{2365}{4}+ \frac{\pi^4}{15}-22 \zeta(3) \right) \alpha^3$ &$15.69\pm 0.05$ \\
		\hline
		-5 & $-2 + 6 \alpha + \left( -\frac{202}{3} + 4 \zeta(3) \right) \alpha^2 + \left( \frac{4003}{3} + \frac{\pi^4}{6} -\frac{166}{3}\zeta(3) \right) \alpha^3 $ & $20.85 \pm 0.07$\\
		\hline 
	\end{tabular}
	\caption{Non-linear DSE in $D=4$ dimensions. $\ln\bar \delta(\alpha)$ is the logarithm of the shift in the renormalization point between MOM- and MS-scheme \cref{Gren_gammaks}. Shown are the first terms of its perturbative power series. $f_{n+1}/f_n$ is the growth rate of the function $f(\alpha)$ in \cref{4d_lndelta_gamma1}.}
	\label{d4_s_ldelta}
\end{table}

\begin{table}[htbp]
	\begin{tabular}{|c|S[table-format=5.53]|}
		\hline
		$s$ & {\qquad  $10^6\cdot S(s)$} \\
		\hline
		5 &  -48.879979612936267148575174247043686402701421680529 \\
		\hline
		4 & -33.683126435179258367949154667346857343063662040223 \\
		\hline
		3 & -16.197057487106552084835982615789341267879644562145 \\
		\hline
		2 & -2.8749310663584041698420077656773118015156356312116\\
		\hline
		1 & -0.0050376438522521046131658646410401520933414352165372 \\
		\hline
		-2 & 87595.552909179124483795447421262990627388017406822 \\
		\hline
		-3 & 17853.256793175269493347991077950813245133374820922 \\
		\hline
		-4 & 6637.5931100379316509518941784586037225957017664650  \\
		\hline
		-5 & 3384.1867616825132279651486289425088074650135043176 \\
		\hline 
	\end{tabular}
	\caption{First 50 digits of the Stokes constant $ S(s)$   for   $D=6$, see \cref{nonlinear_4d_asymptotic}.   }
	\label{d6_s_stokes}
\end{table}

\begin{table}[htbp]
	\begin{tabular}{|c|l|}
		\hline
		$s$ & \qquad  $\ln \delta(\alpha)$ \\
		\hline
		5 & $ -\frac 8 3 + \frac{61}{24}\alpha + \left( -\frac{80213}{7776} + \frac{7}{18}\zeta(3) \right) \alpha^2 + \left( \frac{8813575}{139968} + \frac{7\pi^4 }{2592} -\frac{2563}{1296}\zeta(3) \right) \alpha^3 $ \\
		\hline
		4 & $ -\frac 8 3 + \frac{305}{144}\alpha  + \left( -\frac{331345}{46656} + \frac{5}{18}\zeta(3) \right) \alpha^2 + \left( \frac{119812205}{3359232} + \frac{\pi^4}{648} -\frac{2255}{1944}\zeta(3) \right) \alpha^3  $ \\
		\hline
		3 & $ -\frac 8 3 + \frac{61}{36}\alpha + \left( -\frac{52325}{11664} + \frac{5}{27}\zeta(3) \right) \alpha^2 + \left( \frac{14842891}{839808} + \frac{\pi^4}{1296} -\frac{1177}{1944}\zeta(3) \right) \alpha^3$ \\
		\hline
		2 & $ -\frac 8 3 + \frac{61}{48}\alpha + \left( -\frac{38381}{15552} + \frac 1 9 \zeta(3) \right) \alpha^2 + \left( \frac{3947825}{559872} + \frac{\pi^4}{3240} -\frac{341}{1296}\zeta(3) \right) \alpha^3        $ \\
		\hline
		1 & $ -\frac 8 3 + \frac{61}{72}\alpha + \left( -\frac{24437}{23328} + \frac{1}{18}\zeta(3) \right) \alpha^2 + \left( \frac{1560359}{839808} + \frac{\pi^4}{12960} -\frac{319}{3888}\zeta(3) \right) \alpha^3  $ \\
		\hline
		0 & $ -\frac 8 3 + \frac{61}{144} \alpha + \left( -\frac{10493}{46656} + \frac{1}{54}\zeta(3) \right) \alpha^2 + \left( \frac{518095}{3359232}-\frac{11}{972}\zeta(3) \right) \alpha^3 $ \\
		\hline
		-1 & $ -\frac 8 3  $ \\
		\hline
		-2 & $   -\frac 8 3 -\frac{61}{144} \alpha-\frac{17395}{46656}\alpha^2 + \left(- \frac{114361}{209952} + \frac{11}{3888}\zeta(3) \right) \alpha^3  $ \\
		\hline
		-3 & $ -\frac 8 3 - \frac{61}{72}\alpha + \left( \frac{31339}{23328} + \frac{1}{54}\zeta(3) \right) \alpha^2 + \left( -\frac{359005}{104976} -\frac{\pi^4}{12960} + \frac{187}{3888}\zeta(3) \right) \alpha^3     $ \\
		\hline
		-4 & $ -\frac 8 3 -\frac{61}{48} \alpha + \left( -\frac{45283}{15552} + \frac{1}{18}\zeta(3) \right) \alpha^2 + \left( -\frac{11830593}{1119744}-\frac{\pi^4}{3240} + \frac{121}{648}\zeta(3) \right) \alpha^3     $  \\
		\hline
		-5 & $ -\frac 8 3 -\frac{61}{36}\alpha + \left( -\frac{59227}{11664} + \frac 1 9 \zeta(3) \right) \alpha^2 + \left( -\frac{20089615}{839808} -\frac{\pi^4}{1296} + \frac{913}{1944}\zeta(3) \right) \alpha^3  $\\
		\hline 
	\end{tabular}
	\caption{First perturbative coefficients of $\ln\delta(\alpha)$ for  $D=6$ dimensions.}
	\label{d6_s_ldelta}
\end{table}

\begin{table}[htbp]
	\begin{tabular}{|c|S[table-format=2.20]|S[table-format=3.19]|S[table-format=3.14]|}
		\hline
		$s$ & {\qquad  $S(s)$} & {\qquad $b^{(1)}(s)$}&{\qquad $b^{(2)}(s)$} \\
		\hline
		5 & -0.32358439814031030546 & -33.713129682396588961 & 565.374787298670 \\
		\hline
		4 & -0.39133508371923490586 & -28.505508252042547410 & 405.630022359080\\
		\hline
		3 & -0.48873615802624779599 & -23.352717957250113407 & 273.573399332400\\
		\hline
		2 & -0.62073652944344889658 & -18.337005501361698274 & 169.862094180663\\
		\hline
		1 & -0.59543401151910843904 & -13.869604401089358619 & 98.0525675224480\\
		\hline
		-2 & 0.63661977236758134308 & 4.4674011002723396547 & 7.97883629535726\\
		\hline
		-3 & 0.52618629546780378450 & 9.4831135561607547882 & 40.2545894932164\\
		\hline
		-4 & 0.41925649525660905756 & 14.635903850953188791 & 98.9744451682625\\
		\hline
		-5 & 0.34358721547093244258 & 19.843525281307230343 & 184.621956597118\\
		\hline 
	\end{tabular}
	\caption{First digits of the Stokes constant $S(s)$ and subleading corrections of the asymptotic growth  \cref{nonlinear_4d_asymptotic} of the anomalous dimension  in the toy model \cref{nonlinear_toy}. }
	\label{toy_s_stokes}
\end{table}

\begin{table}[htbp]
	\begin{tabular}{|c|l|}
		\hline
		$s$ & \qquad  $\alpha \ln\bar \delta(\alpha(A)$ \\
		\hline
		5 & $ -6 A - \frac{2009}3 A^2 - \frac{11563106}{45} A^3 - \frac{173306477104}{945} A^4 - \frac{1228737945883358}{6075} A^5 - \frac{46235332362117842849}{147015} A^6  $ \\
		\hline
		4 & $-5 A - \frac{1130}3 A^2 - \frac{4316822}{45} A^3 - \frac{59632972484}{1323} A^4 - \frac{461687074578658}{14175} A^5 - \frac{34025588969113725668}{1029105} A^6 $ \\
		\hline
		3 & $ -4 A - \frac{554}3 A^2 - \frac{1263424}{45} A^3 - \frac{10282878575}{1323} A^4 - \frac{46540947260036}{14175} A^5 - \frac{398737839692532122}{205821} A^6 $ \\
		\hline
		2 & $ -3 A - \frac{217}3 A^2  - \frac{1233338}{225} A^3 - \frac{4881119933}{6615} A^4 -\frac{	3528108924854}{23625} A^5 - \frac{1074400592111547046}{25727625 } A^6$ \\
		\hline
		1 & $  -2 A - \frac{55}3 A^2 - \frac{106898}{225} A^3 - \frac{135875429}{6615} A^4- \frac{272890120256}{212625} A^5 - \frac{2770658834393158}{25727625} A^6  $ \\
		\hline
		0 & $-A - \frac 43 A^2 - \frac{146}{45} A^3 - \frac{8864}{945} A^4 - \frac{417682}{14175} A^5 - \frac{9095176}{93555} A^6  $ \\
		\hline
		-1 & $  0 $ \\
		\hline
		-2 & $ A + 7 A^2 + 242 A^3 + 17771 A^4 + 2189294 A^5 + 404590470 A^6 $ \\
		\hline
		-3 & $ 2 A + 41 A^2 + \frac{92518}{25} A^3 + \frac{503885698}{735} A^4 + \frac{		1639676026462}{7875} A^5 + \frac{266517331818761291}{2858625} A^6 $ \\
		\hline
		-4 & $3 A + \frac{370}3 A^2 + \frac{4782122}{225} A^3 + \frac{48904622516}{6615} A^4 + \frac{887103429351554}{212625} A^5  + \frac{88600913717695595572}{25727625} A^6  $  \\
		\hline
		-5 & $ 4 A + \frac{826}3 A^2 + \frac{3478864}{45} A^3 + \frac{287007344207}{6615} A^4 + \frac{185545372999796}{4725} A^5+\frac{53252838327756373006}{1029105} A^6 $\\
		\hline 
	\end{tabular}
	\caption{First coefficients of $\ln\delta(\alpha)$ in the toy model \cref{nonlinear_toy}, up to order $\alpha^{11}$. Here, $A:= (\alpha \pi)^2 /4$.}
	\label{toy_s_ldelta}
\end{table}

\FloatBarrier

\printbibliography

\end{document}